\documentclass[11pt,a4paper,aps]{article}
\usepackage[applemac]{inputenc}		
\usepackage[T1]{fontenc}
\usepackage{stmaryrd}			
\usepackage{lmodern}					
\usepackage{textcomp}				
\usepackage{verbatim}				
\usepackage[english]{babel}
\usepackage[toc,page]{appendix} 		
\usepackage{amsmath,					
			amsthm,
			amsfonts,
			amssymb}

\usepackage{graphicx}				
\usepackage{tabularx}				
\usepackage{pst-all}					

\usepackage[hmargin=2.2cm,vmargin=2.5cm]{geometry}		
\numberwithin{equation}{section}

\usepackage{hyperref}			

\newcommand{\bigslant}[2]{{\left.\raisebox{.2em}{$#1$}\middle/\raisebox{-.2em}{$#2$}\right.}}
\newcommand{\F}[1]{\mathfrak{#1}}
\newcommand{\C}[1]{\mathcal{#1}}
\newcommand{\R}[1]{\mathrm{#1}}
\newcommand{\B}[1]{\mathbb{#1}}
\newcommand{\Q}{Q_{\C{M}}}
\newcommand{\QQ}{Q_{\bar{\C{M}}}}

\newcommand{\dd}{\mathrm{d}}

\def\g{\mathfrak{g}}

\newtheorem{theoreme}{Theorem}[section]
\newtheorem{proposition}{Proposition}[section]

\newtheorem{lemme}{Lemma}[section]

\newtheorem{definition}{Definition}[section]
\theoremstyle{remark}
\newtheorem*{remarque}{Remark}
\newtheorem*{remarquelol}{Remarks}

\begin{document}

\begin{center}
{\LARGE{Hidden Q-structure and Lie 3-algebra \\[1ex]
for non-abelian superconformal models in six dimensions}}\\[8ex]

{\bf Sylvain Lavau$^a$, Henning Samtleben$^b$, Thomas Strobl$^a$}\\[3ex]

$^a$\,{\em Institut Camille Jordan,
Universit\'e Claude Bernard Lyon 1 \\
43 boulevard du 11 novembre 1918, 69622 Villeurbanne cedex,
France}

\vskip 2mm

$^b$\,{\em Universit\'e de Lyon, Laboratoire de Physique, UMR 5672, CNRS et ENS de Lyon,\\
46 all\'ee d'Italie, F-69364 Lyon CEDEX 07, France} \\

\end{center}

\vskip 14mm

\begin{abstract}
We disclose the mathematical structure underlying the gauge field sector of the recently constructed non-abelian superconformal models in six space-time dimensions. This is a coupled system of 1-form, 2-form, and 3-form gauge fields. We show that the algebraic consistency constraints governing this system permit to define a Lie 3-algebra, generalizing the structural Lie algebra of a standard Yang-Mills theory to the setting of a higher bundle. Reformulating the Lie 3-algebra in terms of a nilpotent degree 1 BRST-type operator $Q$, this higher bundle can be compactly described by means of a Q-bundle; its fiber is the shifted tangent of the Q-manifold corresponding to the Lie 3-algebra and its base the odd tangent bundle of space-time equipped with the de Rham differential. The generalized Bianchi identities can then be retrieved concisely from $Q^2=0$, which encode all the essence of the structural identities. Gauge transformations are identified as vertical inner automorphisms of such a bundle, their algebra being determined from a $Q$-derived bracket.
\end{abstract}

\vskip 4mm

\thispagestyle{empty}

\newpage
\tableofcontents

\thispagestyle{empty}
\setcounter{page}{0}
\newpage


\section{Introduction}\label{intro}
Higher gauge theories, i.e.\ gauge theories with not just 1-form gauge fields but also gauge fields of higher form degrees, have been increasingly in the focus of research over the last years. On the one hand, they are motivated by string theory, 
where $p$-form gauge potentials are inevitably induced in the low energy effective action. While a single family of 
$p$-form gauge fields with fixed $p>1$ necessarily stays abelian (as one easily sees when trying to construct an invariant action functional),
the general tower of $p$-form gauge fields can implicitly carry a non-abelian character and even transcend standard Lie algebras in this case. Concrete examples of such systems have appeared in compactifications with fluxes~\cite{Louis:2002ny,DallAgata:2005mj,DAuria:2005er} and been identified as the generic gauge structure within the so-called gauged supergravity theories (see e.g.~\cite{deWit:2004nw,Samtleben:2005bp}) where already the bosonic sector features a nontrivial, sufficiently intricate system of structural equations that generalize the usual Jacobi identity of Lie algebras. The general structure of these bosonic $p$-form gauge field hierarchies has further been analyzed in~\cite{deWit:2005hv,deWit:2008ta,deWit:2008gc,Bergshoeff:2009ph,Riccioni:2009xr,Coomans:2010xd,Palmkvist:2011vz,Palmkvist:2013vya,Greitz:2013pua}. The mathematical understanding of such a model and its underlying system of equations is the subject of the present investigation. 

On the mathematical side, higher gauge theories attracted recent attention with the general interest in ``higher structures in mathematics and physics'' (cf., e.g., the recent series of conferences in this field), but also before in the context of gerbes and bundle gerbes. One of the possible approaches consists of a so-called ``categorification'', cf., e.g, \cite{Breen:2001ie,Baez:2002jn,Sati:2008eg}. Ordinary principal bundles with their structure groups and connections on them are generalized (in a non-unique way) to some higher versions, usually indexed by an integer $n$ such that for $n=1$ one obtains the respective original mathematical object. So from this point of view one expects a structural Lie $n$-group or its corresponding Lie $n$-algebra (or $n$-term $L_\infty$-algebra) that would correspond to the set of structural equations present in concrete models like the ones above. It is one of the purposes of the present article to confirm this expectation. We will consider the general system of $p\le3$ forms whose non-abelian gauge structure has been worked out in \cite{Samtleben:2011fj,Samtleben:2012fb} in the context of superconformal field theories in six dimensions. Such systems have been under investigation for their relation to the yet elusive dynamics of multiple M5 branes, see e.g.~\cite{Lambert:2010wm,Singh:2011id,Ho:2011ni,Chu:2012um,Akyol:2012cq,Bonetti:2012st,Saemann:2012uq}.
A similar analysis of the general tensor hierarchy developed in \cite{deWit:2008ta} will be presented in a forthcoming paper \cite{Kotov:xxx}.

An alternative to the categorification perspective is the consideration of so-called Q-bundles (cf., e.g., \cite{Kotov:2007nr}), i.e.\ bundles in the category of Q-manifolds. The terminology arose in the context of the BRST-BV formalism, but turns out to be of interest in the above-mentioned context even before this formulation. An example is the following one: Consider a Lie algebra action on a manifold $M$, $\rho \colon \g \to \Gamma(TM)$. Then the standard BRST charge has the form
\begin{equation}
Q = q^a \rho_a(x) \frac{\partial}{\partial x^i} - \frac{1}{2} q^a q^b C_{ab}^c \frac{\partial}{\partial q^c} \, , \label{eq:Q}
\end{equation}
where $q^a$ are degree +1 odd (i.e.~anti-commuting) coordinates on $\g$ and $C_{ab}^c$ are the corresponding structure constants. It is one of the defining properties of a BRST charge that it is an odd vector field and squares to zero. Here it is a vector field on the graded manifold $M \times \g[1]$, where the additional 1 in the brackets indicates that the coordinates on $\g$ are considered as carrying degree +1. On the other hand, a vector field of the type (\ref{eq:Q}) can also be used to \emph{define} what is called a Lie algebroid, and this even in the most compact form (cf.~\cite{Vaintrob}). In this case, the $C^c_{ab}$ are even permitted to depend on the coordinates $x$ on $M$; the odd coordinates $q^a$ are fiber coordinates of a vector bundle $E \to M$ underlying the Lie algebroid. The case of coordinates such that $C^c_{ab}$ are constants corresponds precisely to the situation of a Lie algebra acting on the base $M$ of the algebroid, which then is called an action Lie algebroid (cf., e.g., \cite{daSilva-Weinstein}). 

In the special case of no $x$ coordinates in \eqref{eq:Q} above, a Lie algebra can be \emph{defined} by $Q=\frac{1}{2} q^a q^b C_{ab}^c \frac{\partial}{\partial q^c}$ squaring to zero, which is easily seen to be equivalent to the structure constants satisfying the Jacobi identity. Correspondingly, higher analogues, so-called $L_\infty$-algebras (or, if there are also $x$-coordinates $L_\infty$-algebroids) can be defined by a similar $Q$ vector field squaring to zero, just that in this more general case there are coordinates of different integer-valued degrees on the graded manifold. 

This formulation usually has advantages over the categorification picture. First of all, in the latter language one usually is confronted with several identities or diagrams to be verified one by one, the more of them, the higher is the $n$. In the Q-language, all the structural equations are concisely encoded in the fact that $Q$ squares to zero. Correspondingly, whenever the structural equations are needed, one will just use $Q^2=0$ in that language. This turns out to be a considerable advantage on the level of the gauge transformations, for example, as we will show in detail. Also a definition of Q-bundles for higher $n$ is easy to give: there are just graded coordinates up to degree $n$ in the fibers. Even the generalization of a connection, an $n$-connection, which corresponds to the whole tower of the gauge fields of different form degrees, is easy to formulate: it is a section of the graded bundle (forgetting about the Q-structures, cf, e.g., \cite{Kotov:2007nr}). 

Secondly, the Q-formulation is usually more general, since in contrast to Lie algebras, a Lie $n$-algebroid may not permit an integration to a Lie $n$-groupoid. This is already the case for $n=1$ (cf.~\cite{Crainic-Fernandes} for the necessary and sufficient conditions for the integration of Lie algebroids). Still, as already the example of the two-dimensional Poisson sigma model \cite{Schaller:1994uj,Schaller:1995xk,Ikeda:2001fq,Ikeda:2002wh} shows, the model is meaningful and interesting \cite{Kontsevich:1997vb,Cattaneo:2000iw,Cattaneo:2001bp} also in the case of a non-integrable target Poisson Lie algebroid. The Poisson sigma model is a Chern-Simons type of theory for a Poisson Lie algebroid (cf.~\cite{Strobl:2004im,Kotov:2007nr} making this more precise) and can serve as a toy model for Yang-Mills theories where the structural Lie algebra is replaced by a Lie algebroid \cite{Strobl:2004im,Mayer:2009wf}, and these theories in turn, can serve as a model for higher gauge theories in general since one has the simplest tower of gauge fields, zero forms $X^i$ and 1-forms $A^a$ interacting non-trivially with one another.\footnote{For completeness we mention that also the categorification perspective will have its advantages: for example, when it comes to a generalization of a holonomy or parallel transport, the integrated version, when assumed to exist, is expected to play a prominent role since already in the ordinary $n=1$ Lie algebra case holonomies related to groups.}

We motivated the appearance of $Q$-structures as possible generalizations of Lie algebras and as an elegant alternative formulation of $L_\infty$ algebras. One has \emph{bundles} in this category since the tangent bundle $T\Sigma$ of a space-time or base manifold $\Sigma$, with fiber-linear coordinates considered of carrying degree +1 again and being odd (this graded manifold is conventionally denoted by $T[1]\Sigma$, becomes a Q-manifold in a canonical way by means of the de Rham differential as a $Q$ and, simultaneously, permits to host all different form degrees of our gauge fields as functions on this graded manifold. It remains to see, why, in a general theory of such a type, one should expect  the existence of a Q-structure on the fibers from a physical and not just a mathematical perspective. Here an observation made in \cite{Gruetzmann-Strobl} comes into the game: In all the known examples of such gauge theories there is some notion of (higher) curvatures or, in a physical language, field strengths such that they satisfy \emph{some} notion of Bianchi identities. Formalising this idea in some appropriate sense, as reviewed in the subsequent section, one can prove that this \emph{implies} the existence of a Q-structure in the fibers of the bundle (cf.~Theorem \ref{thm:QBianchi}) below). 

Since the validity of such generalized Bianchi identities, also holds true in the model of~\cite{Samtleben:2011fj}, the existence of such a Q-structure (and its associated $L_\infty$-algebra) characterizing the theory is guaranteed in principle. On the practical level, it is nevertheless a challenge to show that the concrete structural identities of the model \cite{Samtleben:2011fj} indeed imply a particular vector field $Q$ to square to zero, cf.~Theorem \ref{propositionhenning} below. In the construction of \cite{Samtleben:2011fj}, just as in most physical models, the system of $p$-forms is truncated to the set that actually appears in a Lagrangian formulation of the dynamics (typically including forms up to degree $p\le[D/2]$ for space-time dimension $D$). In particular, the Bianchi identities are shown to hold true (on behalf of a set of specified structural identities) only up to this order in the form degree, truncating the highest order by a specific projection. This will have its correspondence in Proposition \ref{propositiontruncation}, stating that every, in general infinite, $L_\infty$-algebra permits a canonical truncation to any arbitrary level (in particular to the level  coinciding with the highest relevant form degree in the given space-time dimension).

The structure of the paper is as follows: In Section \ref{section2} we review briefly the considerations of \cite{Gruetzmann-Strobl}; we restrict to trivial bundles in particular, since one of the main questions in the present article is the identification of the underlying structural $L_\infty$ algebra (but a generalization to nontrivial bundles is rather immediate following the steps as outlined in \cite{Kotov:2007nr}). In Section \ref{section3} we likewise review  the ingredients needed from \cite{Samtleben:2011fj}: the precise content of the gauge fields, a possible set of field strengths, the gauge transformations and, in particular, the eight structural equations, Eqs.~(\ref{lol}) below,  that are imposed for defining the theory.  In the subsequent section we then review the relation of $Q$-manifolds with $L_\infty$-algebras and prove the above-mentioned result on possible truncations. The main results of the present paper are presented in Section \ref{sec:principal}: we define the Q-structure for the six-dimensional tensor hierarchy, work out the Lie 3-algebra and derive the gauge transformations as vertical inner automorphisms of the corresponding bundle.

\paragraph{Acknowledgements:} We would like to thank Melchior Gr\"utzmann, Alexei Kotov, Jakob Palmkvist and 
Robert Wimmer for interesting discussions.


\section{Basics of the  Q-formalism}

This section is mainly a review following the ideas of \cite{Gruetzmann-Strobl}; the lifting to the tangent, that will be important in the end for our application, is inspired by its sequels \cite{Kotov:2007nr} and \cite{Salnikov:2013pwa}.

\label{section2}

\subsection{The Bianchi identities}

For some $n$ dimensional space-time $\Sigma$, the generalization of gauge theories to higher degrees usually involves a tower of gauge fields, which eventually stops at some degree $m$. It means that we have a set of 1-forms $A^{a}$, 2-forms $B^{I}$, 3-forms $C_{t}$, 4-forms $D_{\lambda}$ etc... up to $m$-forms. We have put the index of the 3-forms and of the 4-forms at the bottom of the letters just as a convention in concordance with the second part of this article. From a physical perspective, ``field strengths'' (or ``generalized curvatures'') are an essential ingredient of gauge theories, and there will be as many field strengths as there are gauge fields. The most general natural ansatz in terms of gauge fields for the first three field strengths is as follows:
\begin{align}
	\C{F}^{a}	&= \dd  A^{a}+\frac{1}{2}C_{bc}^{a}A^{b}\wedge A^{c}-t^{a}_{I}B^{I} \;,\label{eq:Fa}\\
	\C{F}^{I}	&= \dd  B^{I}+\Gamma^{I}_{aJ}A^{a}\wedge B^{J}-\frac{1}{6}H^{I}_{abc}A^{a}\wedge A^{b}\wedge A^{c}-t^{It}C_{t}\;,\nonumber\\
	\begin{split}
	\C{F}^{(4)}_{t}	&= \dd  C_{t}-A^{s}_{at}A^{a}\wedge C_{s}-B_{IJt}B^{I}\wedge B^{J}-D_{abIt}A^{a}\wedge A^{b}\wedge B^{I}\\
				&\hspace{4cm}-E_{abcdt}A^{a}\wedge A^{b}\wedge A^{c}\wedge A^{d}-t^{\lambda}_{t}D_{\lambda} \;.\nonumber
	\end{split}
\end{align}
In fact, one could add  on the right-hand side terms containing derivatives of lower-form degree gauge fields or, equivalently, their previously defined field strengths; this may be seen to not change the subsequent analysis of this section.

{}From now on, we will write the gauge fields collectively as $A^{\alpha}$ and the field strengths correspondingly as $\C{F}^{\alpha}$, for $\alpha=a, I, t,..$; so for example, $A^{I}\equiv B^{I}$ and $A_{t}\equiv C_{t}$. To each $\alpha$, we associate a degree, denoted by $|\alpha|$, which is the order of the gauge field as a differential form, that is: $A^{\alpha}$ is an $|\alpha|$-form, and we write $N_{|\alpha|}$ for the number of $|\alpha|$-forms. We put $N=N_{1}+N_{2}+...+N_{m}$. We now want to cast the Bianchi identities into a very general form. Recall that classically (standard YM-theory), we have $D {F}^{a}=0$, that is $\dd  {F}^{a}=-C^{a}_{bc}A^{b}\wedge{F}^{c}$, i.e the de Rham differential of the field strength yields a term containing a field strength. To generalize this idea, we say that a field strength $\C{F}^{\alpha}$ satisfies the generalized Bianchi identity if $\dd  \C{F}^{\alpha}=\lambda^{\alpha}_{\beta}\C{F}^{\beta}$, where the summation is implied and $\lambda^{\alpha}_{\beta}$ can depend on the fields. To summarize:

\begin{definition}
\label{def:Bianchi}
Let ${\cal A}$ be the  subalgebra of $\Omega^{\bullet}(\Sigma)$ generated by the differential forms $A^{\alpha}$ and $\dd  A^{\alpha}$ and $\C{I}\subset {\cal A}$ its both-sided ideal generated by the field strengths $\C{F}^{\alpha}$ (for some fixed choice of constants in their definition). Then the field strengths $\C{F}^{\alpha}$ satisfy the \emph{generalized Bianchi identities}, iff for any choice of the $A^\alpha$s this ideal is a differential ideal, i.e.~iff
		\begin{equation}
		\dd   \C{I}\subset\C{I} \, . 
	\end{equation}
\end{definition}
Usually in examples such a weak notion of Bianchi identities holds true and we will assume this to be the case when talking about ``higher gauge theories''. We will proceed now to show that Q-structures are a good formalism to encode them. 

Recall first the definition of a Q-manifold: it is a $\B{Z}$-graded manifold $\C{M}$ equipped with a degree $+1$ vector field $Q$ such that $[Q,Q]\equiv 2Q^{2}=0$. So in a local chart description, $\C{M}$ has coordinates which carry each an integer degree --- in fact, we will consider only so-called NQ-manifolds, where there are no negative coordinate degrees --- and which are glued together 
on overlaps by degree-preserving diffeomorphisms.  $[X,Y]=X\circ Y-(-1)^{|X||Y|}Y\circ X$ is the graded commutator defined on the space of vector fields $\F{X}(\C{M})$. One of the classical examples of a Q-manifold is the shifted tangent bundle equipped with the de Rham differential. The shifted tangent bundle $T[1]\Sigma$ over a smooth manifold $\Sigma$ is the graded vector bundle modelled on $T\Sigma$ in the sense that the local basis vectors $\frac{\partial}{\partial x^{i}}$ are now carrying a negative degree of $-1$. It implies that the fiber-linear coordinate functions on $T[1]\Sigma$, as function from $T[1]\Sigma$ to $\B{R}$, are functions of degree $+1$. Given a local coordinate system $\{x^i\}_{i = 1}^n$ on $\Sigma$, the (locally defined) differential 1-forms $dx^{i}$ thus are such degree 1 functions. The algebra of smooth functions over $T[1]\Sigma$ is defined as the set of formal series in these coordinate functions, glued together correctly on overlaps by degree preserving diffeomeorphisms, so that one arrives at the identification $\C{C}^{\infty}(T[1]\Sigma)\cong\Omega^{\bullet}(\Sigma)$. The de Rham differential then becomes just a vector field, in local coordinates  $\dd  = \dd  x^{i}\frac{\partial}{\partial x^{i}}$, that evidently raises the degree of a ``function'' (differential form) by +1; and certainly it squares to zero, $\dd  ^2=0$. Thus,  $(T[1]\Sigma,d)$ is a Q-manifold.

\begin{remarque}
	When the grading of the Q-manifold is concentrated in positive degrees, we say that it is an \emph{NQ-manifold}. When it does not involve any coordinate function of degree 0, this NQ-manifold can be seen as a graded vector space, what we will call an \emph{NQ-vector space}. Below we will have bundels in the category of Q-manifolds; in all those considered here, the base will be $(T[1]\Sigma,d)$ and the fiber an NQ-vector spaces.
\end{remarque}

To encode the higher gauge theory and the generalized Bianchi identities described above into the Q-formalism, we take an $N$-dimensional graded vector space  $\C{W}$ concentrated in negative degrees. It can be decomposed as a sum of subspaces which contain homogeneous elements excusively: $\C{W}=\C{W}_{-1}\oplus\C{W}_{-2}\oplus...\oplus\C{W}_{-m}$ where each one of the subspaces $\C{W}_{-|\alpha|}$ is spanned by $N_{|\alpha|}$ basis vectors $q_{\alpha}$ of degree $-|\alpha|$. The most general degree $+1$ vector field (not necessarily squaring to zero) can be parametrized as follows:
\begin{equation}
	\begin{aligned} \label{eq:Qansatz}
	Q = (-\frac{1}{2}C_{ab}^{c}q^{a}q^{b}	&+t^{c}_{I}q^{I})\frac{\partial}{\partial q^{c}}
		+(-\Gamma^{I}_{aJ}q^{a}q^{J}+\frac{1}{6}H^{I}_{abc}q^{a}q^{b}q^{c}+t^{It}q_{t})\frac{\partial}{\partial q^{I}}\\
										&+(A^{s}_{at}q^{a}q_{s}+B_{IJt}q^{I}q^{J}+D_{abIt}q^{a}q^{b}q^{I}+E_{abcdt}q^{a}q^{b}q^{c}q^{d}+t^{\lambda}_{t}q_{\lambda})\frac{\partial}{\partial q_{t}}+\dots\;,
	\end{aligned}
\end{equation}
where $\ldots$ denote higher order terms, such as $\frac{\partial}{\partial q_{\lambda}}$, etc. Recall that up to now, $(\C{W},Q)$ is just a graded manifold equipped with a degree $+1$ vector field, and not a Q-manifold yet.

In such a context, we can interpret the gauge field $A^{\alpha}$ as the pull-back of the dual basis vector $q^{\alpha}$ by some degree preserving map $a:T[1]\Sigma\longrightarrow\C{W}$:
\begin{equation}
	\begin{aligned}
		a^{\ast} :	\C{C}^{\infty}(\C{W})	&\longrightarrow\Omega^{\bullet}(\Sigma)\cong \C{C}^{\infty}(T[1]\Sigma)\\
							q^{\alpha}		&\longmapsto A^{\alpha}=a^{\ast}(q^{\alpha})
	\end{aligned}
	\label{defa}
\end{equation}
It is now easy to see that we can recover the general ansatz of a field strength by means of a degree one vector field $Q$ and the above definitions:
\begin{definition}
\label{def:field}
	The field strength is a degree $+1$ map depending on the map $a$:
	\begin{equation}
	\begin{aligned}
		\C{F} :	\C{C}^{\infty}(\C{W})	&\longrightarrow\Omega^{\bullet}(\Sigma)\cong \C{C}^{\infty}(T[1]\Sigma)\\
							q^{\alpha}	&\longmapsto \C{F}^{\alpha}=\dd  a^{\ast}(q^{\alpha})-a^{\ast} Q(q^{\alpha})
	\end{aligned}
	\end{equation}
\end{definition}
We recommend the reader to reproduce the general ansatz eqs.~(\ref{eq:Fa}) for the field strengths by means of the general ansatz (\ref{eq:Qansatz}) for the vector field $Q$ (with the general parametrizations matching as given --- note also that $Q(q^\alpha)= Q^\alpha$ if $Q \equiv Q^\alpha \partial_\alpha$). For a general $Q$, squaring to zero or not,  $\C{F}$ is the obstruction for $a$ to be a Q-morphism. 

In general, when starting to construct a general gauge theory, it is more common to focus on the gauge symmetries at the very beginning and simultaneously on the action functional or something like generalized field strengths and ``covariant derivatives'' as elementary objects for later construting an invariant functional. This has the disadvantage that one needs to tune two things simultaneously, the symmetries and the invariant/covariant/equivariant objects like the field strengths, both carrying independently their parameters in a general ansatz. One of the main observations in \cite{Gruetzmann-Strobl} was that even such a general notion of Bianchi identities as given by definition \ref{def:Bianchi} above is already very restrictive and turns out to subsequently restrict the gauge symmetries considerably. Splitting the problem like this in two parts, first the definition of the ``elementary'' objects and, only in a second step regarding the gauge symmetries that remain admissible in this context, turns out to be surprisingly effective.

\begin{theoreme}\label{thm:QBianchi}
\emph{\textbf{[Gr\"utzmann-Strobl 05]}}
Let $\dim \Sigma \equiv n \geq m+2$. Then any definition of field strengths 
(such as eqs.~(\ref{eq:Fa})) satisfies the Bianchi identities in the sense of  
Def.~\ref{def:Bianchi}, \emph{if and only if} the corresponding operator $Q$ (such as (\ref{eq:Qansatz}), using def.~\ref{def:field} for coordination of parameters) squares to zero, $Q^2 = 0$.
\end{theoreme}
This can be proven rather easily by means of the notions introduced above, using $\dd \circ \C{F}=\dd (\dd \circ a^{\ast}-a^{\ast}\circ Q)=-\dd \circ a^{\ast}\circ{Q}=-\C{F}\circ Q-a^{\ast}\circ Q^{2}$, but still, to our mind, deserves to be called a theorem due to the fact that it is so fundamental for the construction of higher gauge theories. The condition on the dimension comes from the fact that if the highest degree on $\C{V}$ would be $m=n-1$, for some coordinate $q^{\alpha}$, then $Q^{2}(q^{\alpha})$ is well defined as a not necessarily vanishing degree $n+1$ element of $\C{C}^{\infty}(\C{W})$, but its pull back on $T[1]\Sigma$ would vanish anyway due to the dimension of $\Sigma$.

There are situations where the dimensional condition in the above Theorem is violated, like for the case of ordinary gauge theories in two space-time dimensions since then $n=2$ but the presence of 1-form gauge fields implies $m=1$. Certainly, there exist important examples in two space-time dimensions which are still governed by a $Q$ squaring to zero, like ordinary Yang-Mills gauge theories or also the Poisson sigma model. One may argue that starting from $n=3$, i.e.~in particular for all physically relevant space-time dimensions, the condition becomes generically satisfied, since Hodge duality---induced by a metric on $\Sigma$ as present in physically relevant theory---permits to exchange any $p$-form gauge field for an $(n-p)$-form gauge field, so that it is sufficient to deal with gauge fields of a form-degree equalling essentially half the dimension of space-time. We remark, however, that this argument implies some constraints on the field equations for those potential higher form gauge fields since it should be possible that after the exchange for a lower form-degree gauge field they should turn into Bianchi-type equations. Although this is often assumed in physical theories, one does not necessarily need to assume it.

If the condition on $\Sigma$ is satisfied, the Theorem implies that the generalized Bianchi identities automatically induce a Q-structure to $\C{W}$, thus, as we will review in Section \ref{sec:Linfty} below, turning it into a $L_{m}[1]$ algebra, with associated $L_{m}$ algebra $\C{V}$ (we call an $m$-term $L_\infty$-algebra a $L_{m}$-algebra). The  explicit form of the generalized Bianchi identities for the three first field strengths can be also computed easily using $\dd   \C{F}^\alpha = - \C{F}(Q^\alpha)$, which follows as a simple corollary from the proof of the above theorem. One obtains:
\begin{align}
	\dd  \C{F}^{a}		&= C_{bc}{}^{a}\C{F}^{b}\wedge A^{c}-t^{a}_{I}\C{F}^{I}\;, \label{eq:BianchiF}\\
	\dd  \C{F}^{I}		&= \Gamma^{I}_{aK}(\C{F}^{a}\wedge B^{K} -A^{a}\wedge \C{F}^{K})-\frac{1}{2}H^{I}_{abc}\C{F}^{a}\wedge A^{b}\wedge A^{c}-t^{It}\C{F}^{(4)}_{t}\;,\nonumber\\
	\begin{split}
	\dd  \C{F}^{(4)}_{t}	&= -A_{at}^{s}(\C{F}^{a}\wedge C_{s}-A^{a}\wedge \C{F}^{(4)}_{s})-D_{abIt}(2\C{F}^{a}\wedge A^{b}\wedge B^{I}+A^{a}\wedge A^{b}\wedge \C{F}^{I})\nonumber\\
					&\hspace{2cm}-2B_{IJt}\C{F}^{I}\wedge B^{J}-4E_{abcdt}\C{F}^{a}\wedge A^{b} \wedge A^{c}\wedge A^{d} -t_{t}^{\lambda}\C{F}^{(5)}_{\lambda}\;.\nonumber
	\end{split}
\end{align}


\subsection{The gauge transformations}\label{subsection22}
We will discuss here the gauge transformations and why the Q-formalism is very natural to describe them. In analogy with the generalization we have made with the Bianchi identities, we will start from the classical form of a gauge transformation for a 1-form gauge field: 

\begin{equation}
\delta_{\lambda} A^{a}=D\lambda^{a}\equiv\dd \lambda^{a}+C^{a}_{bc}A^{b}\lambda^{c}\;,
\end{equation}
where $\lambda^{a}$ is some function on the space-time $\Sigma$. Having a tower of gauge fields, we can generalize this picture to the following one (generalized at the end of the subsection by terms containing field strengths):
\begin{align}
	\delta_{\lambda} A^{a}		&=\dd \lambda^{a} -  \bar{C}_{bc}{}^{a}\lambda^{b}A^{c} + \bar{t}^{a}_{I}\lambda^{I}\;,\label{eq:deltaAbar}
	\\
	\delta_{\lambda} B^{I}		&=\dd \lambda^{I} - \bar{\Gamma}^{I}_{aK}(\lambda^{a}B^{K}-A^{a}\wedge \lambda^{K}) + \frac{1}{2}\bar{H}^{I}_{abc}\lambda^{a}A^{b}\wedge A^{c} + \bar{t}^{It}\lambda_{t}\;,\nonumber\\
	\begin{split}
	\delta_{\lambda} C_{t}		&= \dd   \lambda_t+ \bar{A}_{at}^{s}(\lambda^{a} C_{s}-A^{a}\wedge \lambda_{s}) + 2\bar{D}_{abIt}\lambda^{a} A^{b}\wedge B^{I}+\hat{D}_{abIt}A^{a}\wedge A^{b}\wedge \lambda^{I}\nonumber\\
					&\hspace{2cm}+2\bar{B}_{IJt}\lambda^{I}\wedge B^{J}+4\bar{E}_{abcdt}\lambda^{a}\wedge A^{b} \wedge A^{c}\wedge A^{d} + \bar{t}_{t}^{\lambda}\lambda_{\lambda}\;,
	\end{split}
\end{align}
where $\lambda^{\alpha}$ is a $(|\alpha|-1)$-form on $\Sigma$. The barred and hatted coefficients have \emph{a priori} nothing to do with the coefficients of the homological vector field $Q$. But we will see below that in fact they do: requiring an appropriate notion of covariance (cf.~definition \ref{def:covariant} below) forces the barred and hatted quantities to be the same as the plain ones. We generalize the above ansatz slightly by asking that the generic form of a gauge transformations is:
\begin{equation}\label{eq:genericform}
	\delta_{\lambda} A^{\alpha}=\dd \lambda^{\beta}W^{\alpha}_{\beta}
+\lambda^{\beta}V^{\alpha}_{\beta} \, ,
\end{equation}
where the coefficients $W^\alpha_\beta$ and $V^{\alpha}_{\beta}$ are functions of the gauge fields $A^\gamma$ in general (but still not of the field strengths). 
So, for example, the second equation in (\ref{eq:deltaAbar}) corresponds to the choice:
\begin{equation}
W^{I}_{K}=\delta^{I}_{K},\quad
V^{I}_{a}=-\bar{\Gamma}^{I}_{aK}B^{K}+ \frac{1}{2}\bar{H}^{I}_{abc}A^{b}\wedge A^{c},\quad
V^{I}_{K}=-\bar{\Gamma}^{I}_{aK}A^{a}, \quad
V^{It}=\bar{t}^{It} \, ,
\end{equation}
but in general we will permit $W^\alpha_\beta$ to be different from the identity $\delta^\alpha_\beta$. Taking the ansatz (\ref{eq:genericform})  for gauge transformations, we now address the question of the conditions on the coefficient parameters such that the field strengths transform ``covariantly''. This notion is not yet clearly defined without a clear algebraic-geometric understanding of the structural equations and so, 
like for the generalized Bianchi identities, we want to give a definition as general as possible so that at least every known physical example is encompassed:
\begin{definition}
\label{def:covariant}
Let $\C{I}\subset {\cal A}$ be the ideal generated by field strengths $\C{F}^{\alpha}$ as in definition \ref{def:Bianchi}. The field strengths $\C{F}^{\alpha}$ are said to \emph{transform covariantly} under some set of gauge transformations (such as eqs.~(\ref{eq:deltaAbar}) or, more generally, eq.~(\ref{eq:genericform}) for some fixed choice of the parameter functions $W$ and $V$), iff for any choice of the $A^\alpha$s and $\lambda^\alpha$s the ideal is stable, i.e.		\begin{equation}
		\delta \C{I}\subset\C{I} \, . 
	\end{equation}
\end{definition}
Although again this definition is kept very general, it turns out to be surprisingly restrictive. For the following we always assume that the definition of field strengths is such that they satisfy some generalized Bianchi identities in the sense of def.~\ref{def:Bianchi}. 

We will now show that the Q-formalism again helps to determine the conditions on the gauge transformations in a very concise form. First of all, to introduce the gauge parameters $\lambda^{\alpha}$ in the Q-formalism, we have to slightly modify the picture, by extending trivially the range of the map $a$ from $\C{W}$ to $T[1]\Sigma\times\C{W}$, which will be called $\C{M}$. Since $a$ acts as the identity from $T[1]\Sigma$ to itself, the pullback $a^{\ast}$ act as the identity on $\C{C}^{\infty}(T[1]\Sigma)=\Omega^{\bullet}(\Sigma)$, then $a^{\ast}(\lambda^{\alpha})=\lambda^{\alpha}$.

Another important point is that even if we have glued the shifted tangent bundle to $\C{W}$, it does not change the definition of the field strengths and of the generalized Bianchi identities. Indeed, we can turn $\C{M}$ into a Q-manifold as soon as we take the following new homological vector field: $\Q=\dd+Q$, because $\dd^{2}=0$, $\dd\circ Q+Q\circ \dd=0$, and $Q^{2}=0$. Now concerning the Q-formalism, we wonder if there exists a degree zero vertical vector field $U$ on $\C{M}$ such that:
\begin{equation}
	\delta A^{\alpha}=(\delta_{\C{M}} a)^{\ast}(q^{\alpha})=a^{\ast}(U(q^{\alpha}))\label{eq:deltaU}
\end{equation}
and if it does, what are the conditions $U$ has to satisfy to ensure that $\delta\C{I}\subset \C{I}$? Vertical means that it is tangent to $\C{W}$, and one of the main advantages of adding $T[1]\Sigma$ to $\C{W}$ is that now $U$ can naturally depend on the gauge parameters $\lambda^{\alpha}$, which are differential forms on the base manifold $\Sigma$ (as well as on the coordinates $q^{\alpha}$ corresponding to the fields). In particular, one has:
\begin{align}
	\F{X}^{\text{vert}}_{0}(\C{M} \to T[1]\Sigma)= \C{C}^{\infty}(\Sigma)	&\otimes\F{X}_{0}(\C{W})\ \oplus\
									\Omega^{1}(\Sigma)\otimes\F{X}_{-1}(\C{W})\ \oplus\ 
									\Omega^{2}(\Sigma)\otimes\F{X}_{-2}(\C{W})
									\nonumber \\
									&\oplus\hspace{0.5cm}...\hspace{0.5cm}\oplus\
									\Omega^{|\text{dim}(\Sigma)|}(\Sigma)\otimes\F{X}_{-|\text{dim}(\Sigma)|}(\C{W})
\end{align}
It is easy to see that the general type of gauge transformations (\ref{eq:genericform}) can be reproduced by means of (\ref{eq:deltaU}) for the choice:
\begin{equation}
	U\equiv U(\lambda^{\beta},q^{\gamma})=\left(\dd \lambda^{\beta}W^{\alpha}_{\beta}(q^{\gamma})+\lambda^{\beta}V^{\alpha}_{\beta}(q^{\gamma})\right)\frac{\partial}{\partial q^{\alpha}} \, , \label{eq:U}
\end{equation}
where we indicated the arguments of the coefficient functions $W$ and $V$, which result from the corresponding expressions in (\ref{eq:genericform}) by replacing $A^\alpha$ by $q^\alpha$ everywhere. One then finds in a slight generalization of the results of \cite{Gruetzmann-Strobl}:
\begin{proposition}\label{propositiongauge1}
Let $n>m$. 
Gauge transformations of the form (\ref{eq:genericform}) induce covariant transformations of the field strengths in the sense of definition \ref{def:covariant} if and only if the vector field $U$ in eq.~(\ref{eq:U}) takes the form $U=[\Q,\lambda^{\beta}W^{\gamma}_{\beta}\frac{\partial}{\partial q^{\gamma}}]$.
\end{proposition}
\begin{remarque}
	This implies that the vector field $U$ is $ad_{\Q}$-exact. ($ad_{\Q}$ is a differential on the space of vector fields due to $Q^2=0$). The equality on $U$ is equivalent to $V^{\alpha}_{\beta}=[W^{\gamma}_{\beta}\frac{\partial}{\partial q^{\gamma}},Q]^{\alpha}$, or, if we permit the functions $W$ and $V$ to also depend explicitly on $\Sigma$, to $V^{\alpha}_{\beta}=[W^{\gamma}_{\beta}\frac{\partial}{\partial q^{\gamma}},Q+\dd]^{\alpha}$.
\end{remarque}
\begin{proof}
The gauge variation of the field strength $\C{F}^{\alpha}$ is:
\begin{equation}
	\delta_{\lambda}\C{F}^{\alpha}=\left(\dd \circ (a^{\ast} \circ U)-(a^{\ast}\circ U)\circ \Q\right) (q^{\alpha})=\C{F}(U(q^{\alpha}))+(a^{\ast}\circ[\Q,U])(q^{\alpha})\, .
\end{equation}
Assuming that the dimension of $\Sigma$ is at least one bigger than the highest degree $m$ of the graded coordinates (so that the second term does not vanish already by degree reasons),  necessarily $U$ should be $ad_{\Q}$ closed:
\begin{equation}
	0= [\Q,U] \equiv [\dd + Q,U] = \dd\lambda^{\beta} \left(V^{\alpha}_{\beta}\frac{\partial}{\partial q^\alpha} -[W^{\gamma}_{\beta}\frac{\partial}{\partial q^{\gamma}},\Q]\right)-\lambda^{\beta} [V^{\gamma}_{\beta}\frac{\partial}{\partial q^{\gamma}},\Q]
\end{equation}
from which we conclude first that the last term has to vanish by itself (choose locally constant gauge parameters $\lambda^{\beta}$ on $\Sigma$). This in turn implies that the first term has to vanish by itself as well, or  $V^{\alpha}_{\beta}=[W^{\gamma}_{\beta}\frac{\partial}{\partial q^{\gamma}},\Q]^{\alpha}=[W^{\gamma}_{\beta}\frac{\partial}{\partial q^{\gamma}},Q]^{\alpha}$, which then implies the vanishing of the second term as well by $ad_{Q_{\C{M}}}^2=0$ and is also easily seen to imply the statement on $U$.
\end{proof}

We see that the gauge transformations given in Eqs.~(\ref{eq:deltaAbar}), corresponding to $W^{\alpha}_{\beta}=\delta^{\alpha}_{\beta}$, preserve the differential ideal $\C{I}$ if and only if $\delta_{\lambda} A^{\alpha}=a^{\ast}([\lambda^{\beta}\frac{\partial}{\partial q^{\beta}},\Q](q^{\alpha}))$. This in turn implies that, as anticipated already by the notation, the barred/hatted coefficients are the same as the unbarred/unhatted coefficients of the homological vector field $Q$. From now on $W_{\beta}: = W_{\beta}^{\gamma}\frac{\partial}{\partial q^{\gamma}}$ and we define the vertical degree $-1$ vector field $\epsilon := \lambda^{\beta}W_{\beta}\in\F{X}^{\text{vert}}_{-1}(\C{M})$. With this notation, the system of gauge transformations can be written compactly as:
\begin{equation}
	\delta_{\lambda} A^{\alpha}=a^{\ast}\big([\Q,\epsilon](q^{\alpha})\big)
	\;. \label{deltaAalpha}
\end{equation}

The Q-formalism is very convenient to compute the commutator of two gauge transformations. In fact, using  $\Q^{2}=0$, we obtain
\begin{equation}
	[\delta_{\lambda},\delta_{\lambda'}] A^{\alpha}=a^{\ast}\big(\big[[\Q,\epsilon],[\Q,\epsilon']\big](q^{\alpha})\big)=a^{\ast}\big([\Q,\widehat{\epsilon}](q^{\alpha})\big) \, ,  \label{eq:comm}
\end{equation}
where the new gauge parameter $\widehat{\epsilon} \equiv \widehat{\lambda}^\alpha(\lambda,\lambda') W_\alpha$ can be obtained as a derived bracket 
$\widehat{\epsilon} = [\epsilon,\epsilon']_{\Q}\equiv \big[[\Q,\epsilon],\epsilon'\big]$  on $\F{X}^{\text{vert}}_{-1}(\C{M})$. Note, however, that the new degree $-1$ vector field $\widehat{\epsilon}$ parametrizing the gauge symmetries is defined only up to  $\mathrm{ad}_{\Q}$-closed vector fields of degree $-2$. This is important since the derived bracket is in general \emph{not} antisymmetric per se:
\begin{equation}
	[\epsilon,\epsilon']_{\Q}=-[\epsilon',\epsilon]_{\Q}+\big[\Q,[\epsilon,\epsilon']\big] \, ,
	\label{ex}
\end{equation}
but only up to an $\mathrm{ad}_{\Q}$-exact term. Thus, although $(\F{X}^{\text{vert}}_{-1}(\C{M}),[\cdot,\cdot]_{\Q})$ is not a Lie algebra in general, the quotient $\F{X}$ of $\F{X}^{\text{vert}}_{-1}(\C{M})$ by the subspace of $\mathrm{ad}_{\Q}$-closed vector fields is.
For determining (a representative of the equivalence class of) the new parameter $\widehat{\lambda}$ as a function of the two parameters $\lambda$ and $\lambda'$, we can thus restrict to the antisymmetric part of the derived bracket: $\widehat{\epsilon}\equiv \widehat{\lambda}^{\beta}W_{\beta}:=[\epsilon,\epsilon']^{A}_{\Q}$.

For a concrete example, let us take the gauge transformations generated by $\epsilon=\lambda^{\beta}\frac{\partial}{\partial q^{\beta}}$\,,
\begin{align}
	 \label{eq:deltaA}\delta_{\lambda} A^{a}		&=\dd \lambda^{a} -  C_{bc}{}^{a}\lambda^{b}A^{c} + t^{a}_{I}\lambda^{I}\;,\\
	\delta_{\lambda} B^{I}		&=\dd \lambda^{I} - \Gamma^{I}_{aK}(\lambda^{a}B^{K}-A^{a}\wedge \lambda^{K}) + \frac{1}{2}H^{I}_{abc}\lambda^{a}A^{b}\wedge A^{c} + t^{It}\lambda_{t}\;,\nonumber\\
	\begin{split}
	\delta_{\lambda} C_{t}		&=\dd\lambda_{t} + A_{at}^{s}(\lambda^{a} C_{s}-A^{a}\wedge \lambda_{s}) + D_{abIt}(2\lambda^{a} A^{b}\wedge B^{I}+A^{a}\wedge A^{b}\wedge \lambda^{I})\nonumber\\
					&\hspace{2cm}+2B_{IJt}\lambda^{I}\wedge B^{J}+4E_{abcdt}\lambda^{a}\wedge A^{b} \wedge A^{c}\wedge A^{d} + t_{t}^{\lambda}\lambda_{\lambda}\;,\nonumber
	\end{split}
\end{align}
yielding
\begin{align}\label{hutepsilon}
	\widehat{\epsilon}^{a}=[\epsilon,\epsilon']_{\Q}(q^{a})	&=	C_{bc}^{a}\lambda^{b}\lambda'^{c}\\
	\widehat{\epsilon}^{I}=[\epsilon,\epsilon']_{\Q}(q^{I})	&=	\Gamma^{I}_{aK}(\lambda^{a}\lambda'^{K}-\lambda'^{a}\lambda^{K})-H^{I}_{abc}\lambda^{a}\lambda'^{b}q^{c}\;,\nonumber\\
	\begin{split}
	\widehat{\epsilon}_{t}=[\epsilon,\epsilon']_{\Q}(q_{t})	&=	-A_{at}^{s}(\lambda^{a}\lambda'_{s}-\lambda'^{a}\lambda_{s})+2B_{IJt}\lambda^{I} \lambda'^{J}
	-12E_{abcdt}\lambda^{a}\lambda'^{b} q^{c}  q^{d}\;\nonumber\\
								&\hspace{3cm}+2D_{abIt}(\lambda^{a} q^{b}  \lambda'^{I}-\lambda'^{a}  q^{b}  \lambda^{I}-\lambda^{a} \lambda'^{b}  q^{I})\;.\nonumber
	\end{split}
\end{align}
In this example, the derived bracket turns out to already be antisymmetric, since evidently $[\epsilon,\epsilon']=0$ as these vector-fields are ``constant'' along the fibers. Thus, we can directly read off $\widehat{\lambda}^a$, $\widehat{\lambda}^I$, and $\widehat{\lambda}_t$ from the three lines above, as the pullback by $a^{\ast}$ of the components of the vector field $\widehat{\epsilon}$, $\widehat{\lambda}^{\beta}=a^{\ast}(\widehat{\epsilon}^{\beta})$. This calculation is much shorter than calculating the commutator directly using the definition (\ref{eq:deltaA}) and the identities satisfied by the coefficients that are encoded into $Q^2=0$ in the derived-bracket calculation. 

Following \cite{Bojowald:2004wu,Gruetzmann-Strobl}, we can easily extract from this the commutator of two gauge transformations, even if they do \emph{not} close. This works as follows: Using Equation \eqref{eq:comm} as well as $\C{F} \equiv  \dd \circ a^* - a^* \circ Q_{\C{M}}$, we obtain ($Q_{\C{M}}(q^\alpha)=Q(q^\alpha)\equiv Q^\alpha$)
\begin{equation} \label{ll'}
[\delta_{\lambda},\delta_{\lambda'}] A^{\alpha}=a^{\ast}\big([\Q,\widehat{\epsilon}](q^{\alpha})\big) =\dd a^{\ast}(\widehat{\epsilon}^{\alpha})-\C{F}(\widehat{\epsilon}^{\alpha})+a^{\ast}(\widehat{\epsilon}(Q^{\alpha}))=\delta_{\widehat{\lambda}}A^{\alpha}-\C{F}(\widehat{\epsilon}^{\alpha}) \, . 
\end{equation}
While the original gauge parameters $\epsilon^\alpha = \lambda^\alpha$ depend only on coordinates of $T[1]\Sigma$, the new ones depend also explicitly on the coordinates $q^\gamma$ (before the pullback by $a$). This has the following effect: A function on the total bundle $\C{M} =T[1]\Sigma \times \C{W}$ that comes from the base, i.e.~that only depends on coordinates of $T[1]\Sigma$, lies in the kernel of the operator $\C{F}$. Correspondingly, one has $\C{F}(\epsilon^\alpha) = \C{F}(\epsilon'^\alpha) = 0$, which is essential in identifying the gauge transformations as written in Equations \eqref{eq:deltaA} with Equation \eqref{deltaAalpha} (and likewise so for $\delta_{\lambda'}A^\alpha$). However, 
$\C{F}(\widehat{\epsilon}^\alpha) \equiv \C{F}^\beta \wedge a^*\left(\partial_\beta (\widehat{\epsilon}^{\alpha})\right)$ does no more vanish, at least for those components 
of $\widehat{\epsilon}^{\alpha}$ which depend on $q^\gamma$, like the second and third line of Equation \eqref{eq:deltaA}. From \eqref{ll'} we then obtain directly:
\begin{eqnarray}
[\delta_{\lambda},\delta_{\lambda'}] A^a &=& \delta_{\widehat{\lambda}}A^a \label{offen} \\ {}
[\delta_{\lambda},\delta_{\lambda'}] B^I &=& \delta_{\widehat{\lambda}}B^I +H^{I}_{abc}\lambda^{a}\lambda'^{b}\C{F}^{c}
\nonumber \\ {} [\delta_{\lambda},\delta_{\lambda'}] C_t &=& \delta_{\widehat{\lambda}}C_t + 24 E_{abcdt}\lambda^{a}\lambda'^{b}\C{F}^{c}A^{d}+ 2D_{abIt}(\lambda'^{a}  \C{F}^{b}\wedge \lambda^{I}-\lambda^{a} \C{F}^{b}  \wedge \lambda'^{I}+\lambda^{a} \lambda'^{b}  \C{F}^{I}) \, . \nonumber
\end{eqnarray}

\smallskip

The field strengths  of definition \ref{def:field} for some map $a\colon T[1]\Sigma\longrightarrow\C{W}$ encoding the tower of gauge fields can be equivalently described by the map $f \colon T[1]\Sigma\longrightarrow T[1]\C{W}$ covering $a$. Denote by $\{q^{\alpha},\bar\dd  q^{\alpha}\}$ coordinates on $T[1]\C{W}$, where $\bar \dd$ is the de Rham differential on $\C{W}$, then
\begin{equation} \label{eq:fdef}
f^{\ast}(q^{\alpha}) := A^\alpha \equiv a^*(q^\alpha)\, , \quad
f^{\ast}(\bar\dd  q^{\alpha}):= \C{F}^{\alpha} \equiv \C{F}(q^\alpha) \, . 
\end{equation}
Equipping $T[1]\C{W}$ with the homological vector field 
$Q' = \bar{\dd}  + \C{L}_{Q}$, $(T[1]\C{W},Q')$ is a Q-manifold, and the map $f$ turns out to be a chain map or Q-morphism for any choice of $a$ \cite{Kotov:2007nr}: 
\begin{equation}
\label{eq:Qmorph}
	\dd \circ  f^{\ast}=f^{\ast}\circ Q' \, .
\end{equation} 
In the above, $\C{L}_{X}$ denotes the Lie derivative along $X$ on $\C{V}$: 
\begin{equation}
	\C{L}_{X}=[\iota_{X},\bar\dd ] \equiv \iota_{X}\circ\bar\dd +(-1)^{|X|}\bar\dd \circ\iota_{X}   \, .   
\end{equation}
Introducing the map $f$ covering $a$ gives us further freedom in parametrizing the gauge transformations. In particular, it is advantageous in the context of $\C{F}^\alpha$-dependent contributions in the gauge transformations of the gauge fields $A^\beta$. 

We first again turn to the picture with a section in a trivial bundle,  $\bar{\C{M}}=T[1]\Sigma\times T[1]\C{W}\to T[1]\Sigma$, where the total space is equipped with the Q-structure
\begin{equation}
	\QQ = \dd+Q' \, .
\end{equation}
As before we do not change the notation, but denote the section $f \colon T[1]\Sigma \to \bar{\C{M}}$ by the same letter as the corresponding map to the fiber. 

In this language, the previous gauge transformations (\ref{deltaAalpha}) can be easily reproduced by means of 
the use of the Lie derivative of $\epsilon$ \cite{Kotov:2007nr}:
\begin{align}
	f^{\ast}([\QQ,\C{L}_{\epsilon}](q^{\alpha}))
				&=f^{\ast}(\QQ\circ\C{L}_{\epsilon}(q^{\alpha}))+f^{\ast}(\C{L}_{\epsilon}(Q^{\alpha}))+f^{\ast}(\C{L}_{\epsilon}(\bar{\dd}q^{\alpha}))\label{eq:fastaast}\\
				&=\dd f^{\ast}(\epsilon(q^{\alpha}))+f^{\ast}(\epsilon(Q^{\alpha}))+f^{\ast}(-\bar{\dd}(\epsilon^{\alpha}))\nonumber\\
				&=\dd a^{\ast}(\epsilon(q^{\alpha}))+a^{\ast}(\epsilon(Q^{\alpha}))-\C{F}(\epsilon^{\alpha})\nonumber\\
				&=a^{\ast}([Q_{\C{M}},\epsilon](q^{\alpha}))\, ,
\end{align}
which proves the assertion. Here in the second equality we have used that $f^*$ is a chain map, in the third the defining equations (\ref{eq:fdef}), and in the last equality that for the bundle, $\C{F}= \dd \circ a^* - a^* \circ Q_{\C{M}}$, where $Q_{\C{M}} = \dd+Q$, $\dd$ denoting the de Rham differential on the base and thus acting trivially on $q^\alpha$ (but not on $\epsilon^\alpha$ since this is a function on the total bundle). So the concise formula (\ref{deltaAalpha}) is still valid with $f^{\ast}$:
\begin{equation}\label{deltaAalpha2}
	\delta A^{\alpha}=f^{\ast}([\QQ,\C{L}_{\epsilon}](q^{\alpha}))
\end{equation}

But now we are not restricted to vertical vector fields of degree $-1$ that can be written as Lie derivatives, we can permit more general such vector fields, in particular such that they also depend on $\bar{\dd}q^{\alpha}$ producing $\C{F}$-dependent coefficients. To distinguish these more general gauge transformations from the previous ones, we denote the parameters of those gauge transformations by $\Lambda^\alpha$ and add a bar on the symbol for the infinitesimal variation, $\delta_\Lambda = \bar{\delta}$, as well as over the degree $-1$ vector field parametrizing it on the bundle (according with the notation for the transition from $\C{M}$ to $\bar{\C{M}}$). This leads us to
\begin{theoreme}\label{thm:gauge}
Let $\dim \Sigma \equiv n \geq m+1$ and consider gauge transformation of the form 
\begin{equation}\label{eq:genericformtilde}
	\bar{\delta} A^{\alpha}=\dd \Lambda^{\beta}\bar{W}^{\alpha}_{\beta}
+\Lambda^{\beta}\bar{V}^{\alpha}_{\beta} \, ,
\end{equation}
where the coefficients $\bar{W}^{\alpha}_{\beta}$ and $\bar{V}^{\alpha}_{\beta}$ are arbitrary functions of the gauge fields $A^\gamma$  \emph{and} field strengths $\C{F}^\gamma$ (while, by definition, $\Lambda^\alpha \in \Omega^\bullet(\Sigma)$ are field-independent). Let
 $\bar{w}^{\alpha}_{\beta} = \bar{w}^{\alpha}_{\beta}(q,\bar{\dd}q)$ and $\bar{v}^{\alpha}_{\beta}= \bar{v}^{\alpha}_{\beta}(q,\bar{\dd}q)$ be the associated functions on $\bar{\C{M}}$ and ${w}^{\alpha}_{\beta}$ and ${v}^{\alpha}_{\beta}$, respectively, their evaluations on the zero section of $T[1]\C{V}$, i.e., e.g., ${w}^{\alpha}_{\beta}(q) = \bar{w}^{\alpha}_{\beta}(q,0)$.
Then one has:

Gauge transformations of the form (\ref{eq:genericformtilde}) induce covariant transformations of the field strengths in the sense of definition \ref{def:covariant} if and only if 
$v^{\alpha}_{\beta}=[w^{\gamma}_{\beta}\frac{\partial}{\partial q^{\gamma}},Q+\dd]^{\alpha}$. Moreover, the gauge transformations can be generated by means of 
\begin{equation}\label{eq:tildedelta}
\bar\delta A^{\alpha}=f^{\ast}([\QQ,\bar{\epsilon}](q^{\alpha}))
\end{equation}
for the vector field
\begin{equation}\label{eq:tildeepsilon}
\bar{\epsilon} = \C{L}_{\epsilon}+\left(\Lambda^{\beta}(\bar{v}_{\beta}^{\alpha}-v_{\beta}^{\alpha})
+\dd\Lambda^{\beta}(\bar{w}_{\beta}^{\alpha}-w_{\beta}^{\alpha})\right)\frac{\partial}{\partial \bar{\dd}q^{\alpha}}
\end{equation}
where $\epsilon=\Lambda^{\beta}w^{\gamma}_{\beta}\dfrac{\partial}{\partial q^{\gamma}}$.
\end{theoreme}
\begin{remarquelol}
\begin{itemize}
\item[] 
\item
If the coefficient functions in (\ref{eq:genericformtilde}) are not also explicit functions on $\Sigma$, one may drop the de Rham differential $\dd$ in the commutator expression $[w^{\gamma}_{\beta}\frac{\partial}{\partial q^{\gamma}},Q+\dd]^{\alpha}$ and one has simply $v^{\alpha}_{\beta}=[w^{\gamma}_{\beta}\frac{\partial}{\partial q^{\gamma}},Q]^{\alpha}$. 
\item 
If, more generally, the parameters $\Lambda^{\beta}$ depend on the fields and field strengths (or, when viewed upon as parameters on the Q-bundle, on $q^\alpha$ and $\bar{\dd}q^{\alpha}$), formula $(\ref{eq:tildeepsilon})$ has to be replaced by:
\begin{equation}\label{eq:tildeepsilonbis}
\bar{\epsilon} = \Lambda^{\beta}w^{\gamma}_{\beta}\dfrac{\partial}{\partial q^{\gamma}}+\left(\Lambda^{\beta}(\bar{v}_{\beta}^{\alpha}-[w_{\beta}^{\gamma}\frac{\partial}{\partial q^{\gamma}},\QQ]^{\alpha})
+\QQ(\Lambda^{\beta})(\bar{w}_{\beta}^{\alpha}-w_{\beta}^{\alpha})\right)\frac{\partial}{\partial \bar{\dd}q^{\alpha}}
\end{equation}
\item 
It is also straightforward to verify that a transformation induced by $[\QQ,\bar{\epsilon}]$ on the lifted bundle $\bar{\C{M}}\to T[1]\Sigma$ maps a section $f$ of this Q-bundle into another section of it (in the category of Q-manifolds). In particular, Eq.~\eqref{eq:Qmorph} remains true also for the transformed section. We recommend the reader to verify this explicitly.
\end{itemize}
\end{remarquelol}
\begin{proof}
First of all we rewrite equation (\ref{eq:genericformtilde}) as follows:
\begin{equation}\label{eq:genericformtildeb}
	\bar{\delta} A^{\alpha}=\dd \Lambda^{\beta}{W}^{\alpha}_{\beta}
+\Lambda^{\beta}{V}^{\alpha}_{\beta} + \dd \Lambda^{\beta}\left(\bar{W}^{\alpha}_{\beta} -{W}^{\alpha}_{\beta}\right)
+\Lambda^{\beta}\left(\bar{V}^{\alpha}_{\beta} -{V}^{\alpha}_{\beta}\right) \, ,
\end{equation}
with ${W}^{\alpha}_{\beta}=a^*{w}^{\alpha}_{\beta}$ and 
${V}^{\alpha}_{\beta}=a^*{v}^{\alpha}_{\beta}$. Since the coefficients in the brackets on the r.h.s.~lie in the differential ideal $\C{I}$ and the field strengths result from polynomials of $A$ and an application of the differential $\dd$, the field strengths stay covariant if and only if they do with respect to the first two terms in the above equation. The necessary and sufficient conditions for this were found in proposition \ref{propositiongauge1} above to be $v^{\alpha}_{\beta}=[w^{\gamma}_{\beta}\frac{\partial}{\partial q^{\gamma}},Q+\dd]^{\alpha}$ (cf.~in particular the remark following that proposition). 

It remains to check that the vector field (\ref{eq:tildeepsilon}) generates the symmetries by means of (\ref{eq:tildedelta}). The first part was verified already in (\ref{eq:fastaast}), the second part is a simple straightforward calculation:\begin{align}
f^{\ast}([\QQ,\bar{\epsilon}](q^{\alpha}))&=f^{\ast}\left(\QQ(\Lambda^{\beta}w^{\alpha}_{\beta})+\bar{\epsilon}(Q^{\alpha})+\bar{\epsilon}(\dd q^{\alpha})\right)\label{calcul explicite}\\
&=\R{d}(\Lambda^{\beta}W^{\alpha}_{\beta})+\Lambda^{\beta}W^{\gamma}_{\beta}\partial_{\gamma}Q^{\alpha}+\Lambda^{\beta}\bar{V}^{\alpha}_{\beta}-\Lambda^{\beta}f^{\ast}([\bar{w}_{\beta},\QQ]^{\alpha})+\R{d}\Lambda^{\beta}(\bar{W}_{\beta}^{\alpha}-W_{\beta}^{\alpha})\nonumber\\
&=(-1)^{|\beta|-1}\Lambda^{\beta}\R{d}W^{\alpha}_{\beta}+\Lambda^{\beta}W^{\gamma}_{\beta}\partial_{\gamma}Q^{\alpha}+\Lambda^{\beta}\bar{V}^{\alpha}_{\beta}\nonumber\\
&\hspace{2cm}-\Lambda^{\beta}W^{\gamma}_{\beta}\partial_{\gamma}Q^{\alpha}+(-1)^{|\beta|}\Lambda^{\beta}\R{d}W^{\alpha}_{\beta}+\R{d}\Lambda^{\beta}\bar{W}^{\alpha}_{\beta}\nonumber\\
&=\Lambda^{\beta}\bar{V}^{\alpha}_{\beta}+\R{d}\Lambda^{\beta}\bar{W}^{\alpha}_{\beta}=\bar\delta A^{\alpha}\,.\nonumber
\end{align}
\end{proof}

\begin{proposition}
With respect to the gauge transformations (\ref{eq:genericformtilde}),  \eqref{eq:tildedelta} the field strengths transform according to the following formula:
\begin{equation}
\bar\delta \C{F}^{\alpha}=f^{\ast}([\QQ,\bar\epsilon](\bar \dd q^{\alpha}))
\end{equation}
\end{proposition}

\begin{proof}
This is a straightforward calculation, using the formula $\C{F}^{\alpha}=\dd a^{\ast}(q^{\alpha})-a^{\ast}((\dd+\C{L}_{Q})(q^{\alpha}))$:
\begin{align}
	\bar\delta\C{F}^{\alpha}	&=\dd\bar\delta a^{\ast}(q^{\alpha})-\bar\delta a^{\ast}(\C{L}_{Q}(q^{\alpha}))\nonumber\\
						&=\dd f^{\ast}([\QQ,\bar\epsilon](q^{\alpha}))-f^{\ast}([\QQ,\bar\epsilon]\circ\C{L}_{Q}(q^{\alpha}))\nonumber\\
						&=f^{\ast}(\QQ\circ [\QQ,\bar\epsilon](q^{\alpha}))-f^{\ast}([\QQ,\bar\epsilon]\circ\C{L}_{Q}(q^{\alpha}))\nonumber\\
						&=f^{\ast}([\QQ,\bar\epsilon](\bar{\dd}q^{\alpha}))+f^{\ast}([\QQ,[\QQ,\bar\epsilon]](q^{\alpha}))
\end{align}
but $[\QQ,[\QQ,\bar\epsilon]]=\frac{1}{2}[\QQ^{2},\bar\epsilon]=0$, thus the result.
\end{proof}

Since the gauge transformations  (\ref{eq:genericformtilde}) can be expressed in the form  (\ref{eq:genericformtildeb}), their commutator again can be calculated by a derived bracket. We leave the details of this as an exercise to the reader.

\newpage

\section{Basics of the tensor hierarchy and the bosonic model in six dimensions}\label{section3}

In this section, we briefly review the six-dimensional tensor hierarchy constructed in \cite{Samtleben:2011fj},
based on the appearance of similar structures in gauged supergravity~\cite{deWit:2005hv,Bergshoeff:2007ef,deWit:2008ta}. It contains 1-form gauge fields $A^a$, 2-form gauge fields $B^I$ and 3-form gauge fields $C_t$, where the index $t$ is assumed to be ``dual'' to the index $a$ from now on (in other words, we can contract those indices). We still keep the notation with letters from the end of the alphabet for the lower indices so as to facilitate comparison with the general formulas of the previous section, where such a duality was not assumed. The theory is governed by a set of constants $b_{Irs}$, $d^I_{ab}\equiv d^I_{(ab)}$, $g^{Ir}$, $h_I^a$, $f_{ab}{}^c\equiv f_{[ab]}{}^c$ subject to the following relations: 
\begin{align}
	2(d^{J}_{r(u}d^{I}_{v)s}-d^{I}_{rs}d^{J}_{uv})h^{s}_{J}			&=2f_{r(u}{}^{s}d^{I}_{v)s}-b_{Jsr}d^{J}_{uv}g^{Is}\label{lol}\\
	(d^{J}_{rs}b_{Iut}+d^{J}_{rt}b_{Isu}+2d^{K}_{ru}b_{Kst}\delta^{J}_{I})h^{u}_{J}&=f_{rs}{}^{u}b_{Iut}+f_{rt}{}^{u}b_{Isu}+g^{Ju}b_{Iur}b_{Jst}\nonumber \\
	f_{[pq}{}^{u}f_{r]u}{}^{s}-\frac{1}{3}h_{I}^{s}d^{I}_{u[p}f_{qr]}{}^{u}&=0\nonumber \\
		h^{r}_{I}g^{It}													&=0\nonumber \\
	f_{rs}{}^{t}h_{I}^{r}-d^{J}_{rs}h^{t}_{J}h^{r}_{I}				&=0\nonumber \\
	g^{Js}h^{r}_{K}b_{Isr}-2h_{I}^{s}h_{K}^{r}d^{J}_{rs}				&=0\nonumber \\
	-f_{rt}{}^{s}g^{It}+d^{J}_{rt}h^{s}_{J}g^{It}-g^{It}g^{Js}b_{Jtr}	&=0\nonumber \\
	b_{Jr(s}d^{J}_{uv)}												&=0\;.
	\nonumber
\end{align}
Solutions to this system have been constructed in~\cite{Samtleben:2012mi}.

Strictly speaking, to render the theory physically consistent, and sticking to gauge fields up to degree 3 only, the gauge field $C_t$ always has to appear contracted with the tensor $g^{It}$. 
(Note that such contracted fields $\tilde{C}^I:= C_tg^{It}$ are in general \emph{not} independent gauge fields; for example, $h_I^r \tilde{C}^I\equiv 0$ as a consequence of the fourth equation of (\ref{lol}).) Alternatively, one can assume the existence of \emph{sufficiently} many higher form degree gauge fields, 4-forms $D_\lambda$, 5-forms $E_\omega$, $\ldots$, inducing corresponding higher field strengths parametrized by further constants. These  constants are subject to equations similar to (\ref{lol}), cf.~(\ref{extra}) below, and ensure that no further constraints are imposed on the constants above.

The gauge structure of the model is encoded in the tensors
\begin{equation}
	(X_{a})_{b}{}^{c} \equiv X_{ab}{}^c:= -f_{ab}{}^{c}+h^{c}_{I}d^{I}_{ab} 
\;,
\label{gen1}
\end{equation}
satisfying for the ``matrix commutator''
$[X_{a},X_{b}]\,_c{}^d := (X_{a})_{c}{}^{e}(X_{b})_{e}{}^{d}-(X_{b})_{c}{}^{e}(X_{a})_{e}{}^{d}$ the relation
\begin{equation}
	[X_{a},X_{b}]=-X_{ab}{}^{c}\,X_{c}
	\;.
	\label{algebra}
\end{equation}
The validity of this equation is a consequence of the constraints (\ref{lol}).
Note that the `structure constants' $X_{ab}{}^{c}$ can have a symmetric part as well, parametrized by $h^{c}_{I}d^{I}_{ab}$, which vanishes when contracted with $X_{c}$ (as a consequence of (\ref{algebra})). In particular, this implies that the standard Jacobi identities are not satisfied for the antisymmetric part $X_{[ab]}{}^c\equiv - f_{ab}{}^c$, cf.~the third equation of (\ref{lol}).

We can reinterpret Eq.~(\ref{algebra}) in a more abstract setting as the defining relation of a Leibniz algebra $(\mathbb{V}, [ \cdot,\cdot])$, reading the left-hand as the product or bracket in this algebra between basis elements $X_a$ of $\mathbb{V}$ \cite{Kotov:xxx}. Then 
\begin{equation}
	X_{a}\cdot A^{c} := -X_{ab}{}^{c}A^{b}\;.
	\label{covA}
\end{equation}
defines a representation of this Leibniz algebra on the 1-form fields $A^a$ in the sense that the map from the bracket to commutators is a morphism.  Then we also have a canonical action on the 3-form fields $C_r$ since they take values in a dual vector space, $X_{a}\cdot C_{t} = X_{at}{}^{s}C_{s}$. As a consequence of the relations (\ref{lol}), one may verify that 
\begin{align}
	X_{a}\cdot B^{I} &:= -(X_{a})_{J}{}^{I}B^{J} \;,
	\nonumber\\
	(X_{a})_{J}{}^{I} &\equiv2d^{I}_{ap}h^{p}_{J}-g^{It}b_{Jt a}	\;.
\label{covBC}
\end{align}
defines likewise an action on the 2-form fields $B^I$. 

By means of these operations one now defines ``covariant derivatives'' according to 
\begin{equation} \label{eq:D}
D=\dd -A^{a}X_{a} \cdot
\end{equation}
which appear e.g.~in the transformations of each gauge field $A^\alpha$ with respect to ``its own'' gauge parameter $\Lambda^\alpha$: $\delta A^\alpha = D \Lambda^\alpha$. In general, gauge fields $A^\alpha$ of form degree $|\alpha|$ transform under all gauge parameters $\Lambda^\beta$ with form degree up to $|\alpha|$, i.e.~for all $\beta$ with $|\beta| \leq |\alpha|+1$ (cf.~Egs.~(\ref{eq:transf}) below).

The field strengths defined in \cite{Samtleben:2011fj} are of the form 
\begin{align}
	\C{H}^{a} 		&= \dd A^{a}-\frac{1}{2}f_{bc}{}^{a}A^{b}\wedge A^{c} + h^{a}_{I}B^{I}\;,\label{fieldstrengthsH}\\
	\C{H}^{I} 		&= \dd B^{I}+(h^{s}_{K}d^{I}_{as}-g^{Is}b_{Ks a})A^{a}\wedge B^{K}+\frac{1}{6}f_{[ab}{}^{s}d^{I}_{c]s}A^{a}\wedge A^{b}\wedge A^{c} +d^{I}_{ab}A^{a}\wedge \C{H}^{b}+g^{It}C_{t}\;,\nonumber\\
	\C{H}^{(4)}_{t} 	&= \dd C_{t}+(f_{at}{}^{s}-d^{J}_{at}h^{s}_{J})A^{a}\wedge C_{s} +\frac{1}{2}h^{s}_{I}b^{\phantom{}}_{Jt s}B^{I}\wedge B^{J}+\frac{1}{3}h^{v}_{I}b^{\phantom{}}_{Kta}d^{K}_{bv}A^{a}\wedge A^{b}\wedge B^{I}\nonumber\\
 					&\quad{}-\frac{1}{12}b^{\phantom{}}_{Kta}f_{bc}{}^s d^{K}_{ds}A^{a}\wedge A^{b}\wedge A^{c}\wedge A^{d}-\frac{1}{3}b^{\phantom{}}_{Kta}d^{K}_{bc}A^{a}\wedge A^{b}\wedge \C{H}^{c}-b_{It a}B^{I}\wedge \C{H}^{a} +k^{\lambda}_{t}D_{\lambda}\;.\nonumber
\end{align}
Note that in these formulas we assumed already the existence of 4-form gauge fields $D_\lambda$, added at the end of the 4-form field strength with a new set of parameters $k_t^\lambda$ which is, however, supposed to disappear in the contraction with  $g^{It}$, governing the consistent truncation mentioned above. This implies 
\begin{equation}\label{eq:gk}
g^{It}k_t^\lambda =0
\end{equation}
as \emph{one} of the equations to be imposed on these new parameters (for further ones cf.~Eq.~(\ref{extra}) below). 
So, for the truncated system, one should replace $\C{H}^{(4)}_{t} = \dd C_t + \ldots$ by its projection $g^{It}\C{H}^{(4)}_{t} = \dd 
(g^{It}C_t) + \ldots\equiv \dd \tilde{C}^I + \ldots$. Since, as mentioned above, the gauge fields $\tilde{C}^I$ are constrained (not independent), it is, however, often useful to work with the full tower of gauge fields, considered to exist up to possibly arbitrarily high nontrivial order without impeding the lower orders. We will come back to this below.

For the field strengths (\ref{fieldstrengthsH}) we use a different notation  ${\cal H}^\alpha$ as opposed to the 
${\cal F}^\alpha$ introduced in the previous section, since their definition contains at each form degree also lower-degree field strengths. While such terms do not obstruct the analysis of the previous section, we still want to distinguish them by notation from the field strengths ``corrected'' for such contributions in what follows: 
\begin{eqnarray}
	\C{F}^{a}		&=&\C{H}^{a}\;,\label{FH}\\
	\C{F}^{I}		&=&\C{H}^{I}-d^{I}_{ab}A^{a}\wedge\C{H}^{b}\;,\nonumber\\
	\C{F}^{(4)}_{t}	&=&\C{H}^{(4)}_{t}+b_{Jta}\C{H}^{a}\wedge B^{J}+\frac{1}{3}b^{\phantom{}}_{Kt[a}d^{K}_{b]c}A^{a}\wedge A^{b}\wedge\C{H}^{c}\;.\nonumber
\end{eqnarray}
These expressions are now indeed of the form (\ref{eq:Fa}).

The full set of gauge transformations up to the (unconstrained) 3-forms reads
\begin{align}
\label{eq:transf}
	\delta A^{a} &= D\Lambda^{a}-h^{a}_{J}\Lambda^{J}\;,\\
	\delta B^{I} &= D\Lambda^{I}+d^{I}_{ab}A^{a}\wedge\delta A^{b}-2d^{I}_{ab}\Lambda^{a}\C{H}^{b}-g^{It}\Lambda_{t}\;,\nonumber\\
	\delta C_{t} &= D\Lambda_{t}+b_{Jt a}B^{J}\wedge\delta A^{a}+\frac{1}{3}b^{\phantom{}}_{Jt[a}d^{J}_{b]c}A^{a}\wedge A^{b}\wedge\delta A^{c}+b_{Jt a}\Lambda^{a}\C{H}^{J}+b_{Jt a}\Lambda^{J}\wedge \C{H}^{a}-k^{\lambda}_{t}\Lambda_{\lambda}\;.\nonumber
\end{align}
where $\Lambda_\lambda$ are the 3-form gauge parameters appearing in the 
gauge transformations of $D_\lambda$ via an appropriately defined\footnote{This is not yet defined since for it we need to specify the action of the Leibniz algebra on the set of 4-form gauge fields according to Eq.~(\ref{eq:D}). This will be done below only.} ``covariant derivative'': $\delta D_\lambda = D \Lambda_\lambda$; again this contribution disappears under the projection, as it should for consistency of the truncation.

\begin{remarque}
	Notice that the parametrization of the gauge transformations \eqref{eq:transf} is redundant, i.e.~there is a ``gauge symmetry'' of the gauge parameters (which would give rise to ghosts for ghosts in a BV-formalism). In particular, a change of parameters according to  
	\begin{align}\label{jauge}
	\Lambda^{a}			&\mapsto \Lambda^{a}	+ h^{a}_{I}\mu^{I}\;,
	&\Lambda^{I}			&\mapsto \Lambda^I +D\mu^{I}+g^{It}\mu_{t}\;,
	&\Lambda_{t}			&\mapsto \Lambda_t -b_{Jta}\mu^{J}\C{F}^{a}+D\mu_{t}\; 
\end{align}
for arbitrary $\mu^{I}\in\C{C}^{\infty}(\Sigma)$ and $\mu_{t}\in\Omega^{1}(\Sigma)$ does not change the transformations of the fields $A^a$, $B^I$ and $C_t$ above.
	\end{remarque}

One could have chosen another set of gauge parameters,
\begin{equation}\label{eq:changeLambda}
\Lambda^{r}\longmapsto\widetilde{\Lambda}^{r}\,,\quad \Lambda^{I}-d^{I}_{ab}\Lambda^{b}\wedge A^{a}\longmapsto\widetilde{\Lambda}^{I}\,,\quad \Lambda_{t}+b_{Jta}\Lambda^{a}B^{J}+\frac{1}{3}b^{\phantom{}}_{Jt[b}d^{J}_{c]a}\Lambda^{a}A^{b}\wedge A^{c}\longmapsto
\widetilde{\Lambda}_{t} \, ,
\end{equation}
such that in this new basis one has
\begin{align}
	\widetilde{\delta} A^{a}		&=\dd \widetilde{\Lambda}^{a} +f_{bc}{}^{a}\widetilde{\Lambda}^{b}A^{c} -h^{a}_{I}\widetilde{\Lambda}^{I}\;,\label{eq:deltaAnew}
	\\
	\widetilde{\delta} B^{I}		&=\dd \widetilde{\Lambda}^{I} +\left(g^{Is}b_{Ksa} -h^{s}_{K}d^{I}_{as} \right)\left(\widetilde{\Lambda}^{a}B^{K}-A^{a}\wedge \widetilde{\Lambda}^{K}\right) -\frac{1}{2}d^{I}_{s[a}f^{\phantom{}}_{bc]}{}^{s}\widetilde{\Lambda}^{a}A^{b}\wedge A^{c} -g^{It}\widetilde{\Lambda}_{t}-d^{I}_{rs}\widetilde{\Lambda}^{r}\C{F}^{s}\;,\nonumber\\
	\begin{split}
	\widetilde{\delta} C_{t}		&= \dd   \widetilde{\Lambda}_t+ \left( f_{at}{}^{r}-d^{J}_{at}h^{r}_{J}\right)\left(A^{a}\wedge \widetilde{\Lambda}_{r}-\widetilde{\Lambda}^{a} C_{r}\right)  -\frac{1}{3}h^{v}_{I}b^{\phantom{}}_{Kt[a}d^{K}_{b]v}\left(2\widetilde{\Lambda}^{a} A^{b}\wedge B^{I}+A^{a}\wedge A^{b}\wedge \widetilde{\Lambda}^{I}\right)\nonumber\\
					&\hspace{2cm}-h^{s}_{(J}b^{\phantom{}}_{K)ts}\widetilde{\Lambda}^{I}\wedge B^{J}+ \frac{1}{3}b^{\phantom{}}_{Kt[a}f^{\phantom{}}_{bc}{}^{s}d^{K}_{d]s}\widetilde{\Lambda}^{a}\wedge A^{b} \wedge A^{c}\wedge A^{d}-k^{\lambda}_{t}\widetilde{\Lambda}_{\lambda}\\
					&\hspace{4cm}+\left(b_{Jts}\widetilde{\Lambda}^{J}
					+\frac{2}{3}b^{\phantom{}}_{Jt[u}d^{J}_{v]s}\widetilde{\Lambda}^{u}A^{v}\right)\wedge \C{F}^{s}
					\; .
	\end{split}
\end{align}
The transformation of parameters (\ref{eq:changeLambda}) was lead by the following principle: Whenever a gauge transformation for a gauge field $A^\alpha$ carries in addition to the contribution $\dd \Lambda^\alpha$ terms of the form $\dd \Lambda^\beta A^\gamma \ldots$, where the dots can contain constant but also further field dependent contributions, we can rewrite this term according to $\dd \Lambda^\beta A^\gamma \ldots= \dd(\Lambda^\beta A^\gamma \ldots) - (-1)^{|\beta|} \Lambda^\beta \dd( A^\gamma \ldots)$. The first of these two terms, the exact one, combines into a (field-dependent) redefinition of $\Lambda^\alpha$, the second one contains no derivatives on $\Lambda$ anymore. Note that this transformation is always an invertible one for the gauge parameters (for reasons of form degree!) and thus all derivatives on gauge parameters can be reabsorbed into the standard first contribution like above. While in general we cannot get rid of field strength contributions to the gauge transformations in this manner, we can always get rid of terms that contain derivatives on the gauge parameters multiplied by some fields. 

An invertible change of parameters like in (\ref{eq:changeLambda}) corresponds to a change of the \emph{generating} set of gauge transformations (cf., e.g., \cite{Henneaux:1992ig}). In general it changes the algebra of gauge transformations, however, it does not change the feature of such an algebra being closed or ``open''. Geometrically speaking the generators of gauge transformations form a distribution in the space of fields. A closed algebra corresponds to an involutive distribution, while an open one signifies that one needs to add further generators or at least symmetry transformations; these can be also ``trivial ones'' (cf.~\cite{Henneaux:1992ig} for the terminology), i.e.~gauge transformations that vanish on the set of solutions to the field equations, but in that case also they constrain the functionals that can have the symmetries (in form of their field equations). A change of generators of a distribution corresponds to a different choice of a basis for the distribution only. This still changes their algebra: if the distribution is involutive, locally there even always exists an abelian choice. 

Comparing with the previous section, \emph{any} field strength contribution can be added to gauge transformations of the form (\ref{eq:deltaAbar}) without changing the fact if the ideal ${\cal I}$ generated by the field strengths is left stable or not. Thus, \emph{a priori}, i.e.~from the perspective of the definition of field strengths and thus the Q-structure, there is no good reason to favor the gauge symmetry generators (\ref{eq:transf}) or (\ref{eq:deltaAnew}) from, e.g., 
\begin{align}
	\delta_\lambda A^{a}		&=\dd \lambda^{a} +f_{bc}{}^{a}\lambda^{b}A^{c} -h^{a}_{I}\lambda^{I}\;,\label{eq:deltaAnewohneF}
	\\
	\delta_\lambda B^{I}		&=\dd \lambda^{I} +\left(g^{Is}b_{Ksa} -h^{s}_{K}d^{I}_{as} \right)\left(\lambda^{a}B^{K}-A^{a}\wedge \lambda^{K}\right) -\frac{1}{2}d^{I}_{s[a}f^{\phantom{}}_{bc]}{}^{s}\lambda^{a}A^{b}\wedge A^{c} -g^{It}\lambda_{t}\;,\nonumber\\
	\begin{split}
	\delta_\lambda C_{t}		&= \dd   \lambda_t+ \left( f_{at}{}^{r}-d^{J}_{at}h^{r}_{J}\right)\left(A^{a}\wedge \lambda_{r}-\lambda^{a} C_{r}\right)  -\frac{1}{3}h^{v}_{I}b^{\phantom{}}_{Kt[a}d^{K}_{b]v}\left(2\lambda^{a} A^{b}\wedge B^{I}+A^{a}\wedge A^{b}\wedge \lambda^{I}\right)\nonumber\\
					&\hspace{2cm}-h^{s}_{(J}b^{\phantom{}}_{K)ts}\lambda^{I}\wedge B^{J}+ \frac{1}{3}b^{\phantom{}}_{Kt[a}f^{\phantom{}}_{bc}{}^{s}d^{K}_{d]s}\lambda^{a}\wedge A^{b} \wedge A^{c}\wedge A^{d}-k^{\lambda}_{t}\lambda_{\lambda}
					\; .
	\end{split}
\end{align}
which differs from (\ref{eq:deltaAnew}) by simply dropping the field strength terms. While superficially this seems to simplify life, one needs to be aware of the fact that this step can change the nature of the constraint algebra, from closed to open. It is a much more drastic transition than the one from (\ref{eq:transf}) to (\ref{eq:deltaAnew}).  We will come back to this quesion in more detail in section 5 below.

With respect to the gauge transformations (\ref{eq:transf}) and (\ref{eq:deltaAnew}) 
the field strengths $\C{H}^\alpha$ have the remarkable property that they transform according to the respective representation of the Leibniz algebra:
\begin{equation}\label{eq:deltaH}
\delta\C{H}^{\alpha} =  -(X_{a})_{\beta}{}^{\alpha}\Lambda^{a}{\C{H}}^{\beta}\;.
\end{equation}
So they do \emph{not} transform with respect to the ``higher'' gauge parameters $\Lambda^I$, $\Lambda_r$ at all, they are strictly invariant with respect to these transformations. And with respect to the original Leibniz algebra, they follow the respective representation. This is of great advantage if one wants to construct invariant action functionals.\footnote{While there still exist examples of invariant functionals if such a condition is not satisfied, cf., e.g., \cite{Kotov:2010wr} and section \ref{sec:Action} below.}

While the definition of the first field strength  in (\ref{fieldstrengthsH}) can be motivated also e.g.~by $\frac{1}{2}[D,D]\Lambda^\alpha = \C{H}^a (X_{a})_\beta^\alpha \Lambda^\beta$, 
it is precisely the Chern-Simons-like contributions to the higher field strengths that ensures this property and, at the same time, makes them differ from their counterparts (\ref{FH}) that do not contain such terms. The 2-form part contribution in the 2-form field strength is also easily seen to be necessary for a non-Lie Leibniz algebra, if one decides for $\delta A^a = D\Lambda^a$:  the ``standard'' Yang-Mills field strength $\bar{\C{H}}^{a} \equiv \dd A^{a}-\frac{1}{2}f_{bc}{}^{a}A^{b}\wedge A^{c}$ transforms according to $\delta\bar{\C{H}}^{c} = -X_{ab}{}^{c}\Lambda^{a}\bar{\C{H}}^{b}-X_{(ab)}{}^{c}(A^{a}\wedge\delta A^{b}-2\Lambda^{a}\bar{\C{H}}^{b})$, i.e.~non-covariant in the very general sense of definition~\ref{def:covariant}; this is cured by the addition of the $B$-term to $\C{H}^a$ with the simultaneous requirement of how $B^I$ transforms w.r.t.~$\Lambda^a$.

Similarly, the Bianchi identities satisfied by this set of field strengths show a more particular structure as that one can extract from the general definition \ref{def:Bianchi}: ``naked'' gauge fields only appear linear in combination with the exterior derivative $\dd$ so as to combine into the above defined ``covariant derivatives'' $D$. In particular, one finds
\begin{align}
	D\C{H}^{a} 				&= h^{a}_{I}\C{H}^{I}\;, \label{eq:BianchiH}\\
	D\C{H}^{I} 				&= d^{I}_{ab}\C{H}^{a}\wedge \C{H}^{b}+g^{It}\C{H}^{(4)}_{t}
	\;,\nonumber\\
	D\left(g^{It} \C{H}^{(4)}_{t}\right)	 &= -g^{It}\,b_{Jt a}\,\C{H}^{a}\wedge \C{H}^{J}
\;.\nonumber
\end{align}
The rather strong covariance properties and Bianchi identities of the field strengths~(\ref{fieldstrengthsH})
turn out to be crucial for the construction of supersymmetric gauge invariant action functionals for these six-dimensional models~\cite{Samtleben:2011fj,Samtleben:2012fb,Bandos:2013jva}. 

Let us finally note, that the above system of forms $\{A^a, B^I, g^{Ir}C_{r}\}$ can be embedded
into a larger system of forms $\{A^a, B^I, C_{r}, k_r^\lambda D_\lambda\}$ in which the full
unprojected set of 3-forms $C_{r}$ appears together with 4-forms $D_\lambda$, which enter under
projection with the matrix $k_r^\lambda$ as in the definition of the field strengths (\ref{fieldstrengthsH}).
The Bianchi identities of this larger system are given by
\begin{align}
	D\C{H}^{a} 				&= h^{a}_{I}\C{H}^{I}\;,\\
	D\C{H}^{I} 				&= d^{I}_{ab}\C{H}^{a}\wedge \C{H}^{b}+g^{It}\C{H}^{(4)}_{t}
	\;,\nonumber\\
	D\C{H}^{(4)}_{t}			&= -b_{Jt a}\C{H}^{a}\wedge \C{H}^{J}+k^{\lambda}_{t}\C{H}^{(5)}_{\lambda}\;,\nonumber\\
	D \left(k_r^\lambda \C{H}^{(5)}_{\lambda} \right)
		&= k_r^\lambda\,c_{\lambda IJ}\,\C{H}^{I}\wedge\C{H}^{J}
		-k_r^\lambda\,c_{\lambda a}^{t}\,\C{H}^{a}\wedge\C{H}^{(4)}_{t}\;.
		\nonumber
\end{align}
Note that an eventual 6-form field strength does not appear on the r.h.s.~of the last equation since it disappears under the contraction with $k_r^\lambda$ effecting the trunctation one level up. Also, when contracted with $k_r^\lambda$, the action of $D$ on the l.h.s.~is understood as acting on an object with a lower $r$-index.
The new parameters appearing on the r.h.s.\ of the Bianchi identities are 
constrained to satisfy the conditions
\begin{align}
        4d^J_{ab} \,c_{\lambda\,IJ} &= b_{Ira} \,c_{\lambda b}^r + b_{Irb} \,c_{\lambda a}^r  \;,\label{extra}\\
	k^{\lambda}_{r}c_{\lambda IJ}&=h^{s}_{[I}b_{J]rs}\;, \nonumber\\
	k^{\lambda}_{r}c^{t}_{\lambda s}&=f_{rs}{}^{t}-b_{Irs}g^{It}+d^{I}_{rs}h_{I}^{t}\nonumber
\end{align}
together with (\ref{eq:gk}).
Despite first appearance, these conditions do not imply any further constraints on the
constants appearing in (\ref{lol}). E.g.\ combining the second  equation of (\ref{extra}) and Eq.~(\ref{eq:gk}) gives rise to the condition
\begin{equation}
g^{Kr}\,h^{s}_{[I}\,b_{J]rs} = 0\;,
\end{equation}
which can be shown to be a consequence of the previous set of equations (\ref{lol}).
In this sense, the system $\{A^a, B^I, g^{Ir}C_{r}\}$ is a consistent
truncation of the extended system $\{A^a, B^I, C_{r}, k_r^\lambda D_\lambda\}$.
In the next section, we will present a more systematic understanding of such truncations.

\newpage

\section{\texorpdfstring{$L_{\infty}$}{L infinity}-algebras and canonical truncations of Lie \texorpdfstring{$n$}{n}-algebras}\label{section4}\label{sec:Linfty}

In the present section we first want to recall the definition of an $L_\infty$-algebra and its equivalence with pointed Q-manifolds. Those Q-manifolds of relevance for the physics under consideration are always positively graded. If the highest non-zero degree of this socalled NQ-manifold is $p$, this corresponds to a $p$-term $L_\infty$-algebra or what we call a Lie $p$-algebra. After providing some of the details on this correspondence, where we will also partially follow the conventions of \cite{Mehta:2012} and  \cite{Voronov2005}, we address the canonical truncation of a Lie $p$-algebra to a Lie $q$-algebra for any $p \geq q\geq 2$ as it plays an important role in the tensor hierarchy.

Let $\C{V}=\oplus_{i\in\B{Z}}\C{V}_{i}$ be a graded vector space, and let us call $(\Lambda\C{V},\wedge)$ the free graded commutative algebra generated by $\C{V}$; let us remark that this is a purely polynomial algebra, without any completion. Given $j\geq1$, let $S_{j}$ be the permutation group of $j$ elements. A $(k,j-k)$-\emph{unshuffle} is a permutation $\sigma\in S_{j}$ such that $\sigma(1)<...<\sigma(k)$ and, if $k$ is strictly bigger than $j$, in addition $\sigma(k+1)<...<\sigma(j)$. We denote by $Un(j,k)$ the set of unshuffles. If one takes $j$ homogeneous elements $v_{1},...,v_{j}$ and some $\sigma\in S_{j}$, the \emph{Koszul sign} $\epsilon(\sigma)$ is defined as:
\begin{equation}
	q_{1}\wedge...\wedge q_{j}=\epsilon(\sigma)q_{\sigma(1)}\wedge...\wedge q_{\sigma(j)}
\end{equation}
where the elements $q_{k}$ are assumed to be homogeneous. Another sign of relevance below is $\chi(\sigma):=\text{sgn}(\sigma)\epsilon(\sigma)$. Then we can define an $L_\infty$- or, more generally, an $L_{n}$-algebra as follows:
%
%
\begin{definition}\label{def:Linfty}
	An $L_{n}$-algebra for some fixed $n \in \mathbb{N} \cup \{\infty\}$ is a graded vector space $\C{V}$ equipped with a collection of linear maps $[...]_{j} : \Lambda^{j}\C{V}\longrightarrow \C{V}$ of degree $2-j$, for $1\leq j\leq n+1\leq\infty$, such that for all homogeneous elements $q_{1},...,q_{m}\in\C{V}$:
	\begin{align}
		\mathrm{antisymmetry}\hspace{1cm}	& \forall\ \sigma\in S_{m}\hspace{1.5cm} [q_{1},...,q_{m}]_{m}=\chi(\sigma)[q_{\sigma(1)},...,q_{\sigma(m)}]_{m}\\
		\mathrm{Jacobi\ identity}\hspace{1cm}	& \sum_{k=1}^{m}(-1)^{k(m-k)}\sum_{\sigma\in Un(m,k)}\chi(\sigma)\big[[q_{\sigma(1)},...,q_{\sigma(k)}]_{k},q_{\sigma(k+1)},...,q_{\sigma(m)}\big]_{m-k+1}=0\, . \label{genJac}
	\end{align}
\end{definition}
This definition is similar to the one used in \cite{Mehta:2012,Lada:1992wc}, but differs for instance from \cite{Lada:1994mn}, in which the degree of the bracket $[...]_{j}$ is $j-2$; this is equivalent to the above definition upon reversal of the grading of $\C{V}$. In some cases, the notation $l_{j}(\ldots)$ is used  for the $j$-bracket, instead of $[\ldots]_{j}$. Note that any $L_n$-algebra is also trivially a (partially degenerate) $L_{n+m}$-algebra for any $m \in \mathbb{N}$ (the higher brackets vanish identically starting from $j=n+2$). 

The ``antisymmetry'' and ``Jacobi identity'' are to be understood in a generalized sense certainly. It is inspired by the case of an ordinary Lie algebra, which results from the above as follows: It is an $L_{1}$-algebra where all elements have degree 0 and where $[\cdot]_{1}$  vanishes identically. One then is left with the binary bracket $[\cdot,\cdot]_{2}$, which is antisymmetric and satisfies the Jacobi identity in the usual sense. Similarly, a $\mathbb{Z}$-graded Lie algebra is an $L_{1}$-algebra with $[\cdot]_{1}\equiv 0$.

The 1-bracket $[\cdot]_{1}$  equips $\C{V}$ with the structure of a complex due to the first Jacobi identity,  $[[\cdot]]_{1}=0$. We will denote this coboundary map also by $t$, or if the degree is to be specified, by
\begin{equation}\label{t}
t_{(p)} := [\cdot]_{1}|_{\C{V}_{-p}} \colon \C{V}_{-p}\to \C{V}_{-p+1} \, .
\end{equation}
It satisfies $t_{(p-1)} \circ t_{(p)} = 0$. A general $L_1$-algebra is then a differential graded Lie algebra, the compatibility of the bracket with the differential following from \eqref{genJac} for $m=2$. 

Note that an $L_n$-algebra has vanishing brackets starting from $j=n+2$, $[ \ldots ]_{n+2} \equiv 0$, but still the generalized Jacobi identities provide restrictions up to $m=n+2$.  For an $L_1$ algebra, it is the condition \eqref{genJac} for $m=3$ that gives the Jacobi identity (in that case only the terms with $k=2$ contribute). Similarly, it is this condition that controls the violation of the standard Jacobi identity of the 2-bracket for an $L_2$ algebra, where in this case there are also contributions with $k=3$ and $k=1$; the former ones correspond to terms which are of the form $t([\ldots]_3)$, the latter ones are of the form $[t(\cdot), \cdot, \cdot]_3$. The condition \eqref{genJac} for $m=4$ then expresses the fact that the 3-bracket has to obey some "Jacobiator identity" (see \cite{Baez:2003fs} for the case of Lie 2-algebras).

There is also a useful, equivalent version of the above definition, resulting from a shift of degree:
\begin{definition}
	An $L_{n}[1]$-algebra, $n \in \mathbb{N} \cup \{\infty\}$, is a graded vector space $\C{W}$ equipped with a collection of linear maps $\llbracket . , \ldots , . \rrbracket_{j} : \Lambda^{j}\C{W}\longrightarrow \C{W}$ of degree $1$, for $1\leq j\leq n+1\leq\infty$, such that for all homogeneous elements $q_{1},...,q_{m}\in\C{W}$:
	\begin{align}
		\mathrm{symmetry}\hspace{1cm}	& \forall\ \sigma\in S_{m}\hspace{1.5cm} \llbracket q_{1},\ldots ,q_{m}\rrbracket_{m}=\epsilon(\sigma) \llbracket
q_{\sigma(1)},\ldots ,q_{\sigma(m)} \rrbracket_{m} \label{sym'}\\
		\mathrm{Jacobi\ identity}\hspace{1cm}	& \sum_{k=1}^{m}\ \sum_{\sigma\in Un(m,k)}\epsilon(\sigma)\big\llbracket \llbracket q_{\sigma(1)},\ldots ,q_{\sigma(k)}\rrbracket_{k},q_{\sigma(k+1)},\ldots,q_{\sigma(m)}\big\rrbracket_{m-k+1}=0\,.\label{Jac'}
	\end{align}
\end{definition}

There is a natural bijection between an $L_{n}$-algebra structure on $\C{V}$ and an $L_{n}[1]$-algebra structures on $\C{W}=\C{V}[1]$, where $\C{V}[1]$ is the graded vector space modelled after $\C{V}$ such that $(\C{V}[1])_{i}=\C{V}_{i+1}$. If we write $[-1]$ for the \emph{shift functor} which shifts the degree of a homogeneous element of $\C{V}$ by $-1$ (i.e $[-1]\C{V}=\C{V}[1]$), then for all homogeneous elements $q_{1},...,q_{j}\in\C{V}$ we have the following relationship between $\llbracket \ldots \rrbracket_{j}$ and $[\ldots]_{j}$:
	\begin{equation}\label{crochetsbrackets}
		[-1][q_{1},\ldots,q_{j}]_{j}=(-1)^{(j-1)|q_{1}|+(j-2)|q_{2}|+\ldots+|q_{j-1}|}  \cdot \llbracket[-1] q_{1},\ldots,[-1] q_{j}\rrbracket_{j}   \, .
	\end{equation}
This version is better adapted to a direct comparison with (a priori $\mathbb{Z}$-graded) Q-manifolds. We will be more explicit on this in what follows.

Given a graded vector space $\C{W}$, there is a bijection between $L_{n}[1]$-algebra structures on $\C{W}$ and Q-structures on $\C{W}$ which are pointed, i.e.~which vanish at the origin $q^\alpha=0$,
\begin{equation}\label{Q0}
Q|_0 = 0 \, .
\end{equation}
The identification works as follows: In addition to the $\mathbb{Z}$-grading of $\C{W}$, the polynomials on this space have an \emph{independent} homogeneity degree (taking values in $\mathbb{N}_0$). Due to the condition \eqref{Q0}, this implies that we can decompose the vector field $Q$ according to $Q=\sum_{j=1}^\infty Q_{(j)}$ where $Q_{(j)}$ is a homogeneous vector field with a polynomial coefficient of homogeneity degree $j$. Note each $Q_{(j)}$ still has degree +1 and decomposes into several parts of different degrees in derivatives. To be more explicit, let  $\{q_{\alpha}\}$ denote  a basis of $\C{W}$, where each of the  $\{q_{\alpha}\}$s has a particular degree $-|\alpha| \in \mathbb{Z}$,  then a dual basis  $\{q_{\alpha}^{\ast}\equiv : q^\alpha\}$ induces coordinates $q^\alpha$ of degree $|\alpha|$ on $\C{W}$. The homological vector field $Q$ on $\C{W}$ can now be expressed in terms of the formal power series in this dual basis:
\begin{equation}\label{Qsum}
Q=  \sum_{j=1}^\infty Q_{(j)} \; , \qquad Q_{(j)}\equiv \frac{1}{j!} C_{\alpha_1 \ldots \alpha_j}^\beta q^{\alpha_1} \cdots q^{\alpha_j} \frac{\partial}{\partial q^\beta}\, ,
\end{equation}
where $C_{\alpha_1 \ldots \alpha_j}^\beta \in \mathbb{R}$ (and, as usually, the sum over repeating indices is understood).
Now, $Q_{(j)}$ corresponds precisely to the $j$-th bracket of the $L_\infty[1]$-algebra, and thus also to the corresponding $L_\infty$-algebra. In fact, 
\begin{equation}\label{derivedbracket}
	\llbracket q_{\alpha_{1}},\ldots,q_{\alpha_{j}}\rrbracket_{j}=C^{\beta}_{\alpha_{1}\ldots \alpha_{j}}q_{\beta}\: .
\end{equation}
Now it is straightforward to verify that the generalized Jacobi identity \eqref{Jac'} is equivalent to $Q^2=0$, while \eqref{sym'} already follows from the definitions \eqref{Qsum}, \eqref{derivedbracket}.

The Q-manifolds $\C{M}$ that appear in the present context are NQ-vector spaces, i.e.~Q-manifolds where the coordinate functions have positive degrees only. Since coordinates are elements in the dual, this implies that the corresponding graded vector space $\C{M}\equiv \C{W}=\C{V}[1]$ underlying the $L_\infty[1]$-algebra is concentrated in negative degrees, and the vector space of the $L_\infty$-algebra in non-positive degrees, $\C{V}=\oplus_{i=0}^\infty\C{V}_{-i}$.\footnote{In another possible convention, mentioned right after definition \ref{def:Linfty}, this would correspond to non-negative degrees. We stick to the present convention, however, since it is closer to the one used for NQ-manifolds, where then the coordinates have strictly positive degrees.}
\begin{definition} We call an $L_\infty$-algebra on $\C{V}=\oplus_{i=0}^{n-1}\C{V}_{-i}$ an \emph{$n$-term $L_\infty$-algebra} or simply a  \emph{Lie~$n$-algebra}. \end{definition}

Note that for degree reasons a Lie $n$-algebra is always an $L_n$-algebra (but certainly not vice versa). For example, a Lie 2-algebra is an $L_\infty$ structure defined on \begin{equation} \label{Lie2}
\C{V}_{-1} \stackrel{t}{\to} \C{V}_0 \, . 
\end{equation} Since $l_j \equiv [ \ldots ]_j$ has degree $2-j$ the only non-trivial brackets can be $l_1 \equiv t$, $l_2 \equiv [ \cdot , \cdot ]_2$ as well as a 3-bracket that is non-trivial when fed by three elements of degree 0, taking values in degree minus one then; thus here $l_3$ corresponds to an element of $\Lambda^3 \C{V}_{0}^* \otimes 
\C{V}_{-1}$. For vanishing $l_3$ this corresponds to a crossed module of Lie algebras or what is called a strict Lie 2-algebra in \cite{Baez:2003fs}, while a general Lie 2-algebra in the terminology used here is called semi-strict there (cf.~also \cite{Gruetzmann-Strobl} for more details on Lie 2-algebras and Lie 2-algebroids).  

From the above considerations we conclude that 
\begin{proposition} 
\label{propositiontruncation}
An NQ-vector space $(\C{W},Q)$ of highest coordinate degree $n$ is in\\ 1:1 correspondence with a Lie $n$-algebra $(\C{V}\equiv \oplus_{i=0}^{n-1}\C{V}_{-i}, [ \ldots ]_{j=1..n+1})$.
\end{proposition} 

For the sequel it will be important to observe that any Lie $n$-algebra can be truncated in a canonical way to a Lie $m$-algebra with $2 \leq m \leq n$.\footnote{A similar statement is also true for Lie $n$-algebroids, or even the more general $L_n$-algebras/algebroids, but these generalizations are of no relevance for the present considerations.} This is most easily seen in the Q-language: 
\begin{lemme} \label{lemmatrunc} Any NQ-vector space $(\C{W},Q)$ of highest coordinate degree $n\geq 2$ projects in a canonical way to an NQ-vector space $(\widetilde{\C{W}},\widetilde{Q})$ of highest coordinate degree $m$, $2 \leq m \leq n$. Up to degree $m-1$ the ring of functions on $\C{W}$ and $\widetilde{\C{W}}$ agrees, while $C^\infty(\widetilde{\C{W}})$ is completed by adding the image of the degree $m-1$ functions under the action of $Q_{(1)}$.\end{lemme}
\begin{proof}
Since coordinates on $\C{W}$ have positive degrees, the action of the part $Q_{(1)}$ of $Q$ in \eqref{Qsum} results in the same total algebra as if we act with all of $Q$ (all remaining parts of $Q$ acting on the degree $m-1$ functions are already generated by lower degree coordinates).
In this way $C^\infty(\widetilde{\C{W}})$ is obviously a subring of $C^\infty({\C{W}})$ closed under the differential and the inclusion of differential graded algebras corresponds to a projection of $Q$-manifolds. 
\end{proof}
The operator  $Q_{(1)}$ (when restricted to homogeneous coordinates) is \emph{dual} to the map 
$\llbracket \cdot \rrbracket_1$ and, by restriction to the degree $m-1$ coordinates on $\C{W}\equiv \C{V}[1]$, which are a dual basis of $\C{W}_{-m+1}$, this corresponds to the map $t_{(m-1)}$, cf.~eqs.~\eqref{t} and \eqref{crochetsbrackets} as well as the following diagram: 
\begin{align}
&\C{V}_{-m+1}\quad{\overset{t_{(m-1)}}{\xrightarrow{\hspace*{2cm}}}}\quad\C{V}_{-m+2}\nonumber \\
&[-1] \Big\downarrow\hspace{3.3cm}\Big\downarrow[-1]\nonumber \\
&\C{W}_{-m}\quad{\overset{\llbracket . \rrbracket_{1}}{\xrightarrow{\hspace*{2cm}}}}\quad\C{W}_{-m+1}\label{diagram}\\
&\hspace{0.35cm}\ast \Big\downarrow\hspace{3.3cm}\Big\downarrow\ast\nonumber\\
&\C{W}_{-m}^{\ast}\quad{\overset{Q_{(1)}}{\xleftarrow{\hspace*{2cm}}}}\quad\C{W}_{-m+1}^{\ast}\nonumber
\end{align}

In this way we see that the original complex of the Lie $n$-algebra, 
\begin{equation} \label{original}
\C{V}_{-n+1} \stackrel{t_{(n-1)}}{\to} \C{V}_{-n+2} \stackrel{t_{(n-2)}}{\to} \ldots \stackrel{t_{(m)}}{\to} \C{V}_{-m+1} \stackrel{t_{(m-1)}}{\to}\C{V}_{-m+2} \stackrel{t_{(m-2)}}{\to} \ldots \stackrel{t_{(2)}}{\to}\C{V}_{-1}\stackrel{t_{(1)}}{\to}\C{V}_{0} \, , \end{equation}
is shifted to
\begin{equation}
\C{W}_{-n} \stackrel{\llbracket.\rrbracket_{1}}{\to} \C{W}_{-n+1} \stackrel{\llbracket.\rrbracket_{1}}{\to} \ldots \stackrel{\llbracket.\rrbracket_{1}}{\to} \C{W}_{-m} \stackrel{\llbracket.\rrbracket_{1}}{\to}\C{W}_{-m+1} \stackrel{\llbracket.\rrbracket_{1}}{\to} \ldots \stackrel{\llbracket.\rrbracket_{1}}{\to}\C{W}_{-2}\stackrel{\llbracket.\rrbracket_{1}}{\to}\C{W}_{-1} \, , \end{equation}
so that the corresponding truncated NQ-vector space is most easily seen from the dual picture to become
\begin{equation}\C{W}_{-m}^{\ast}\supset\R{Im}(Q_{(1)}\big|_{\C{W}_{-m+1}^{\ast}})\overset{Q_{(1)}}{\longleftarrow}\C{W}_{-m+1}^{\ast}\overset{Q_{(1)}}{\longleftarrow} \ldots \overset{Q_{(1)}}{\longleftarrow}\C{W}_{-2}^{\ast}\overset{Q_{(1)}}{\longleftarrow}\C{W}_{-1}^{\ast} \,  .
\label{W*complex}
\end{equation}
It is an easy exercise in linear algebra to verify that for any map $f \colon \mathbb{V} \to \mathbb{W}$ between finite dimensional vector spaces the image $\mathrm{Im}(f^*)$ of the dual map $f^* \colon \mathbb{W}^* \to \mathbb{V}^*$ can be identified with $\left( \mathbb{V}/ \mathrm{ker}f\right)^*$. Applying this to our situation, we thus obtain from the above $\R{Im}(Q_{(1)}\big|_{\C{W}_{-m+1}^{\ast}})=\left(\C{W}_{-m} /\R{Ker}(\llbracket.\rrbracket_{1})\right)^{\ast}$, which implies that the 
 complex \eqref{W*complex} is dual to 
\begin{equation}
\bigslant{\C{W}_{-m}}{\R{Ker}(\llbracket.\rrbracket_{1})}\stackrel{\llbracket.\rrbracket_{1}}{\to}\C{W}_{-m+1} \stackrel{\llbracket.\rrbracket_{1}}{\to} \ldots \stackrel{\llbracket.\rrbracket_{1}}{\to}\C{W}_{-2}\stackrel{\llbracket.\rrbracket_{1}}{\to}\C{W}_{-1}\, .\end{equation}
Summing up, we see that together with Proposition \ref{propositiontruncation} the above Lemma implies that any  Lie $n$-algebra defined on the complex \eqref{original} can be truncated canonically to a Lie $m$-algebra defined on the complex 
\begin{equation}\bigslant{\C{V}_{-m+1}}{\mathrm{Ker}(t_{(m-1)})}\stackrel{\widetilde{t}_{(m-1)}}{\to} \C{V}_{-m+2} \stackrel{t_{(m-2)}}{\to} \ldots \stackrel{t_{(2)}}{\to}\C{V}_{-1}\stackrel{t_{(1)}}{\to}\C{V}_{0}\, ,
\label{quotient}\end{equation}
where ${\widetilde{t}_{(m-1)}}$ is the canonical descendant of the map ${t}_{(m-1)}$ to the quotient. Denote by $\tau$ the map from \eqref{original} to \eqref{quotient}, which, between degree 0 and degree $-m+2$ is the identity map, at degree $-m+1$ maps any element to its equivalence class generated by elements in the kernel of $ {t}_{(m-1)}$, and at all higher degrees is the zero map. Then the brackets $\widetilde{l}_{j}$ on the truncated 
Lie $m$-algebra are related to the original brackets on the Lie $n$-algebra $(\C{V},l_{j})$ by means of 
\begin{equation} \label{down}
\widetilde{l}_{j}\circ\tau^{\otimes j}=\tau\circ l_{j}\, .
\end{equation}
In other words, 
\begin{proposition}\label{proptrunc}
For any Lie $n$-algebra $(\C{V},l_{j})$, 
\begin{equation} \label{original2}
\C{V}_{-n+1} \stackrel{t_{(n-1)}}{\to} \C{V}_{-n+2} \stackrel{t_{(n-2)}}{\to}  \ldots \stackrel{t_{(2)}}{\to}\C{V}_{-1}\stackrel{t_{(1)}}{\to}\C{V}_{0} \, , \end{equation}
the brackets $l_j$ factor through the canonical projection of \eqref{original2} to \eqref{quotient} for any $2 \leq m \leq n \leq \infty$. Thus there are canonically defined multi-linear maps $\widetilde{l}_{j}$ satisfying \eqref{down} which equip \eqref{quotient} with the structure of a Lie $m$-algebra.
\end{proposition}

Let us illustrate this at the simplest example where $n=m=2$, i.e.~the truncation of a given Lie 2-algebra $(\C{V}_\bullet\equiv \C{V}_{-1} \oplus \C{V}_{0}, l_1, l_2, l_3)$, defined on  \eqref{Lie2}, to a simpler one defined on 
\begin{equation}
\widetilde{\C{V}}_{-1}\stackrel{\widetilde{t}}{\to} \C{V}_0 \, ,
\end{equation}
where $\widetilde{\C{V}}_{-1} \equiv \C{V}_{-1}/\mathrm{Ker}( t)$. For homogeneous elements we write $x\in\C{V}_{-1}$, $\widetilde{x}\in\widetilde{\C{V}}_{-1}$, and $y\in\C{V}_{0}$, where therefore $\widetilde{x}$ are equivalence classes in $\C{V}_{-1}$ ($x_1, x_2$ $\in \widetilde{x}$ $\Leftrightarrow$ $\exists x_3 \in \mathrm{Ker}(t)$ s.t. $x_1 = x_2 + x_3$). In the previous notation, $\tau(x) \equiv \widetilde{x}$, while we did not distinguish elements in $ \C{V}_0 $ on which $\tau$ is just the identity. We need to convince ourselves that all the three brackets descend to the quotient. For the 1-bracket $\widetilde{l}_1 \equiv \widetilde{t}$ this is trivial:  by definition $\widetilde{t}(\widetilde{x}) := t(x)$ which is well-defined since in the notation from above we have 
$t(x_1)=t(x_2+x_3)=t(x_2)$. For the 2-bracket we only need to check the case when one element is of degree -1 and one of degree 0 (otherwise the bracket either vanishes identically or remains unmodified by the quotient). For this we need to check that the difference between the result of two elements $x_1, x_2$ $\in \widetilde{x}$ in $[x_i, y]_2$ lies in the kernel of the map $t$. This is, however, a consequence of \eqref{genJac} 
for $m=2$: $t([x_1-x_2,y]_2) \equiv t([x_3,y]_2) = [t(x_3),y]_2 = 0$, since the difference between two such elements is an element $x_3$ in the kernel of $t$. For the 3-bracket, on the other hand, nothing is to be checked since only $[y_1,y_2,y_3]\equiv l_3(y_1,y_2,y_3)$ is non-vanishing and we just need to take the equivalence class of the result so as to define $\widetilde{l}_3(y_1,y_2,y_3)$.

In Appendix \ref{AppendixB} we study also the case of a truncation of a Lie 3-algebra $(\C{V}_\bullet\equiv \C{V}_{-2} \oplus \C{V}_{-1} \oplus \C{V}_{0}, l_1, l_2, l_3, l_4)$  to a Lie 2-algebra  $(\widetilde{\C{V}}_\bullet\equiv \widetilde{\C{V}}_{-1} \oplus \C{V}_{0}, \widetilde{l}_1, \widetilde{l}_2, \widetilde{l}_3)$ in some detail. Here one does not only need to check that the brackets descend correspondingly to the quotient, but also that on the quotient the simplified Jacobi identities are satisfied (i.e.~here in particular that \eqref{genJac} for $m=4$ on $\C{V}_\bullet$ with a non-vanishing $l_4$ indeed descends to an equation with a vanishing 4-bracket $\widetilde{l}_4 \equiv 0$). 
It is here where the first time it also plays a role that the image of $t_{(2)}$ lies in the kernel of the subsequent map $t_{(1)}$.

In fact, we could have also stayed with the image af the preceding map in the truncation. There is a likewise truncation of a Lie $n$-algebra \eqref{original2} to a Lie $m$-algebra
defined over 
\begin{equation}\bigslant{\C{V}_{-m+1}}{\mathrm{Im}(t_{(m)})}\stackrel{\widehat{t}_{(m-1)}}{\to} \C{V}_{-m+2} \stackrel{t_{(m-2)}}{\to} \ldots \stackrel{t_{(2)}}{\to}\C{V}_{-1}\stackrel{t_{(1)}}{\to}\C{V}_{0}\, ,
\label{quotient2}\end{equation}
which, for a non-vanishing cohomology of the complex \eqref{original2} at level (degree) $-m+1$, yields a slightly bigger Lie $m$-algebra than the one of \eqref{original}. Here 
$\widehat{t}_{(m-1)}$ is the map induced by ${t}_{(m-1)}$ on the quotient (which is as before well-defined, since the image of  ${t}_{(m)}$ is lying in the kernel of  ${t}_{(m-1)}$). This corresponding proposition that one can formulate in analogy to Proposition \ref{proptrunc} follows from a lemma similar to \ref{lemmatrunc}, by replacing the image of $Q_{(1)}$ by the kernel of $Q_{(1)}$ inside $\C{W}^{\ast}$ at degree $m$. We also comment on this at the example of $n=3$, $m=2$ in Appendix \ref{AppendixB}. However, for the present application to the tensor hierarchy in six dimensions it is Proposition \ref{proptrunc} that we will need.

It is evident that a general Lie $n$-algebra cannot be truncated \emph{naively} to a Lie $m$-algebra (by just setting to zero the vector spaces of degrees starting from $-m$ and correspondingly with the higher brackets excluded in this way), simply since Jacobi identities do not hold on the nose for the lower brackets, but receive corrections from the higher ones that cannot just be dropped. The main lesson of this section is that this truncation works to all orders, cf.~Proposition \ref{proptrunc}, if in degree $-m+1$ we take the quotient with respect to the kernel of the subsequent chain map $t$. While this is not so straightforward to verify in the original language of multiple brackets, it is rather evident from the dual Q-language, cf.~Lemma \ref{lemmatrunc}. 


\section{Gauge theory description of the truncated tensor hierarchy}
\label{sec:principal}


\subsection{The Q-structure and Bianchi identities}
\label{The Q-structure and Bianchi identities}
We will now apply the Q-formalism defined in Section \ref{section2} to the particular tensor hierarchy described in Section \ref{section3}. In this example, the space-time is six-dimensional, so there will not be any $p$-form field of degree higher than 6. The Q-formalism should  apply to the entire content of the theory, but to keep the presentation transparent, we will limit ourselves to the 1-form $A^{a}$, 2-form $B^{I}$ and the projected 3-form $g^{It}C_{t}$. These are the forms that appear in the Lagrangian formulation of the dynamics, cf.~\cite{Bandos:2013jva}. 
The results of section \ref{section4} show that the Q-formalism can be consistently truncated from the full set of fields to the subset of fields $A^{a}, B^{I}, g^{It}C_{t}$. Moreover, for convenience we will not always spell out the explicit projection $g^{It}C_{t}$ on the 3-forms. Within this section all expressions involving $C_{t}$ etc  are to be understood as contracted with $g^{It}$.

We start from the graded vector space $\C{W}$ whose basis vectors are $\{q_{a},q_{I},q^{t},q^{\lambda}...\}$, with degrees $-1$, $-2$, $-3$, $-4$, $\ldots$, respectively and a degree preserving map $a\colon T[1]\Sigma\longrightarrow\C{W}$ such that $p$-form fields are defined as in (\ref{defa}) according to  $A^{\alpha}=a^{\ast}(q^{\alpha})$. 
In the spirit of the tensor hierarchy, we assume that there is a possibly infinite Q-structure, i.e.~that the vector field 
\begin{equation}
	\begin{aligned}\label{eq:Qagain}
	Q = (-\frac{1}{2}C_{ab}^{c}q^{a}q^{b}	&+t^{c}_{I}q^{I})\frac{\partial}{\partial q^{c}}
		+(-\Gamma^{I}_{aJ}q^{a}q^{J}+\frac{1}{6}H^{I}_{abc}q^{a}q^{b}q^{c}+t^{It}q_{t})\frac{\partial}{\partial q^{I}}\\										&+(A^{s}_{at}q^{a}q_{s}+B_{IJt}q^{I}q^{J}+D_{abIt}q^{a}q^{b}q^{I}+E_{abcdt}q^{a}q^{b}q^{c}q^{d}+t^{\lambda}_{t}q_{\lambda})\frac{\partial}{\partial q_{t}}+\ldots
	\end{aligned}
\end{equation}
has some extension with higher degree coordinates such that it squares to zero. To relate the above coefficients with those present in the tensor hierarchy we  only need to compare the formulas for the field strengths~(\ref{eq:Fa}), (\ref{fieldstrengthsH}), and (\ref{FH}), which yields\footnote{Equivalently, and apparently simpler, we can compare the gauge transformations \eqref{eq:deltaAnewohneF} with \eqref{eq:deltaA}. Note, however, that to obtain the former formulas we needed to know the precise form of the field strengths as well, so as to arrive from \eqref{eq:deltaAnew} to \eqref{eq:deltaAnewohneF}. The situation may change, however, if one finds gauge transformations of the form \eqref{eq:deltaAnewohneF}, i.e.~without any derivatives except for a leading $\dd \lambda^\alpha$, directly!}
\begin{align}
	C^{c}_{ab}		&= -f_{ab}{}^{c}\;,
	& t^{c}_{I}		&= -h^{c}_{I}\;,
	& \Gamma^{I}_{aK}	&= h^{s}_{K}d^{I}_{as}-g^{Is}b_{Ksa}\;,\label{translate}\\[1ex]
	H^{I}_{abc}		&= -d^{I}_{s[a}f^{\phantom{}}_{bc]}{}^{s}\;,
	& t^{It}			&= -g^{It}\;,
	& A^{r}_{at}		&= -f_{at}{}^{r}+d^{J}_{at}h^{r}_{J}\;,\nonumber\\
	B_{JKt}			&= -\frac{1}{2}h^{s}_{(J}b^{\phantom{}}_{K)ts}\;,
	& D_{abJt}			&= -\frac{1}{3}h^{v}_{J}b^{\phantom{}}_{Kt[a}d^{K}_{b]v}\;,
	& E_{abcdt}		&= \frac{1}{12}b^{\phantom{}}_{Kt[a}f^{\phantom{}}_{bc}{}^{s}d^{K}_{d]s}\;.
	\nonumber
\end{align}
Note that in general the naive truncation at any level $p$ of the vector field corresponding to the tensor hierarchy recalled in section 3 does \emph{not} square to zero. So, if in particular, we omit all the terms indicated by $\ldots$ in the above $Q$, it will \emph{not} square to zero, $Q^2 \neq 0$. 
The question of the extendability of a given $Q$-structure to the next order is an interesting mathematical question that we intend to address elsewhere. Here, we either assume that such an extension exists (which may, e.g., give a lower non-vanishing bound on the dimension of the subsequent vector space to be added and spanned by the next level of coordinates), or, equally well, we may directly focus on the truncated version (corresponding to the truncated system of gauge fields $A^{a}$, $B^{I}$, and $g^{It}C_{t}$, which is the main focus of this paper. 

The point is that the Bianchi identities in the tensor hierarchy are assumed to hold only up to field strengths of the next level always; only when truncating appropriately at the last level, one obtains a closed system of Bianchi identities from which one can conclude the nilpotency of the vector field by means of Theorem \ref{thm:QBianchi}. On the level of the graded manifolds, the truncation was subject of the previous section, cf.~in particular Lemma \ref{lemmatrunc} and Proposition \ref{proptrunc}. 
We will now describe the truncated homological vector field explicitly, restricting the coordinates to degree 3. 

For technical reasons we embedded this  $\widetilde{\C{W}}$ into  a degree 3 graded N-manifold $\widehat{\C{W}}$ of the form 
\begin{equation}\label{What}
\widehat{\C{W}}_\bullet : = {\C{W}_{-2}}[1] \stackrel{}{\to}\C{W}_{-2}\stackrel{}{\to}\C{W}_{-1} \, ,
\end{equation}
equipped with the homogeneous coordinates  $\{\widehat{q}^{k},\widehat{q}^{I},\widehat{\widehat{q}}^I\}$ of degrees 1, 2, and 3, respectively. Let $\phi\colon \C{W}\longrightarrow\widehat{\C{W}}$ be the degree preserving map defined implicitly by means of:
\begin{align}
\phi^*(\widehat{q}^{k})&=q^{k}\, ,\nonumber \\
\phi^*(\widehat{q}^{I})&=q^{I}\, , \label{tildetilde}\\
\phi^*(\widehat{\widehat{q}}^{I})&= Q_{(1)}(q^{I})= - g^{It} q_t \, .\nonumber 
\end{align}
The last equation shows that $\phi^{\ast}$ is a surjection from $\C{W}_{-2}[1]^{\ast}$ to $\R{Im}(Q_{(1)}\big|_{\C{W}_{-2}^{\ast}})=\left(\C{W}_{-3} /\R{Ker}(\llbracket.\rrbracket_{1})\right)^{\ast}$, which can be identified with $\left(\R{Im}(\llbracket.\rrbracket_{1})\right)^{\ast}\subset(\C{W}_{-2})^{\ast}$ (except for the shifted grading). Thus the restriction of $\phi^{\ast}$ to $\C{W}_{-2}[1]^{\ast}$ acts as an isomorphism (of graded vector spaces) between $\left(\R{Im}(\llbracket.\rrbracket_{1})[1]\right)^{\ast}$ and $\R{Im}(Q_{(1)}\big|_{\C{W}_{-2}^{\ast}})$, and following section \ref{sec:Linfty}, it implies that
\begin{equation}
\R{Im}(\phi)  \cong \bigslant{\C{W}_{-3}}{\R{Ker}(\llbracket.\rrbracket_{1})} \oplus\C{W}_{-2}\oplus\C{W}_{-1} \, .
\end{equation}
Then on $\widehat{\C{W}}$ we define the vector field:
\begin{eqnarray}
	\widehat{Q} &=& (-\frac{1}{2}C_{ab}^{c}\widehat{q}^{a}\widehat{q}^{b}	 +t^{c}_{I}\widehat{q}^{I})\frac{\partial}{\partial \widehat{q}^{c}}
		+(-\Gamma^{I}_{aJ}\widehat{q}^{a}\widehat{q}^{J}+\frac{1}{6}H^{I}_{abc}\widehat{q}^{a}\widehat{q}^{b}\widehat{q}^{c}+\widehat{\widehat{q}}^{I})\frac{\partial}{\partial \widehat{q}^{I}}\label{Qhatneu}\\
&&{}+
g^{It}(b_{Jta}\widehat{q}^{a}\widehat{\widehat{q}}^{J}-B_{KJt}\widehat{q}^{K}\widehat{q}^{J}-D_{abJt}\widehat{q}^{a}\widehat{q}^{b}\widehat{q}^{J}-E_{abcdt}\widehat{q}^{a}\widehat{q}^{b}\widehat{q}^{c}\widehat{q}^{d})\frac{\partial}{\partial \widehat{\widehat{q}}^{I}}
\;.\nonumber
\end{eqnarray}
Let us remark that $(\widehat{\C{W}},\widehat{Q})$ is in general \emph{not} (yet) the truncated Q-manifold $(\widetilde{\C{W}},\widetilde{Q})$ of Lemma \ref{lemmatrunc}. The technical difficulty here is that in general the linear map on the r.h.s.~of the third equation \eqref{tildetilde} is not a surjection. In other words, the coordinates $\widehat{\widehat{q}}^{I}$ form an overcomplete set of (graded) coordinates on the manifold of our interest.  This (graded) manifold is $\widetilde{\C{W}} := \mathrm{Im} (\phi) \subset \widehat{\C{W}}$. On  $\widetilde{\C{W}}$ the coordinates of degree 3 satisfy the constraints $v_I \widehat{\widehat{q}}^{I} \equiv 0$ for every $v$ such that $v_I g^{It} =0$. The vector field $\widehat{Q}$ restricts to the submanifold  $\widetilde{\C{W}} \subset \widehat{\C{W}}$, however: this becomes obvious from the fact that on $\widehat{\C{W}}$ one has $\widehat{Q}(v_I \widehat{\widehat{q}}^{I}) \equiv 0$. Let us call $\widehat{Q}|_{\widetilde{\C{W}}} := \widetilde{Q}$. Let us denote the surjective map, induced by $\phi$,  from $\C{W}$ to $\widetilde{\C{W}} \subset \widehat{\C{W}}$ by $\widetilde{\tau}$. 
\emph{Assuming} that $(\C{W},Q)$ is a (possibly infinite dimensional) Q-manifold (corresponding to a Lie $\infty$-algebra), the map $\widetilde{\tau} \colon ({\C{W}},{Q}) \to (\widetilde{\C{W}},\widetilde{Q})$ is the Q-morphism of Lemma \ref{lemmatrunc} corresponding to the truncation of the original ``big'' Lie $n$-algebra to the Lie 3-algebra of the interest for the six-dimensional model under consideration. $\widetilde{\tau}$ would then be a Q-morphism, which in the language of $L_\infty$-algebras corresponds to an $L_\infty$-morphism (cf.~the map $\tau$ in the previous section). The situation is summarized schematically in the figure~\ref{figure1}.

\begin{figure}[bt]
   \centering
   \includegraphics[width=7cm]{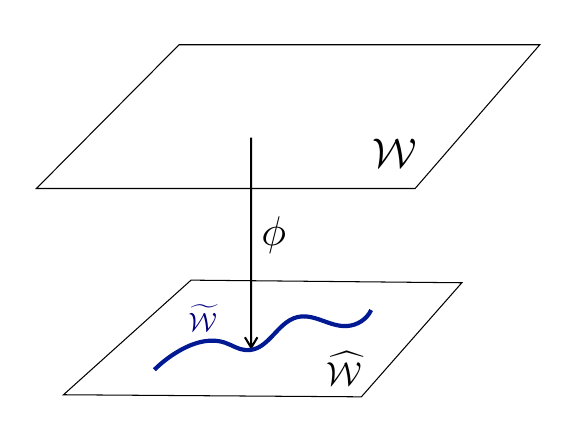}
   \caption{{\small $({\C{W}},{Q})$ is the NQ-manifold corresponding to a ``huge Lie $n$-algebra'' (possibly $n=\infty$) and  $\phi$ the Q-morphism onto   $(\widetilde{\C{W}},\widetilde{Q})$ corresponding to the Lie 3-algebra of relevance for the tensor hierarchy in the six dimensional theory. For technical reasons we embedded $(\widetilde{\C{W}},\widetilde{Q})$ into the graded N-manifold $(\widehat{\C{W}},\widehat{Q})$, which is also of maximal degree 3, but not a Q-manifold; in general  $\widehat{Q}^2\neq 0$.}}
   \label{figure1}
\end{figure}

While we do not address the existence of the ``big'' Lie $\infty$-algebra, as mentioned in the beginning of this section, for the present purposes it is sufficient to prove that 
$(\widetilde{\C{W}},\widetilde{Q})$ is a Q-manifold, yielding a Lie 3-algebra, which subsequently will be made explicit in the following subsection.

One of the main results of the present paper is then summarized in the following theorem, giving the Q-structure governing the six-dimensional tensor hierarchy:

\begin{theoreme}\label{propositionhenning}
The system of algebraic equations (\ref{lol}) ensures that 
the degree one vector field (\ref{Qhatneu})  with the coefficients given in (\ref{translate}) above define a Q-structure $\widetilde{Q}$ on $\widetilde{\C{W}} \equiv \R{Im}(\phi)\subset\widehat{\C{W}}$, i.e.:
\begin{equation}\label{eq:Imphi}
\widetilde{Q}^2 =0
	\;.
\end{equation}
\end{theoreme}
\begin{remarque} One does \emph{not} have $\widehat{Q}^2=0$ as such. One crucially needs to restrict to the image of the map $\phi$. Eq.~(\ref{eq:Imphi}), i.e.~$ (\widehat{Q}\big|_{\R{Im}(\phi)})^{2}=0$,  is equivalent to $\phi^* \circ \widehat{Q}^2=0$.
\end{remarque}

\begin{proof} 
Either by a lengthy, direct calculation that is deferred to Appendix~\ref{AppendixA} \emph{or} by proving the Bianchi identities (\ref{eq:BianchiH}) and applying Theorem \ref{thm:QBianchi}. We remark that using the ``covariant derivatives'' $D$ introduced in section \ref{section3} and identities on their square, the second option is significantly simpler.
\end{proof}

By construction, now the field strengths $\C{F}^\alpha$ are obtained by means of $\C{F}^\alpha = f^* (\bar{\dd} q^\alpha)$. But also the field strengths $\C{H}^\alpha$ can be obtained in a similar way (cf.~Eqs.~(\ref{FH})); one has \begin{equation}
\C{H}^\alpha = f^* (\omega^\alpha)
\end{equation} for some 1-forms on the fibers $T[1]\C{W}$ of the total bundle, where
\begin{equation}\label{eq:omega}
\omega^a = \bar{\dd} q^a \, , \quad \omega^I= \bar{\dd} q^I +d^{I}_{ab}q^{a} \bar{\dd} q^b \, , \quad \omega_t = \bar{\dd} q_t -b_{Jta}q^J\bar{\dd} q^a -\frac{1}{3}b^{\phantom{}}_{Kt[a}d^{K}_{b]c}q^{a} q^{b}\bar{\dd} q^c \, .
\end{equation}
Except for the first of these 1-forms, none of them is exact. For example,
$d^{I}_{ab}$ is \emph{symmetric} in the lower indices while $q^a$ and $q^b$ are anti-commuting; if it were  antisymmetric, on the other hand, 
$\omega^I$ would be exact. If the 1-forms $\omega^\alpha$ were exact, moreover, we could perform a (graded) change of coordinates and present the field strengths in the simpler form of the second equation of  (\ref{eq:fdef}) (in the new coordinates). However, this is not the case.
It is in fact precisely these terms that provide Chern-Simons type of contributions and it is only here where quantities like $d^{I}_{ab}$ appear ``naked'' for the first time (since in Q-structure some of them never do, cf.~Eqs.~(\ref{translate})). The fact that the omegas are not exact is intimately related to the fact that Chern-Simons type of terms cannot be expressed in terms of field strengths. It would be interesting to see, if such terms in the tensor hierarchy correspond to characteristic classes along the lines of \cite{Kotov:2007nr}.

Both ways of proving the Theorem \ref{propositionhenning} indicated are essentially equivalent according to Theorem \ref{thm:QBianchi}. The second path may appear shorter at first sight; this is because some of the computation has been transferred into proving properties of the operator $D$ and the representations of the Leibniz algebra by means of the underlying set of algebraic equations (\ref{lol}). (For example, one can show that $\frac{1}{2}[D,D] \Lambda^\alpha = \C{H}^a X_{a\beta}{}^\alpha \Lambda^\beta$ holds true). However, if $Q$ squaring to zero is once established, one can use this fact to find the Bianchi identities in a most efficient way. All information about Bianchi, also its concrete form, is contained in the formula $\dd \C{F}^\alpha = -\C{F}(Q^{\alpha})$, or, equivalently, in the equation 
(\ref{eq:Qmorph}). 

We illustrate this at the lowest two levels, where there is no difference between $Q$ and $\widetilde{Q}|_{\R{Im}(\phi)}$. We use both versions mentioned above. We start with the 2-form field strength:
\begin{equation}
\dd\C{H}^{a}	\equiv \dd\C{F}^a=-\C{F}(Q^{a})\equiv -\C{F}(
\frac{1}{2}f_{bc}{}^{a}q^{b}q^{c}	-h^{a}_{I}q^{I}))= 
f_{bc}{}^{a}A^{b}\wedge\C{F}^{c}+h^{a}_{I}\C{F}^{I}
\end{equation}
which is in fact the first equation of (\ref{eq:BianchiF}) (note that for the last equality we used the fact that $\C{F}$ satisfies a Leibniz-type property \cite{Bojowald:2004wu}, which follows directly from its definition Def.~\ref{def:field}). We still have to change the $\C{F}^I$ field strength on the rhs for $\C{H}^I$ so as to arrive at the first Bianchi identity in the form of (\ref{eq:BianchiH}). To do so we can employ the equations (\ref{FH}) as well as the definition 
(\ref{eq:D}) together with that of the relevant representation (\ref{gen1}), which then immediately yields the first equation of  (\ref{eq:BianchiH}). We could proceed equally with the higher form degree field strengths, specializing the equations (\ref{eq:BianchiF}) and translating to (\ref{eq:BianchiH}) like in the special case just exposed.

It is illustrative, however, to proceed in a slightly different form here that reveals some of the underlying structure. First we note that the equations (\ref{eq:omega}) correspond to a change of basis of the 1-forms on the graded manifold under consideration, from a holonomic basis $\bar \dd q^\alpha$ to the non-holonomic one $\omega^\alpha\equiv \omega^\alpha_\beta(q)\bar \dd q^\beta$. Let us denote by $\nu^\alpha_\beta$ the inverse matrix to $\omega^\alpha_\beta$,  $\bar \dd q^\alpha= \nu^\alpha_\beta\omega^\beta$. With this we have in full generality
\begin{equation}\label{eq:dHalpha}
\dd \C{H}^\alpha \equiv \dd f^* (\omega^\alpha) =  f^*\left( (\bar\dd + \C{L}_Q)(\omega^\alpha_\beta(q)\bar \dd q^\beta ) \right) \, .
\end{equation}
If $\omega^\alpha$ were an exact basis, the $\bar \dd$ would not act and we could (anti)commute the Lie-derivative along $Q$ through that differential so as to act as just as $Q$ on the primitive. But it is a non-holonomic basis, so that several further terms remain. Before continuing the general discussion, we illustrate this at the 3-form field strength: 
\begin{equation}
\label{eq:exI}
\dd \C{H}^I =  f^*\left( (\bar\dd + \C{L}_Q)(\bar{\dd} q^I +d^{I}_{ab}q^{a} \bar{\dd} q^b  ) \right) = f^*(-\bar{\dd} Q^I +d^{I}_{ab} \bar{\dd}q^{a} \bar{\dd} q^b+
d^{I}_{ab}Q^{a} \bar{\dd} q^b-d^{I}_{ab}q^{a} \bar{\dd} Q^b)
\, .
\end{equation}
Now one needs to read off the components $Q^I$ and $Q^a$ from equations (\ref{eq:Qagain}), (\ref{translate}), change back to the non-holonomic basis by means of $\bar \dd q^a=\omega^a$, $\bar \dd q^I =  \omega^I - d^I_{ab}q^a \omega^b$, and finally use the fact that $f^*$ is a morphism of algebras such that e.g.~$ f^*(d^{I}_{ab} \omega^{a} \omega^b)=d^{I}_{ab} \C{H}^{a} \wedge \C{H}^b$ or $f^*(-d^I_{ab}h^a_J q^J \omega^b)=-d^I_{ab}h^a_J B^J\wedge \C{H}^b$. Like this one ends up with the second equation in (\ref{eq:BianchiH}), the first of the two example terms appearing on the r.h.s.~of the resulting equation, the second term cancelling against the term in $\bar \dd Q^I$ containing $\Gamma$ (cf.~Eq.~(\ref{translate})). Although with this non-holonomic basis the calculation is considerably longer than with the holonomic one, it is important to note that \emph{none} of the algebraic identities (\ref{lol}) need to be used anymore, their content enters already in the first equality of (\ref{eq:exI}), expressing that $f^*$ is a chain map. 

We conclude with the calculation in full generality (valid for \emph{every} basis $\omega^\alpha$ in \emph{every} higher gauge theory):
the rhs of Eq.~(\ref{eq:dHalpha}) gives 
\begin{align}
\begin{split}
&f^*\left( (\bar\dd q^\gamma + Q^\gamma)\frac{\partial \omega^\alpha_\beta(q) }{\partial q^\gamma}\bar \dd q^\beta + (-1)^{|\alpha| - |\beta|} \omega^\alpha_\beta(q)(-\bar \dd Q^\beta) \right)\nonumber\\
&= f^*\left( (\nu^\gamma_\delta(q) \omega^\delta + Q^\gamma(q))\frac{\partial \omega^\alpha_\beta(q) }{\partial q^\gamma}\nu^\beta_\mu(q) \omega^\mu - (-1)^{|\alpha| - |\beta|} \omega^\alpha_\beta(q)\nu^\gamma_\delta(q) \omega^\delta  \frac{\partial Q^\beta(q)}{\partial q^\gamma} \right) \, , 
\end{split}
\end{align}

\noindent where we displayed all the $q$-dependences for clarity. The only effect of $f^*$ is now to replace all $q^\alpha$ by $A^\alpha$ and all $\omega^\alpha$ by $\C{H}^\alpha$, always keeping the given order of objects. So the general Bianchi identities for any choice of basis of 1-forms $\omega^\alpha$ have the form
\begin{equation}
\dd \C{H}^\alpha = (\nu^\gamma_\delta \C{H}^\delta + Q^\gamma)\omega^\alpha_{\beta,\gamma}\nu^\beta_\mu \C{H}^\mu - (-1)^{|\alpha| - |\beta|} \omega^\alpha_\beta\nu^\gamma_\delta \C{H}^\delta  Q^\beta_{,\gamma} \, , 
\end{equation}
where we did no more display the $A$-dependence explicitly and we abbreviated the left-derivatives by a comma w.r.t.~the corresponding coordinate. We leave it as an exercise to the reader to reproduce (\ref{eq:BianchiH}) by specializing the above equation.

\subsection{The structural \texorpdfstring{$L_{3}$}{L3}-algebra}
\label{structural}
We will now apply the discussion of section \ref{sec:Linfty} to the present case.\footnote{
The results of this subsection have some overlap with what
was found independently in~\cite{Palmer:2013pka}. 
}
 We formally have an NQ-manifold $\C{W}$ with homogeneous coordinates $\{q^{r}, q^{I}, q_{t}, q_{\lambda}, \ldots\}$, with an associated Lie $n$-algebra (with $n$ sufficiently big, possibly $n=\infty$):
\begin{equation}
\ldots\stackrel{}{\to}\C{V}_{-3} \stackrel{t_{(3)}}{\to}\C{V}_{-2} \stackrel{t_{(2)}}{\to}\C{V}_{-1}\stackrel{t_{(1)}}{\to}\C{V}_{0} \, . \end{equation}
We have seen in the last section in Theorem (\ref{propositionhenning}) that the NQ-manifold $\C{W}$ can be projected to another NQ-manifold $\widetilde{\C{W}} = \R{Im}(\phi)$ of degree 3,
\begin{equation}\label{Wker}
\widetilde{\C{W}}_\bullet  \cong \bigslant{\C{W}_{-3}}{\R{Ker}(\llbracket.\rrbracket_{1})} \stackrel{\llbracket.\rrbracket_{1}}{\to}\C{W}_{-2}\stackrel{\llbracket.\rrbracket_{1}}{\to}\C{W}_{-1} \, , 
\end{equation}
i.e.~$\widetilde{\C{W}}_{-1} \cong \C{W}_{-1}$,  $\widetilde{\C{W}}_{-2} \cong \C{W}_{-2}$, and  $\widetilde{\C{W}}_{-3} \cong \bigslant{\C{W}_{-3}}{\R{Ker}(\llbracket.\rrbracket_{1})}$, which corresponds to the truncation of the above Lie $n$-algebra to a Lie 3-algebra:
\begin{equation}
\bigslant{\C{V}_{-2}}{\R{Ker}(t_{(2)})} \stackrel{t_{(2)}}{\to}\C{V}_{-1}\stackrel{t_{(1)}}{\to}\C{V}_{0} \, .
\end{equation}
For clarity, we relate \eqref{Wker} more explicitly into the Q-language: $\widetilde{Q}$ is a vector field on $\widetilde{\C{W}}$, it thus acts on functions on $\widetilde{\C{W}}$. The set of homogeneous linear functions on $\widetilde{\C{W}}$ corresponds to the complex dual to \eqref{Wker}, yielding (cf.~Eqs.~\eqref{Qsum}, \eqref{diagram}, and \eqref{W*complex})
\begin{equation}
\mathrm{Im}(\widetilde{Q}_{(1)}) \stackrel{\widetilde{Q}_{(1)}}{\longleftarrow}\C{W}^*_{-2}\stackrel{\widetilde{Q}_{(1)}}{\longleftarrow}\C{W}^*_{-1} \, ,
\end{equation}
where $\widetilde{Q}_{(1)}$ is the part of $\widetilde{Q}$ linear in the coordinates. It was convenient to embed $\widetilde{\C{W}}_\bullet$ into $\widehat{\C{W}}_\bullet$, cf.~Equations \eqref{What} as well as figure \ref{figure1}, even if the corresponding vector field  \eqref{Qhatneu} squared to zero only on the restriction to $\widetilde{\C{W}}$. 

We will now specify coordinates on $\widehat{\C{W}}$ adapted to the truncation determined by $\phi$. It is an elementary fact in linear algebra that the following exact sequence
\begin{equation}\label{basischoice}
0\longrightarrow\R{Im}(\llbracket.\rrbracket_{1}) \longrightarrow\C{W}_{-2}\longrightarrow\bigslant{\C{W}_{-2}}{\R{Im}(\llbracket.\rrbracket_{1})}\longrightarrow 0 \, , 
\end{equation}
admits a splitting, that is: $\C{W}_{-2}$ is isomorphic (as a vector space) to the direct sum $\R{Im}(\llbracket.\rrbracket_{1})\oplus\bigslant{\C{W}_{-2}}{\R{Im}(\llbracket.\rrbracket_{1})}$. Then given any set of degree 2 coordinates $\{q^{A}\}$ on $\R{Im}(\llbracket.\rrbracket_{1})$, and any set of degree 2 coordinates $\{q^{R}\}$ on the quotient $\bigslant{\C{W}_{-2}}{\R{Im}(\llbracket.\rrbracket_{1})}$, the reunion (of their image in $\C{W}_{-2}$) form a set of coordinates on $\C{W}_{-2}$. We will denote this reunion as $\{q^{I}\}=\{q^{A},q^{R}\}$, without further notice to the inclusion maps. The letter $A$ (resp. $R$) from the beginning (resp. the end) of the alphabet has been chosen to emphasize the fact that the set of elements labelled by this letter is a basis of $\R{Im}(\llbracket.\rrbracket_{1})$ (resp. a supplementary subspace isomorphic to the quotient $\bigslant{\C{W}_{-2}}{\R{Im}(\llbracket.\rrbracket_{1})}$). With respect to this basis, the dual map $Q_{(1)}:\R{Im}(\llbracket.\rrbracket_{1})^{\ast}\longrightarrow\R{Im}(Q_{(1)}\big|_{\C{W}_{-2}^{\ast}})$ is an isomorphism (cf the discussion following Equation $(\ref{tildetilde})$). We will denote by $\{\widehat{q}^{I}\}=\{\widehat{q}^{A},\widehat{q}^{R}\}$ the coordinates on $\C{W}_{-2}$ when seen as a subspace of $\widehat{W}$. Their suspension $\widehat{\widehat{q}}^{I}=[1]\widehat{q}^{I}$ gives a set of coordinates on $\C{W}_{-2}[1]$ adapted to the truncation, as explained below. If one picks up a set of degree 1 coordinates $\{\widehat{q}^{k}\}$ on $\C{W}_{-1}$, we ends up with a complete set of coordinates for $\widehat{\C{W}}$, with respect to which the action of the map $\phi$ is easily defined:
\begin{align}
\phi^*(\widehat{q}^{k})&=q^{k}\, ,\nonumber \\
\phi^*(\widehat{q}^{I})&=q^{I}\, , \label{tildetildebis}\\
\phi^*(\widehat{\widehat{q}}^{A})&= Q_{(1)}(q^{A})=- g^{At} q_t \, .\nonumber\\
\phi^*(\widehat{\widehat{q}}^{R})&= 0 \, .\nonumber 
\end{align}
By construction, and following the above discussion, $\phi^{\ast}$ acts as an isomorphism between $\R{Im}(\llbracket.\rrbracket_{1})[1]^{\ast}$ and $\R{Im}(Q_{(1)}\big|_{\C{W}_{-2}^{\ast}})$. That is, $\{\widehat{q}^{k},\widehat{q}^{I},\widehat{\widehat{q}}^{A}\}$ is a set of adapted coordinates to $\widetilde{\C{W}}=\R{Im}(\phi)$, so that their dual $\{\widehat{q}_{k},\widehat{q}_{I},\widehat{\widehat{q}}_{A}\}$ form a basis of $\widetilde{\C{W}}$. The corresponding basis elements in the unshifted truncated algebra $\widetilde{\C{V}}$ will be respectively noted $\{v_{r}\}$, $\{v_{I}\}=\{v_{A},v_{R}\}$ and $w_{A}$, of respective degrees $0,-1,-2$. Note that the above equations implies that with this choice of particular basis in $\C{W}_{-2}$, we have that $g^{Rt}=0$ and $h_{A}^{a}=0$.

In this setup, the homological vector field $\widetilde{Q}$ is given by (it is the same structure as in equation $(\ref{Qhatneu})$, but written directly in adapted coordinates on $\R{Im}(\phi)$):
\begin{eqnarray}
	\widetilde{Q} &=& \Big(-\frac{1}{2}C_{ab}^{c}\widehat{q}^{a}\widehat{q}^{b}	 -h^{c}_{I}\widehat{q}^{I}\Big)\frac{\partial}{\partial \widehat{q}^{c}}
		+\Big(-\Gamma^{A}_{aJ}\widehat{q}^{a}\widehat{q}^{J}+\frac{1}{6}H^{A}_{abc}\widehat{q}^{a}\widehat{q}^{b}\widehat{q}^{c}+\widehat{\widehat{q}}^{A}\Big)\,\frac{\partial}{\partial \widehat{q}^{A}}\nonumber\\
		&&+\Big(-\Gamma^{R}_{aJ}\widehat{q}^{a}\widehat{q}^{J}+\frac{1}{6}H^{R}_{abc}\widehat{q}^{a}\widehat{q}^{b}\widehat{q}^{c}\Big)\,\frac{\partial}{\partial \widehat{q}^{R}}\label{eq:tildeQ}\\
&&{}+
g^{At}\Big(b_{Bta}\widehat{q}^{a}\widehat{\widehat{q}}^{B}-B_{KJt}\widehat{q}^{K}\widehat{q}^{J}-D_{abJt}\widehat{q}^{a}\widehat{q}^{b}\widehat{q}^{J}-E_{abcdt}\widehat{q}^{a}\widehat{q}^{b}\widehat{q}^{c}\widehat{q}^{d}\Big)\,\frac{\partial}{\partial \widehat{\widehat{q}}^{A}}
\;,\nonumber
\end{eqnarray}

We can read the structure of the Lie 3-algebra on this homological vector field, and using Equations (\ref{crochetsbrackets},\ref{derivedbracket}) we obtain:
\begin{align}
	[v_{I}]_{1}				&= -h^{a}_{I}v_{a}\;,\\
	[w_{A}]_{1}				&= v_{A}\;,\nonumber\\
	[v_{a},v_{b}]_{2}				&= f_{ab}{}^{c}v_{c}\;,\nonumber\\
	[v_{a},v_{K}]_{2}			&= -(h^{s}_{K}d^{I}_{as}-g^{Is}b_{Ksa})v_{I}\;,\nonumber\\
	[v_{a},w_{A}]_{2}			&= g^{Bt}b_{Ata}w_{B}\;,\nonumber\\
	[v_{J},v_{K}]_{2}	&= -h^{s}_{(J}b^{\phantom{}}_{K)ts}g^{At}w_{A}\;,\nonumber\\
	[v_{a},v_{b},v_{c}]_{3}			&= -d^{I}_{s[a}f^{\phantom{}}_{bc]}{}^{s}v_{I}\;,\nonumber\\
	[v_{a},v_{b},v_{J}]_{3}	&= \frac{2}{3}h^{v}_{J}b^{\phantom{}}_{Kt[a}d^{K}_{b]v}g^{At}w_{A}\;,\nonumber\\
	[v_{a},v_{b},v_{c},v_{d}]_{4}	&= -2b^{\phantom{}}_{Kt[a}f^{\phantom{}}_{bc}{}^{s}d^{K}_{d]s}g^{At}w_{A}\;.\nonumber
\end{align}
In Section \ref{sec:Linfty}, the definition of the Jacobi identity in the context of $L_{n}$-algebras was given. We can check them on the $L_{3}$-algebra given above. Since the equation $\widetilde{Q}^{2}=0$ corresponds to 12 equations, there will be 12 Jacobi identities which are not identically trivial, and for example we are giving the Jacobi identities for $j=1,2,3$. They can be compared to Equations (\ref{bobbis1}) in the Appendix~\ref{AppendixA} (having in mind that $g^{Rt}=0$ and $h_{A}^{a}=0$):
\begin{align}
	&[[w_{A}]_{1}]_{1}=-h_{A}^{a}v_{a}=0\\
	&[[v_{a},v_{I}]_{2}]_{1}+[[v_{I}]_{1},v_{a}]_{2}=h^{s}_{I}(f_{sa}{}^{b}-h_{J}^{b}d^{J}_{as})v_{b}=0\nonumber\\
	&[[v_{a},w_{A}]_{2}]_{1}+[[w_{A}]_{1},v_{a}]_{2}=g^{Bt}b_{Ata}v_{B}+\Gamma^{B}_{aA}v_{B}=0\nonumber\\
	&[[v_{I}]_{1},v_{J}]_{2}+[[v_{J}]_{1},v_{I}]_{2}-[[v_{I},v_{J}]_{2}]_{1}=2h^{a}_{(I|}\Gamma^{K}_{a|J)}v_{K}+h^{s}_{(I}b^{\phantom{}}_{J)ts}g^{At}v_{A}=0\nonumber\\
	&[[v_{a},v_{b}]_{2},v_{c}]_{2}+[[v_{b},v_{c}]_{2},v_{a}]_{2}-[[v_{a},v_{c}]_{2},v_{b}]_{2}+[[v_{a},v_{b},v_{c}]_{3}]_{1}\nonumber\\
	&\hspace{4.5cm}=\left(-3f_{[ab}{}^{s}f_{c]s}{}^{d}+h_{I}^{d}d^{I}_{s[a}f_{bc]}{}^{s}\right)v_{d}=0\nonumber\\
	&[[v_{a},v_{b}]_{2},v_{I}]_{2}+[[v_{b},v_{I}]_{2},v_{a}]_{2}-[[v_{a},v_{I}]_{2},v_{b}]_{2}+[[v_{a},v_{b},v_{I}]_{3}]_{1}+[[v_{I}]_{1},v_{a},v_{b}]_{3}\nonumber\\
	&\hspace{4.5cm}=\left(-f_{ab}{}^{s}\Gamma^{J}_{sI}-h_{I}^{s}H^{J}_{sab}+2\Gamma^{K}_{[a|I}\Gamma^{J}_{|b]K}-2g^{Jt}D_{abIt}\right)v_{J}=0\nonumber\\
	&[[v_{a},v_{b}]_{2},w_{A}]_{2}-[[v_{a},w_{A}]_{2},v_{b}]_{2}+[[v_{b},w_{A}]_{2},v_{a}]_{2}+[[w_{A}]_{1},v_{a},v_{b}]_{3}\nonumber\\
	&\hspace{4.5cm}=g^{Bt}\left(f_{ab}{}^{d}b_{Atd}+2g^{Cs}b_{At[a|}b_{Ct|b]}-2D_{abAt}\right)w_{B}=0\nonumber\\
	&[[v_{a},v_{J}]_{2},v_{K}]_{2}+[[v_{J},v_{K}]_{2},v_{a}]_{2}+[[v_{a},v_{K}]_{2},v_{J}]_{2}-[[v_{J}]_{1},v_{a},v_{K}]_{3}-[[v_{K}]_{1},v_{a},v_{J}]_{3}\nonumber\\
	&\hspace{4.5cm}=g^{At}\left(-\Gamma^{L}_{a(J}B_{K)L t}+4h^{s}_{(J|}D_{as|K)t}-2B_{JKs}g^{Bs}b_{Bta}\right)w_{A}=0\;.\nonumber
\end{align}

\begin{remarque}
The truncation of this Lie 3-algebra to a Lie 2-algebra is deferred to Appendix \ref{AppendixB}.
\end{remarque}


\subsection{The gauge transformations as inner vertical automorphisms}
\label{subsec:inner}
We now turn to the gauge transformations in the present formalism, providing them with a geometrical interpretation. We already prepared the stage for these considerations in the previous sections so that it mostly amounts to collect the respective formulas. 

However, we will now perform a slight change of paradigm: Before we \emph{assumed} that there is a Q-structure \eqref{eq:Qagain} on a bigger graded manifold $\C{W}$ that then projects down by a Q-morphism $\phi$ to $(\widetilde{\C{W}},\widetilde{Q})$, cf.~Eq.~\eqref{eq:tildeQ} and Figure \ref{figure1}, that is relevant for the tensor hierarchy in six dimensional space-times $\Sigma$. In this article we only \emph{prove} that $(\widetilde{\C{W}},\widetilde{Q})$ is a Q-manifold, cf.~Theorem \ref{propositionhenning}. We also showed that if we truncate the vector field \eqref{eq:Qagain} before the dots and the graded manifold $\C{V}$ at degree smaller than minus three, then the map $\phi$ defined by means of \eqref{tildetilde} preserves the vector field $Q$ (even if possibly $Q^2 \neq 0$ by the truncation!). In other words, we want to use $(\widetilde{\C{W}},\widetilde{Q})$ as the target Q-manifold (since neither $({\C{W}},{Q})$ nor even $(\widehat{\C{W}},\widehat{Q})$ are guaranteed to be Q-manifolds, cf.~figure \ref{figure1}). On the other hand, we do not want to introduce particular coordinates on $\widetilde{\C{W}}$, but, in some sense, use those of ${\C{W}}$ with their  corresponding fields $A^a$, $B^I$, and $C_t$. For the degree 1 and degree 2 coordinates and fields this does not pose any problem. Still, when using $q_t$ and $C_t$ one needs to be careful: Only those combinations of $q_t$ can be used that project down to $\widetilde{\C{W}}$, i.e.~which can be obtained as the pullback of functions on $\widetilde{\C{W}}$. A complete generating set of these degree 3 coordinates is provided by $g^{It}q_t$. Likewise, it is only the combination $g^{It}C_t$ of the 3-form fields that may be considered in this section. 

Similarly we need to be careful with what vector fields on ${\C{W}}$ we use. They must as well project down to $\widetilde{\C{W}}$. In other words, they must not leave the set of ``admissible'' functions on ${\C{W}}$. In this context it will be useful to note that \emph{every} degree $-1$ vector field has this property (simply by degree reasons) and the degree $+1$ vector field $Q$ has this property in view of Lemma \ref{lemmatrunc} (as one may certainly also verify with the explicit expressions). Since gauge transformations  are generated by commutators of such vector fields, we do not encounter any problems using the Q-formalism (with its essential defining property $Q^2=0$---that holds true on the set of functions considered!). We thus can consider recovering gauge transformations of the form 
\eqref{eq:deltaAnewohneF} simply by contracting both sides of its last equation by $g^{It}$ (this also has the effect of annihilating the very last term, $k_{t}^{\lambda} \lambda_\lambda$, which is the only appearance of the ``unwanted'' next-level contribution of the gauge parameters  $\lambda_\lambda$).

We now start with these gauge transformations, also since they are the simplest ones in the present formalism. Using the vector field $Q$ of Equations \eqref{eq:Qagain} together with \eqref{translate}, we observe that the gauge transformations \eqref{eq:deltaAnewohneF} are precisely of the form \eqref{eq:deltaA}. It thus remains to specialise Equation \eqref{offen} to the present context, taking care of the contraction issue of the 3-form fields mentioned above, however. Using \eqref{translate}, this yields 
\begin{eqnarray}
[\delta_{\lambda},\delta_{\lambda'}]\, A^a &=& \delta_{\widehat{\lambda}}\,A^a \nonumber \, , \\ {}
[\delta_{\lambda},\delta_{\lambda'}] \,B^I &=& \delta_{\widehat{\lambda}}\,B^I -d^{I}_{p[a}f_{bc]}{}^{p}\lambda^{a}\lambda'^{b}\C{F}^{c}\label{doffen} \, ,\\ {} [\delta_{\lambda},\delta_{\lambda'}]\,
\left(g^{It} C_t \right)&=& \delta_{\widehat{\lambda}}\, \left(g^{It}C_t\right) + 2g^{It}b_{Jt[a}f_{bc}{}^{p}d_{d]p}^{J}\lambda^{a}\lambda'^{b}\C{F}^{c}A^{d}\nonumber \\
&&\hspace{2cm}-\frac{2}{3}g^{It}h_{J}^{p}b_{Kt[a}d^{K}_{b]p}(\lambda^{a} \lambda'^{b}  \C{F}^{J}-\lambda^{a} \C{F}^{b}  \lambda'^{J}+\lambda'^{a}  \C{F}^{b}\lambda^{J})\, ,\nonumber
\end{eqnarray}
where the new parameters $\widehat{\lambda}$ are given by (cf.~Equation \eqref{hutepsilon}):
\begin{eqnarray}
\widehat{\lambda}^a &=&   -f_{bc}{}^{a}\lambda^{b}\lambda'^{c} \, , \nonumber \\
\widehat{\lambda}^I &=& (h^{s}_{K}d^{I}_{as}-g^{Is}b_{Ksa})(\lambda^{a}\lambda'^{K}-\lambda'^{a}\lambda^{K})+d^{I}_{s[a}f^{\phantom{}}_{bc]}{}^{s}\lambda^{a}\lambda'^{b}A^{c} \, , \label{something} \\
\widehat{\lambda}_t &=&
		-(-f_{at}{}^{s}+d^{J}_{at}h^{s}_{J})(\lambda^{a}\lambda'_{s}-\lambda'^{a}\lambda_{s})-h^{s}_{(I}b^{\phantom{}}_{J)ts}\lambda^{I} \lambda'^{J}
	-b^{\phantom{}}_{Kt[a}f^{\phantom{}}_{bc}{}^{s}d^{K}_{d]s}\lambda^{a}\lambda'^{b} A^{c}  A^{d}\;\nonumber\\
								&&\hspace{2cm}-\frac{2}{3}h^{v}_{I}b^{\phantom{}}_{Kt[a}d^{K}_{b]v}(\lambda^{a} A^{b}  \lambda'^{I}-\lambda'^{a}  A^{b}  \lambda^{I}-\lambda^{a} \lambda'^{b}  B^{I})   \, . \nonumber
\end{eqnarray}
Note that certainly $g^{It}$ does not change under gauge transformations since these are mere constants. However, this projection is \emph{needed} to derive the last of the  three equations of \eqref{offen}  by means of the derived bracket constructions outlined in section \ref{subsection22} since this construction relies essentially on $Q^2=0$. On the other hand in the last equation of \eqref{something} such a projection is \emph{not} needed, as the corresponding vector field is of negative degree, as argued already above.

We observe that the gauge symmetries \eqref{eq:deltaAnewohneF} do not close. Still, by construction, the terms obstructing closure lie inside the ideal $\C{I}$ generated by the field strengths (since each individual gauge transformation \eqref{eq:deltaAnewohneF} or \eqref{eq:deltaA} of the commutator was required to have this property---at least when respecting the projection by $g^{It}$ in the second case). Still, the commutators \eqref{doffen} were derived from a picture in which they correspond to infinitesimal inner automorphisms of a Q-bundle---which, as such, certainly form a closed algebra, as illustrated by figure~\ref{Figure:2} (cf.~also \cite{Bojowald:2004wu,Kotov:2010wr,Gruetzmann-Strobl} for further details).
The transformations $\delta A^\alpha$ in this picture correspond to the restriction of the vector field $V_\epsilon$ to the section $a$ that corresponds to the given $A^\alpha$. $V_\epsilon$ is an (infinitesimal) automorphism since it is ad$_{Q_\C{N}}$-closed, it is inner, since it is even ad$_{Q_\C{N}}$-exact. As such, also when restricted to the vertical ones, the set $\{V_\epsilon\}$  certainly forms a ((closed)) Lie algebra with respect to the commutator as the bracket.

\begin{figure}[bt]
   \centering
   \includegraphics[width=8cm]{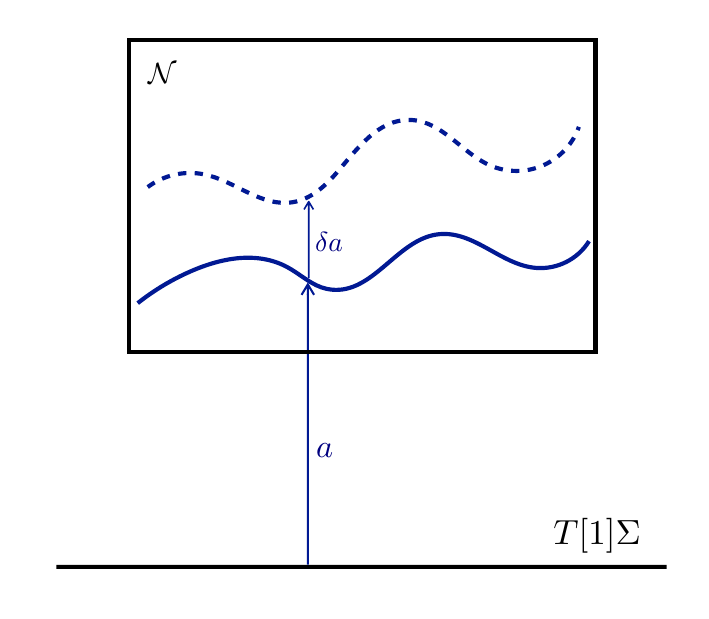}
   \caption{{\small Schematic picture of the Q-bundle $(\C{N},Q_{\C{N}})\to (T[1]\Sigma,\mathrm{d})$ with typical  fiber  $(\widetilde{\C{W}},\widetilde{Q})$. The collection of gauge fields $(A^\alpha)=(A^a,B^I,C_\cdot)$ is encoded into a section $a \colon T[1]\Sigma \to \C{N}$. In general, this is a section in the category of graded manifolds only; it has vanishing field strengths $F^\alpha$ precisely if it is a section in the category of Q-manifolds also. The joint infinitesimal gauge variations  $\delta A^\alpha$ correspond to an infinitesimal variation $\delta a$ of the given section. This variation is generated by a vertical vector field  $V_\epsilon\equiv [Q_{\C{N}},\epsilon]$, the restriction of which to the image of this section $a$ gives 
$\delta a$.  $V_\epsilon$ generates an inner vertical automorphism of the Q-bundle $\C{N}$.}
  \label{Figure:2}}
\end{figure}

For any fixed choice of $a$ or $A^\alpha$, the gauge transformations \eqref{eq:deltaAnewohneF} are in one-to-one correspondence with infinitesimal vertical inner automorphisms $V_\epsilon$ with an $\epsilon$ that is ``\emph{constant}'' along the fibers. Here ``constant'' is to be understood in a particular coordinate system along the fibers or, 
equivalently, by means of an auxiliary connection chosen on the fibers. The constant value chosen at each fiber over $(\sigma, \mathrm{d}\sigma) \in T[1]\Sigma$ is determined by the value of $\lambda$ at this point by requiring 
that $V_\epsilon$ restricted to the point $a(\sigma, \mathrm{d}\sigma)$ generates precisely the given infinitesimal transformation 
$\delta_{\lambda} A^\alpha$ (by means of $a^*(V_\epsilon(q^\alpha))$). Now, the 
main point is that the set of such ``constant automorphisms'' forms a subspace 
of the Lie algebra of all (infinitesimal vertical inner) automorphisms, but not a Lie subalgebra.
  
The fact that an eventual symmetry algebra is what is called ``open'' restrains the possible functionals $S$ for which they can be symmetries. Clearly, the set of \emph{all}  
(also appropriately restricted, like ``local'') symmetries of a functional $S$ form a (generically infinite dimensional) Lie group and infinitesimally its corresponding Lie algebra. Why one still considers ``open algebras of symmetries'' is related to the fact that \emph{any} (local) functional has so-called ``trival'' gauge symmetries, which are symmetries of the functional $S$ leaving any solution to its Euler-Lagrange equation invariant (cf.~also \cite{Henneaux:1992ig} for further details). Talking about a gauge theory implies implicitly that $S$ has non-trivial gauge symmetries. While the set of \emph{all} infinitesimal gauge symmetries of $S$ again forms a Lie algebra, a subset of generators does not necessarily: The trivial ones form an ideal that can be factored out, but the non-trivial ones do not form a Lie subalgebra in general. Correspondingly, when choosing representatives of the generators in each class, their Lie brackets do not necessarily close and, even worse, if the exact sequence of Lie groups or algebras
\begin{align}
\text{``trivial}\hspace{1.6cm}&&\text{``all}\hspace{1.9cm}&&\text{``non-trivial}\hspace{1cm} \nonumber\\
 \text{gauge} \hspace{1.6cm}&\longrightarrow&   \text{gauge} \hspace{1.6cm} &\longrightarrow&\text{gauge symmetries}\hspace{0.3cm}\\
\text{symmetries''} \hspace{1cm}&&  \text{symmetries''}\hspace{1cm}&& \text{modulo trivial ones''}\nonumber
\end{align}
does not split, there even does not exist \emph{any choice} of generators such that the symmetries close.\footnote{This then is the situation where for a quantization standard BRST-methods are no more sufficient, but one needs to rely on the Batalin-Vilkovisky formalism \cite{Batalin:1981jr,Henneaux:1992ig}.} 

So, if one wants to use the transformations \eqref{eq:deltaAnewohneF} as (infinitesimal) gauge symmetries of a functional $S$, the choice of such a functional $S$ is necessarily restricted by the terms on the right-hand-side of the commutator algebra \eqref{doffen}. In particular, any such a functional $S$ \emph{must} have
\begin{eqnarray}
d^{I}_{p[a}f_{bc]}{}^{p}\C{F}^{c} &=& 0 \, , \nonumber \\
g^{It}b_{Jt[a}f_{bc}{}^{p}d_{d]p}^{J}\C{F}^{c}A^{d} &=& 0 \, , \label{contractedF}\\
g^{It}h_{J}^{p}b_{Kt[a}d^{K}_{b]p}\C{F}^{b} &=& 0 \, , \nonumber \\
g^{It}h_{J}^{p}b_{Kt[a}d^{K}_{b]p} \C{F}^{J}&=& 0 \,\nonumber 
\end{eqnarray}
among its field equations (where some of the first three equations may be redundant). Sufficient would be certainly that both 2-form and 3-form field strengths vanish on-shell\footnote{In the literature, often equalities that hold true only on behalf of the Euler-Lagrange equations are denoted by $\approx$.}
\begin{eqnarray}
\C{F}^a  &\approx & 0 \, , \label{F}\\
  \C{F}^I &\approx & 0 \, . \nonumber
\end{eqnarray}

There may be occasions, where restrictions of the form \eqref{contractedF} or \eqref{F} do not pose any problem and may be even welcome in some sense. If, for some reason, however, these constraints seem unphysical in a particular context, one needs to find other transformations, which impose weaker constraints or even no constraints---as in the case of a closed algebra.
Actually, in the context of the six-dimensional tensor hierarchy studied in the present paper, such a set of generators forming a closed algebra \emph{is} known: As we will find below, the transformations \eqref{eq:transf} do form a Lie algebra---and thus also the transformations \eqref{eq:deltaAnew}, since they only differ by a change of gauge parameters from the former ones\footnote{Consider a foliation on a manifold $M$ generated by vector fields $\rho_a$. The generators are thus involutive: $[\rho_a,\rho_b]=C^{c}_{ab}\rho_c$, or,
 for any constants $\epsilon^a \in \mathbb{R}$,  $[\epsilon^a \rho_a,\epsilon^b\rho_b]=\left(C^{c}_{ab}\epsilon^a\epsilon^b\right)\rho_c$. If one generalizes the last equality by permitting  $\epsilon^a$ to be functions on $M$, 
 the algebra changes by derivatives of those functions, while it evidently still will close. Essentially this option corresponds to a change of generators for a fixed foliation.\label{footnote}}
In the following we want to verify this on the one hand, but, on the other hand, we also want to apply the formalism of section \ref{section2} for the following two reasons: First, a geometric interpretation such as in figure \ref{Figure:2} above is conceptually interesting and illuminating. Second, the use of the derived bracket construction may yield pronounced technical advantages for the explicit computation.\footnote{In the general setting this is certainly the case since the heavy use of identities such as those of \eqref{lol} are replaced by the simple condition $Q^2=0$. In the concrete case here, one may retrieve some of those properties encoded into these equations by means of properties satisfied by the ``covariant derivatives'' \eqref{eq:D}, the use of which can likewise simplify the computation that otherwise would need the repeated use of \eqref{lol}.}

The symmetries \eqref{eq:deltaAnew} differ from the transformations \eqref{eq:deltaAnewohneF} only by adding some (particularly chosen) $F^\alpha$-dependent contributions. As such they also stay inside the ideal $\C{I}$  by construction (here we implicitly use the fact that $\dd  \C{I} \subset \C{I}$). 

We now have to address the question, if the gauge symmetries \eqref{eq:deltaAnew}  can be written as inner automorphisms as in figure \ref{Figure:2} above. Here two cases need to be distinguished: First, if the parameters $\widetilde{\Lambda}^\alpha$ are assumed to be field-independent. Second, if they are permitted to be field-dependent. This second option is not just academic or artificial since we saw above that even if one starts with field-independent parameters, the commutator of two such transformations may very well be field-dependent, cf., e.g., Equation (\ref{something}).

To clarify the situation properly, we first recall that for the degree minus one vertical vector field $\epsilon \equiv \epsilon^\alpha(\sigma, \dd \sigma, q) \partial_\alpha$ and a section $a \colon T[1]\Sigma \to \C{N}$ one has (in similarity with the consideration yielding the second equality in \eqref{ll'})
\begin{equation}
a^*\left( [Q_{\C{N}},\epsilon] q^\alpha \right) = \dd \left(a^* \epsilon^\alpha \right)+ a^{\ast}({\epsilon}^\beta) a^*( \partial_\beta(Q^{\alpha})) -  F^\beta\, a^*(\partial_\beta{\epsilon}^{\alpha}) \, . \label{astar}
\end{equation}
Let us compare this to gauge transformations of the form 
\begin{equation}
	{\delta} A^{\alpha}=\dd \widetilde{\Lambda}^{\alpha}
+\widetilde{\Lambda}^{\beta}{V}^{\alpha}_{\beta} + F^\beta N_{\beta\gamma}^\alpha \widetilde{\Lambda}^\gamma , \label{deltaA}
\end{equation}
where the coefficients $V^{\alpha}_{\beta}$ and $N_{\beta\gamma}^\alpha$ are polynomials in $A^\alpha$ (of the corresponding degrees so as to make the above equation homogenous in form degree). Both, the symmetries \eqref{eq:deltaAnewohneF} as well as the symmetries \eqref{eq:deltaAnew} are of the above form, the former ones with vanishing coefficients $N_{\beta\gamma}^\alpha$, the latter ones with non-vanishing coefficients. 

Let us now first assume that the gauge parameters $\widetilde{\Lambda}$ are field-independent. The comparison of the two equations \eqref{astar} and \eqref{deltaA} show that $\epsilon^\alpha$ has to be q-independent.\footnote{We do not consider the exotic option of permitting $\epsilon^\alpha$ to depend on both, $A^\alpha$ and $q^\alpha$, so that under $a^*$ these two contributions could cancel against one another.} This, however, implies that the $F^\beta$-contributions vanish in \eqref{astar}, while we know them to be present in the transformations \eqref{eq:deltaAnew} under consideration. And this evidently is a general result: Whenever the parameters in gauge transformations \eqref{deltaA} are not permitted to depend on the fields and, at the same time, not all the coefficients $ N_{\beta\gamma}^\alpha$ vanish identically, the symmetries cannot be cast into the form \eqref{astar}.

The situation is more intricate, however, if the parameters  $\widetilde{\Lambda}^\alpha$ are permitted to depend also on the fieds, i.e.~there are field-independent parameters $\mu^\beta\in \Omega^\bullet(\Sigma)$ such that
\begin{equation}
\widetilde{\Lambda}^\alpha = K^\alpha_\delta \mu^\delta \, , \label{Lambda}
\end{equation}
where $K^\alpha_\beta$ are polynomials in $A^\gamma$ of appropriate degrees. Denote by $k^\alpha_\beta$ the corresponding function depending on $q$, i.e.~$K^\alpha_\beta = a^*(k^\alpha_\beta)$. Then comparison of the first terms in \eqref{astar} and \eqref{deltaA} show that necessarily $\epsilon^\alpha = k^\alpha_\beta(q) \mu^\beta$. But then comparison of the \emph{last} term of each of the two equations implies that $\partial_\beta k^\alpha_\delta = -n^\alpha_{\beta \gamma} k^\gamma_\delta$, where $N^\alpha_{\beta \gamma}=a^*(n^\alpha_{\beta \gamma})$, or, equivalently, that the coefficient matrix $K$ in \eqref{Lambda} satisfies the likewise differential equation
\begin{equation}
\partial_\beta K^\alpha_\delta +N^\alpha_{\beta \gamma} K^\gamma_\delta=0\;, \label{covconst}
\end{equation}
for the coefficients $N$ appearing in \eqref{deltaA}. We note that in this equation the index delta does not play an important role. Considering it to be fixed (or having it ``disappear'' by permitting some $K^\alpha$ to depend (linearly) on these parameters), we observe that
the equation has the form of a vector $K^\alpha$ being ``covariantly constant'' for the connexion coefficients $N^\alpha_{\beta \gamma}$. Since, moreover, as a consequence of \eqref{Lambda} we need that $K^\alpha$ is non-vanishing for \emph{all} values of alpha, we see that this ``connexion'' $N$ needs to be flat, i.e.~that it needs to satisfy the consistency condition of a (graded) Riemann tensor $R_{\eta\beta\gamma}{}^{\alpha}$ of this connexion to vanish:\footnote{One just differentiates Equation \eqref{Lambda} one more time, uses the equation to eliminate first derivatives and then (graded) anti-symmetrizes in the index $\beta$ and the new one from the derivative. The resulting equation then takes the form $R_{\eta\beta\gamma}{}^{\alpha} K^{\gamma}_{\epsilon} =0$ with the expression of $R$ given in Equation \eqref{zerocurvature}. Since we do not want to restrict the variations of $\widetilde{\Lambda}^\alpha$, moreover, we can drop the contraction with $K^{\alpha}_{\epsilon}$ and obtain the condition $R_{\eta\beta\gamma}{}^{\alpha}=0$ from this.}
\begin{equation}
\partial_{\eta}(N^{\alpha}_{\beta\gamma})+(-1)^{(|\alpha|-|\gamma|)|\beta|}N^{\alpha}_{\eta\delta}N^{\delta}_{\beta\gamma}-(-1)^{|\eta||\beta|}\left(\partial_{\beta}(N^{\alpha}_{\eta\gamma})+(-1)^{(|\alpha|-|\gamma|)|\eta|}N^{\alpha}_{\beta\delta}N^{\delta}_{\eta\gamma}\right)=0\;.\label{zerocurvature}
\end{equation}

The vanishing of this tensor is a \emph{necessary} condition for casting the symmetries \eqref{deltaA} into the form \eqref{astar}. It is remarkable that it does \emph{not} yet depend on the choice of $Q$ itself. But this is also why the condition is not sufficient, at least in the case that we assume $Q$ to be known already. However, if not yet known, it may be a way to read off $Q$ from the symmetries without prior investigation of other equations. 

Let us now check the condition \eqref{zerocurvature} for the transformations \eqref{eq:deltaAnew}. While the $N$ of the second line poses no problem, the last line  gives the curvature 
\begin{equation}
R_{abc\ t}=2b_{Jtc}d^{J}_{ab}
\;,
\end{equation}
which in general is non-vanishing. This implies that we cannot use the formalism underlying \eqref{astar} for those symmetries. We thus look at the ``lifted'' picture explained at the end of section~\ref{section2}. 

We thus apply Theorem \ref{thm:gauge} to the symmetries \eqref{eq:deltaAnew}. In particular, we can read off from \eqref{eq:tildeepsilon} the degree minus vector field generating these symmetries by means of the lifted inner automorphism \eqref{eq:tildedelta}. One easily finds: 
\begin{align}\label{tildeepsilon}
	\widetilde{\epsilon}	&=\widetilde{\Lambda}^{a}\frac{\partial}{\partial q^{a}} + \widetilde{\Lambda}^{I}\frac{\partial}{\partial q^{I}} + \widetilde{\Lambda}_{t}\frac{\partial}{\partial q_{t}}-d^{I}_{ab}\widetilde{\Lambda}^{a}\bar\dd q^{b}\frac{\partial}{\partial \bar\dd q^{I}}\\
	&\hspace{1cm}+\Big[b_{Jta}\widetilde{\Lambda}^{J}\bar\dd q^{a}+\frac{2}{3}b_{Jt[a}d^{J}_{b]c}\widetilde{\Lambda}^{a}q^{b}\bar\dd q^{c}\Big]\frac{\partial}{\partial \bar\dd q_{t}}\; ,\nonumber
\end{align}
where the first three terms are nothing but $\C{L}_{\epsilon}$ (note that there are no ``tangent derivatives'' in this case since $w$ is the Kronnecker delta and $\widetilde \Lambda^\alpha$ is assumed to be field-independent for now) and the contributions from $\bar{v}^\alpha_\beta - {v}^\alpha_\beta$ is linear in $\bar{\dd}q^\gamma$ and, up to a sign, thus agrees with the ``connexion coefficients'' $n^\alpha_{\beta \gamma}\bar{\dd}q^\gamma$ introduced above. 

There is one more consistency consideration to perform at the stage: We argued above, in the non-lifted picture, that the vector fields of degree $-1$ cannot lead out of the admissible functions---recall that for technical reasons we stick to the coordinates on $\C{W}$, which, in the lifted picture, has to be replaced by $T[1]\C{W}$. Admissible coordinates on $T[1]\C{W}$ are $q^a$, $q^I$, $g^{It}q_t$ as well as  $\bar{\dd}q^a$, $\bar{\dd}q^I$, $g^{It}\bar{\dd}q_t$ of degree 1, 2, 3, and 2, 3, 4, respectively. Here we now need to check that the derivative with respect to $\bar{\dd}q_t$ does not have a ``naked'' $q_s$---but if $q_s$ appears, which now for degree reasons \emph{is} permitted, it has to appear in a contracted way, $g^{Js}q_s$. Indeed, the $\bar{\dd}q_t$-derivative part of $\widetilde{\epsilon}$ in \eqref{tildeepsilon} does not contain \emph{any} $q_s$ variable. We are thus perfectly permitted to use the formalism also before the projection. 
In order to obtain the commutator of two gauge transformations \eqref{eq:deltaAnew}, we may now simply apply the derived bracket construction for them, in analogy to the discussion following \eqref{deltaAalpha}. One finds\footnote{We anticipate the notation $(\bar{\C{N}}, Q_{\bar{\C{N}}})$ for the Q-bundle with the new ``tangent fiber'', cf.~Figure \ref{Figure:3} below.}
\begin{align} \label{commtildeps}
	[\widetilde{\epsilon},\widetilde{\epsilon}']_{Q_{\bar{\C{N}}}}	&= \widehat{\widetilde{\Lambda}}^{a}\frac{\partial}{\partial q^{a}} + \widehat{\widetilde{\Lambda}}^{I}\frac{\partial}{\partial q^{I}} + \widehat{\widetilde{\Lambda}}_{t}\frac{\partial}{\partial q_{t}}-d^{I}_{ab}\widehat{\widetilde{\Lambda}}^{a}\bar\dd q^{b}\frac{\partial}{\partial \bar\dd q^{I}}\\
	&\hspace{1cm}+\Big[b_{Jts}\widehat{\widetilde{\Lambda}}^{J}\bar\dd q^{s}+\frac{2}{3}b_{Jt[u}d^{J}_{v]s}\widehat{\widetilde{\Lambda}}^{u}q^{v}\bar\dd q^{s}+k_{t}^{\alpha}(\ldots)_{\alpha}\Big]\frac{\partial}{\partial \bar\dd q_{t}}\; ,\nonumber
\end{align}
where  
\begin{align}\label{closedalgebra}
	\widehat{\widetilde{\Lambda}}^{a}=[\widetilde{\epsilon},\widetilde{\epsilon}']_{Q_{\bar{\C{N}}}}^{a}	&=	C_{bc}^{a}\widetilde{\Lambda}^{b}\widetilde{\Lambda}'^{c}\nonumber\\
	\widehat{\widetilde{\Lambda}}^{I}=[\widetilde{\epsilon},\widetilde{\epsilon}']_{Q_{\bar{\C{N}}}}^{I}	&=	\Gamma^{I}_{aK}(\widetilde{\Lambda}^{a}\widetilde{\Lambda}'^{K}-\widetilde{\Lambda}'^{a}\widetilde{\Lambda}^{K})-H^{I}_{abc}\widetilde{\Lambda}^{a}\widetilde{\Lambda}'^{b}q^{c}\;,\\
	\widehat{\widetilde{\Lambda}}_{t}=[\widetilde{\epsilon},\widetilde{\epsilon}']_{Q_{\bar{\C{N}}}\ t}	&=	-A_{at}^{s}(\widetilde{\Lambda}^{a}\widetilde{\Lambda}'_{s}-\widetilde{\Lambda}'^{a}\widetilde{\Lambda}_{s})+2B_{IJt}\widetilde{\Lambda}^{I}\widetilde{\Lambda}'^{J}-12E_{abcdt}\widetilde{\Lambda}^{a}\widetilde{\Lambda}^{b} q^{c} q^{d}\;,\nonumber\\
								&\hspace{1cm}-2D_{abIt}(\widetilde{\Lambda}^{a} \widetilde{\Lambda}'^{b} q^{I}-\widetilde{\Lambda}^{a} q^{b} \widetilde{\Lambda}'^{I}+\widetilde{\Lambda}'^{a} q^{b} \widetilde{\Lambda}^{I})-\frac{2}{3}b^{\phantom{}}_{Jt[b}d^{J}_{c]a}\widetilde{\Lambda}^{b}\widetilde{\Lambda}'^{c}\bar\dd q^{a}\;.\nonumber
\end{align}
We did not specify the terms proportional to $k_t^\alpha$ in the last line of \eqref{commtildeps} as they will drop out in what follows.
Note also that this bracket is already antisymmetric under the exchange of  $\widetilde{\epsilon}$ and $\widetilde{\epsilon}'$. 

Now, although the new parameters depend on fiber coordinates, here the algebra does close. Indeed one obtains
\begin{align}
[\delta_{\widetilde{\Lambda}},\delta_{\widetilde{\Lambda}'}]A^a 
&=\delta_{\widehat{\widetilde{\Lambda}}} A^a \nonumber \\
[\delta_{\widetilde{\Lambda}},\delta_{\widetilde{\Lambda}'}]B^I
&=\delta_{\widehat{\widetilde{\Lambda}}} B^I \label{commtildeLambda} \\
[\delta_{\widetilde{\Lambda}},\delta_{\widetilde{\Lambda}'}]\left( g^{It} C_t\right) 
&=\delta_{\widehat{\widetilde{\Lambda}}} \left( g^{It} C_t\right) \, .\nonumber 
\end{align}
The $k$-contributions do not contribute in the last line because  Equation \eqref{eq:gk} implies that \\ $[Q_{\bar{\C{N}}},k_{t}^{\alpha}(\ldots)_{\alpha}\frac{\partial}{\partial \bar\dd q_{t}}] \left( g^{Is} q_s \right) = g^{It} k_t^\alpha 
(\dots )_\alpha \equiv 0$. 

The reason for the closure is two-fold: First we see that the derived bracket \eqref{commtildeps} of two vector fields \eqref{tildeepsilon} has a likewise form, albeit with the new parameters $\widehat{\widetilde{\Lambda}}$ being fiber-dependent. 
And that this does not play a role is the second reason: In the lifted picture the analogous formula to \eqref{astar} reads as follows
\begin{equation}
f^*\left( [Q_{\bar{\C{N}}},\widetilde{\epsilon}] q^\alpha \right) = \dd \left(f^* \widetilde{\epsilon}^\alpha \right)+ f^{\ast}({\widetilde{\epsilon}}(Q^{\alpha}+\bar{\dd} q^\alpha)) \, .   \label{fstar}
\end{equation}
In fact this is the \emph{same} formula as \eqref{astar}, but with $f$ replacing $a$ (and $\C{L}_{Q_{\C{N}}} + \bar{\dd}$ replacing $Q$ when acting on $q^\alpha$). The last term, proportional to the derivative of the parameters in \eqref{astar} is absent in \eqref{fstar} since $f$ is by construction a Q-morphism, cf.~Equation \eqref{eq:Qmorph}, so that the ``field strength of the field strength'' vanishes identically. 

We now want to also provide a geometrical interpretation of the symmetries generated by \eqref{tildeepsilon}. In fact, also they are vertical inner automorphisms, just of another bundle. We needed to replace the original fiber $\widetilde{\C{W}}$, corresponding to the Lie 3-algebra, by its (shifted) tangent space $T[1]\widetilde{\C{W}}$ (that this can be done also in a canonical way for non-trivial bundles is explained for instance in \cite{Gruetzmann-Strobl}). In this bundle the symmetries \eqref{eq:deltaAnew} are generated again as vertical inner automorphisms following Equation \eqref{fstar}. The difference to the previous situation is illustrated also in Figure \ref{Figure:3} (to be compared with the Figure \ref{Figure:2} before).
\begin{figure}[bt]
   \centering
   \includegraphics[width=11.6cm]{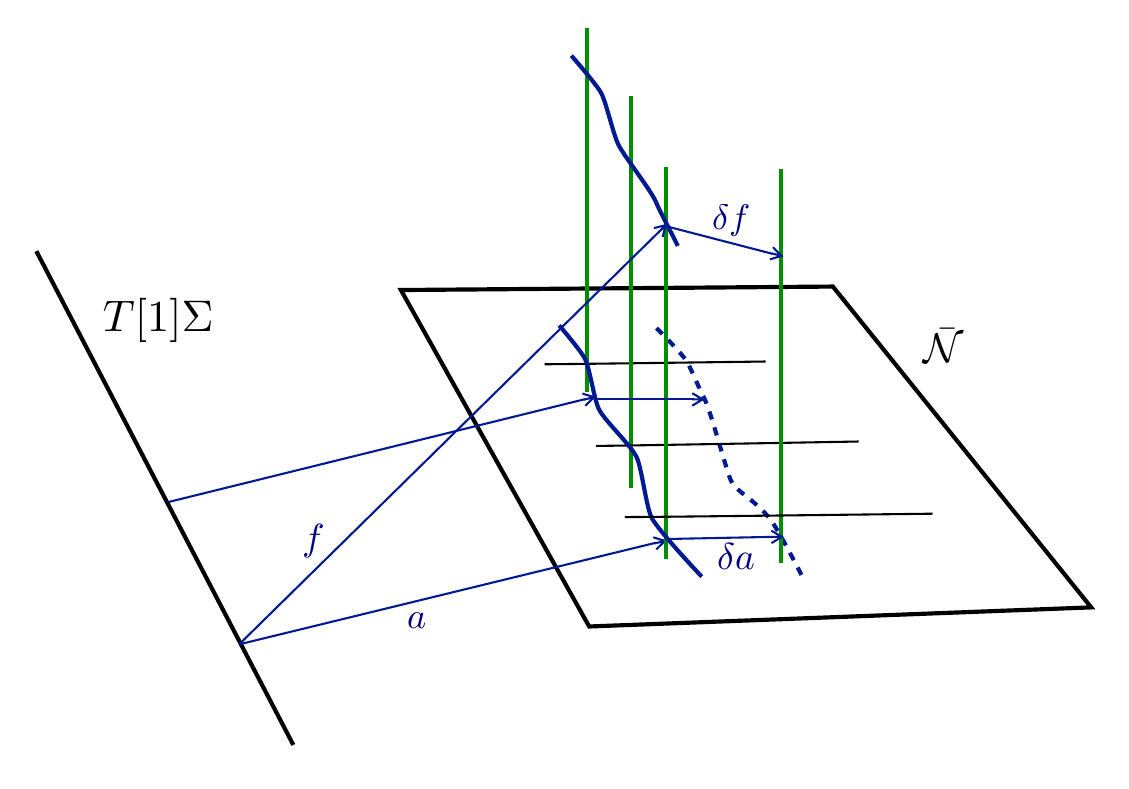}
   \caption{{\small  Schematic picture of the Q-bundle $(\bar{\C{N}},Q_{\bar{\C{N}}})\to (T[1]\Sigma,\mathrm{d})$ with typical  fiber  $(T[1]\widetilde{\C{W}},\C{L}_{\widetilde{Q}} +\bar{\mathrm{d}})$, which can be constructed canonically from the bundle $\C{N}$ of Figure \ref{Figure:2}.  
The vertical green lines are the tangent directions to $\widetilde{\C{W}}$, the bundle $\C{N}$ (together with the section $a$) being recovered by projection onto the horizontal plane in the figure.  $f$ now is the unique, canonically defined lift of $a$ which is a section of the extended bundle $(\bar{\C{N}},Q_{\bar{\C{N}}})$ in the category of Q-manifolds. The structural Lie-3 algebra of $\C{N}$ is turned into a Lie 4-algebra of $\bar{\C{N}}$ (cf.~also Appendix \ref{AppendixC}). Due to the extension, the variation $\delta f$ has many more degrees of freedom compared to the old variation $\delta a$. The gauge symmetries of the non-abelian models in six dimensions of \cite{Samtleben:2011fj}
cannot be described as vertical inner automorphisms of the simpler Q-bundle $\C{N}$ of Figure \ref{Figure:2}, but turn out to have such an interpretation in the lifted  bundle $(\bar{\C{N}},Q_{\bar{\C{N}}})$ of the present figure with its structural Lie 4-algebra.}}
   \label{Figure:3}
\end{figure}

What happens essentially is that the vertical vector field generating the symmetries does not only depend on the value of the section $a\sim A^\alpha$ at a given point in $T[1]\Sigma$, but also at its first ``germ'' $\sim \dd A^\alpha \sim F^\alpha$. The section $f$ is doing precisely this: It gives a canonical lift of the section $a$, where different first derivatives of $a$ give different values (``heights'' in the picture) in the tangent direction of the fibers. For the same value of $a$, the gauge symmetries are or can be different for different lifts. This is taken into account by the vector field generating an inner automorphism on the tangent of the old fiber. 

The symmetries in their original form \eqref{eq:transf} differ from the above ones by a field- or fiber-dependent redefinition of parameters. Concretely, we need just the 
transformation inverse to Equation \eqref{eq:changeLambda}, now written in terms of coordinates on the fiber, i.e.~
\begin{equation}\label{eq:changeLambdainv}
\widetilde{\Lambda}^{r}=\Lambda^{r} \,,\quad \widetilde{\Lambda}^{I}=\Lambda^{I}-d^{I}_{ab}\Lambda^{b}q^{a}\,,\quad 
\widetilde{\Lambda}_{t} = \Lambda_{t}+b_{Jta}\Lambda^{a}q^{J}+\frac{1}{3}b^{\phantom{}}_{Jt[b}d^{J}_{c]a}\Lambda^{a}q^{b}q^{c}\, .
\end{equation}
Implementing this into the vector field \eqref{tildeepsilon}, yields the vector field 
\begin{align} \label{main}
	\epsilon	&= \Lambda^{a}\frac{\partial}{\partial q^{a}} + (\Lambda^{I}-d^{I}_{ab}\Lambda^{a}q^{b})\frac{\partial}{\partial q^{I}} + (\Lambda_{t}+b_{Jta}\Lambda^{a}q^{J}+\frac{1}{3}b^{\phantom{}}_{Jt[b}d^{J}_{c]a}\Lambda^{a}q^{b}q^{c})\frac{\partial}{\partial q_{t}}\\
	&\hspace{1cm}-d^{I}_{ab}\Lambda^{a}\bar\dd q^{b}\frac{\partial}{\partial \bar\dd q^{I}}+\Big[b_{Jts}\Lambda^{J}\bar\dd q^{s}+\frac{4}{3}b_{Jt[u}d^{J}_{s]v}\Lambda^{u}q^{v}\bar\dd q^{s}\Big]\frac{\partial}{\partial \bar\dd q_{t}}\;,\nonumber
\end{align}
which indeed generates the gauge transformations \eqref{eq:transf} by means of 
\begin{align}
\delta_\Lambda A^a &= f^*([Q_{\bar{\C{N}}},\epsilon ]  q^a) \, ,\nonumber\\
\delta_\Lambda B^{I} &= f^*([Q_{\bar{\C{N}}},\epsilon ]  q^I) \, ,\\
\delta_\Lambda (g^{It}q_t) &= f^*([Q_{\bar{\C{N}}},\epsilon ]  (g^{It}q_t)) \, .\nonumber
\end{align}
This change of parameters can be compared with a change of generators of a given (possibly singular) foliation---here the foliation corresponds to gauge equivalence classes within the space of all fields. Vector fields generating a foliation are tangent to the foliation by definition and by $C^\infty$-linear combinations we can obtain any other vector field tangent to the foliation, at least locally. Clearly, this can and in general will change the structure functions, i.e.~the commutators between the generators, but it will not spoil the fact that there is a closed algebra between them (cf.~also footnote \ref{footnote}).\footnote{In the finite dimensional setting we even know that for any foliation we can even locally introduce coordinates and tangent generators such that the latter ones are holonomic, i.e.~of the form of ordinary partial derivatives, so that locally we can always bring such structure functions to zero by changes of generators. A similar result holds also true for gauge symmetries, at least in the conventional setting, and is called ``abelianization'', cf.~\cite{Henneaux:1992ig}.} We thus cannot just implement the parameter change  \eqref{eq:changeLambdainv} into the previous structure functions \eqref{closedalgebra}. Thus we need to reapply the derived bracket construction for the vector fields of the type \eqref{main}. 

Before doing so, however, we recall that in section \ref{section3} we mentioned  an equivalence of parameters in the above symmetries. This is in fact standard for higher gauge theories and qualitatively best visible when one considers abelian higher form degree gauge fields: 
$\delta B = \dd \Lambda$ leaves invariant the (abelian) 3-form field strength $H=\dd B$, but $\Lambda$ and $\Lambda + \dd \mu$ evidently generate the \emph{same} changes of $B$. It is comforting to see that the much more involved invariance \eqref{jauge} of the parametrization in the non-abelian context has a nice cohomological description in the present formulation. Indeed, $\mathrm{ad}_{Q_{\bar{\C{N}}}} \equiv [ Q_{\bar{\C{N}}}, \cdot ]$ is a differential on the space of (also vertical) vector fields. The parametrization enters the symmetries in terms of $\mathrm{ad}_{Q_{\bar{\C{N}}}}(\epsilon)$, with $\epsilon$ given by Equation \eqref{main}, and  \eqref{jauge} corresponds to merely adding  $\mathrm{ad}_{Q_{\bar{\C{N}}}}(\eta)$ to  $\epsilon$, where 
\begin{align}
	\eta	&=\mu^{I}\frac{\partial}{\partial q^{I}}+\mu_{t}\frac{\partial}{\partial q_{t}}+b_{Jta}\mu^{J}\bar{\dd}q^{a}\frac{\partial}{\partial \bar{\dd}q_{t}}\label{eta}
\end{align}
where, as before, $\mu^{I}$ and $\mu_{t}$ are 0-forms and  1-forms, respectively.

We are now left with determining the commutators of the gauge symmetries \eqref{main}. Here the derived bracket turns out to no more be automatically antisymmetric. But since we know that its symmetric part is $\mathrm{ad}_{Q_{\bar{\C{N}}}}$-exact (cf.~Equation \eqref{ex}), we only display its antisymmetric part (with respect to two vector fields $\epsilon$ and $\epsilon'$ having both the form \eqref{main}):
\begin{align}\label{maincomm}
	[\epsilon,\epsilon']^{A}_{Q_{\bar{\C{N}}}}	&= \widehat{\Lambda}^{a}\frac{\partial}{\partial q^{a}} + (\widehat{\Lambda}^{I}-d^{I}_{ab}\widehat{\Lambda}^{a}q^{b})\frac{\partial}{\partial q^{I}} + (\widehat{\Lambda}_{t}+b_{Jta}\widehat{\Lambda}^{a}q^{J}+\frac{1}{3}b^{\phantom{}}_{Jt[b}d^{J}_{c]a}\widehat{\Lambda}^{a}q^{b}q^{c})\frac{\partial}{\partial q_{t}}\\
	&\hspace{1cm}-d^{I}_{ab}\widehat{\Lambda}^{a}\bar\dd q^{b}\frac{\partial}{\partial \bar\dd q^{I}}+\Big[b_{Jts}\widehat{\Lambda}^{J}\bar\dd q^{s}+\frac{4}{3}b_{Jt[u}d^{J}_{s]v}\widehat{\Lambda}^{u}q^{v}\bar\dd q^{s}+k_{t}^{\alpha}(\ldots)_{\alpha}\Big]\frac{\partial}{\partial \bar\dd q_{t}}\;,\nonumber
\end{align}
where, after some calculations and writing $D$ for $\mathrm{d}-q^{a}X_{a}\cdot $ :
\begin{align}\label{mainhat}
	\widehat{\Lambda}^{a}	&= -f_{bc}{}^{a}\Lambda^{b}\Lambda'^{c}\;,\nonumber\\
	\widehat{\Lambda}^{I} 	&= d^{I}_{ab}(\Lambda^{a}D\Lambda'^{b}-\Lambda'^{a}D\Lambda^{b}) +\frac{1}{2}g^{It}b_{Ktb}(\Lambda'^{b}\Lambda^{K}-\Lambda^{b}\Lambda'^{K})\;,\nonumber\\
	\widehat{\Lambda}_{t} 	&= -h_{(I}^r b_{J)\,tr}  \Lambda_{1}^I\wedge  \Lambda_{2}^J
+ b_{I\,tr}\,(\Lambda^I\wedge D \Lambda'^r-\Lambda'^{I}\wedge D\Lambda^{r})\\
						&\hspace{2cm} +2b_{Jt[p} d^J_{q]s} \,\Lambda^p\Lambda'^q \, \bar{\mathrm{d}}q^s+\frac{1}{2}D\big(b_{Ktb}(\Lambda'^{b}\Lambda^{K}-\Lambda^{b}\Lambda'^{K})\big)+k_{t}^{\lambda}(\ldots)_{\lambda}\;.\nonumber
\end{align}
The unspecified terms at the end of Equation \eqref{maincomm} and \eqref{mainhat} turn out to not enter the transformations of the (projected) fields. In the first case, the argument is the same as for \eqref{commtildeps} (given below \eqref{commtildeLambda}). In the second case we may refer to \eqref{fstar} with $q^{\alpha}$ replaced by $g^{It}q_{t}$:
\begin{equation}
f^*\left( [Q_{\bar{\C{N}}},\widehat{\epsilon}] (g^{It}q_{t}) \right) = \dd f^*\left( g^{It}\widehat{\Lambda}_{t} \right)+ f^{\ast}\left({\widehat{\epsilon}}(g^{It}Q_{t}+g^{It}\bar{\dd} q_{t})\right) \, ,
\end{equation}
where $\widehat{\epsilon}\equiv[\epsilon,\epsilon']^{A}_{Q_{\bar{\C{N}}}}$. The contraction of $\widehat{\Lambda}_{t}$ with $g^{It}$ in the first term implies that the last term in \eqref{mainhat} does not contribute there. However in the second term, all the $q_t$-dependence is contained in  $g^{It}Q_{t}=g^{It}A_{ta}^{s}q^{a}q_{s}+\ldots$ (cf.~Equation \eqref{eq:Qagain}). Using the last equation in \eqref{extra} as well as \eqref{eq:gk}, one has $g^{It}A_{ta}^{s}q^{a}\widehat{\Lambda}_{s}=g^{It}b_{Jta}g^{Js}q^{a}\widehat{\Lambda}_{s}$ so that the unspecified term $k_{t}^{\lambda}(\ldots)_{\lambda}$ in \eqref{mainhat} will also not contribute in this case.

We now observe that the last term of the second and third equations in \eqref{mainhat} are nothing but residual gauge transformations, as described in (\ref{jauge}). We thus can drop it as another $\mathrm{ad}_{Q_{\bar{\C{N}}}}$-exact contribution for the choice 
\begin{equation}
\eta = \frac{1}{2}g^{It}b_{Ktb}(\Lambda'^{b}\Lambda^{K}-\Lambda^{b}\Lambda'^{K})\frac{\partial}{\partial q_{t}}
\end{equation}
in Equation \eqref{eta}. This yields the simplified new parameters\footnote{So $\delta_{\widehat{\Lambda}} A^\alpha \equiv \delta_{{\Lambda''}} A^\alpha$.}
\begin{align}\label{mainhat2}
	{\Lambda''}^{a}	&= -f_{bc}{}^{a}\Lambda^{b}\Lambda'^{c}\;,\nonumber\\
	{\Lambda''}^{I} 	&= d^{I}_{ab}(\Lambda^{a}D\Lambda'^{b}-\Lambda'^{a}D\Lambda^{b})\;,\nonumber\\
	{\Lambda''}_{t} 	&= -h_{(I}^r b_{J)\,tr}  \Lambda_{1}^I\wedge  \Lambda_{2}^J
+ b_{I\,tr}\,(\Lambda^I\wedge D \Lambda'^r-\Lambda'^{I}\wedge D\Lambda^{r})\\
						&\hspace{2cm} +2b_{Jt[p} d^J_{q]s} \,\Lambda^p\Lambda'^q \, \bar{\mathrm{d}}q^s+k_{t}^{\lambda}(\ldots)_{\lambda}\;.\nonumber
\end{align}
For these parameters we have:
\begin{align}
[\delta_{\Lambda},\delta_{\Lambda'}]A^a 
&=\delta_{{\Lambda''}} A^a \nonumber \\
[\delta_{\Lambda},\delta_{\Lambda'}]B^I 
&=\delta_{{\Lambda''}} B^I \\
[\delta_{\Lambda},\delta_{\Lambda'}]\left(g^{It}C_{t}\right) 
&=\delta_{{\Lambda''}} \left(g^{It}C_{t}\right) \nonumber
\end{align}


\subsection{Invariant action functionals}
\label{sec:Action}
For physical applications, one of the main goals in the construction of higher gauge theories consists in finding action functionals with the searched-for field content and gauge invariance. The field content of the present model in six space-time dimensions consists of 1-forms $A^a$, 2-forms $B^I$ and (projected) 3-forms $g^{It}C_t$. 
Here the first complication is that the 3-forms have to enter a functional in the projected form.

An essential ingredient of the present approach is the choice of the ideal
$\C{I}$  of field strengths which was tantamount to the choice of the structural Lie $n$-algebra (due to Theorem \ref{thm:QBianchi} and Proposition \ref{propositiontruncation}). This, however, does not yet fix the symmetries completely, but only restricts them. We provided two examples for a different choice of the symmetries.\footnote{The interpolating transformation \eqref{eq:deltaAnew} differs from \eqref{eq:transf} just by a  field dependent transformation of parameters, which does not change the invariance of an action functional.}

The first one of these two options corresponds to inner vertical automorphisms of the simpler bundle $\C{N}$, cf.~Figure \ref{Figure:2}, where the fiber is directly the Q-manifold $\widetilde{\C{W}}$ corresponding to the structural Lie $3$-algebra. Infinitesimally, these are generated by the transformation formulas \eqref{eq:deltaAnewohneF} or, equivalently, by means of  $\delta_{\lambda} A^\alpha = a^* ([Q_{\C{N}},\epsilon])$ together with 
\begin{equation}
\epsilon = \lambda^a \frac{\partial}{\partial q^a} + \lambda^I \frac{\partial}{\partial q^I}+\lambda_t \frac{\partial}{\partial q_t} \, .
\end{equation}
The other option presented corresponds to inner vertical automorphisms of the bundle $\bar{\C{N}}$, cf.~Figure \ref{Figure:3}, where the fiber is the shifted tangent bundle $T[1]\widetilde{\C{W}}$, corresponding to a structural Lie 4-algebra (but one that arises from a canonical lift of a Lie 3-algebra, cf.~Appendix \ref{AppendixC}). One possible set of generators for them is given by \eqref{eq:deltaAnewohneF} or, equivalently, by means of  $\delta_{\lambda} A^\alpha = f^* ([Q_{\bar{\C{N}}},\epsilon])$ together with the vector field $\epsilon$ given in formula \eqref{main}.

The main difference between these two choices is the fact that the first option leads to an open algebra of gauge symmetries whereas the second one yields a closed algebra. This has important consequences for the construction of the action functionals as we are going to see now. In the first case, openness of the gauge algebra requires that the field equations include the vanishing of the contracted field strengths (\ref{contractedF}). A possible action functional is given by~\cite{Mayer:2009wf,Kotov:2010wr}
\begin{equation}\label{S1}
S[A^a,B^I,C_t,\beta_a,\beta_I] = \int_\Sigma \beta_a \wedge \C{F}^a + \beta_I \wedge \C{F}^I + \frac{1}{2} M^{st} \C{F}_s \wedge *\C{F}_t
\end{equation}
with Lagrange multipliers $\beta_a$, $\beta_I$ implementing the Equations \eqref{F}, and a constant symmetric tensor $M^{st}$ which 
satisfies
\begin{equation}
X_{ru}{}^{s} M^{ut} + X_{ru}{}^{t} M^{su} ~=~0
\;,
\label{xs}
\end{equation}
in order to ensure gauge invariance of the action (\ref{S1}). Here the constants $X_{ab}^c$ appearing in Equation \eqref{xs} are those of Equation \eqref{gen1}. 
Moreover, the tensor $M^{st}$ should factorize according to 
\begin{equation}\label{factor}
M^{st} = g^{sI} g^{tJ} m_{IJ}
\;,
\end{equation}
in order to guarantee that only the projected 3-form gauge potentials $g^{sI} C_s$ enter the action. This follows from the fact that the $C_t$-dependence of $\C{F}_t$ is given by $\mathrm{D} C_t$, with the ``covariant derivative'' $\mathrm{D}$ of Equation \eqref{eq:D}, which commutes with the contraction by the ``invariant tensor'' $g^{sI}$ (due to the last equation of \eqref{lolbis1bis}).
Note also that in (\ref{S1}) we could replace $\C{F}^\alpha$ by $\C{H}^\alpha$ due to the presence of the Lagrange multipliers, since the latter ones differ from the former ones only by the addition of combination containing $\C{F}^a$ and $\C{F}^I$, cf.~Equation \eqref{FH}. 

The case of a closed algebra of gauge symmetries on the other hand allows for 
more options of gauge invariant action functionals. Any action of the form
\begin{equation}\label{S2}
S[A^a,B^I,C_t] = \frac{1}{2} \int_\Sigma  N_{ab} \C{H}^a \wedge *\C{H}^b +M_{IJ} \C{H}^I \wedge *\C{H}^J +  M^{st} \C{H}_s \wedge *\C{H}_t
\end{equation}
is gauge invariant, provided the constants $N_{ab}$, $M_{IJ}$, and $M^{st}$ are the components of invariant tensors in the
respective product spaces, in generalization of Equation (\ref{xs}). In addition, as before, $M^{st}$ has to be of the form (\ref{factor}) in order
to guarantee the correct field content.

The action functionals of the supersymmetric models constructed in~\cite{Samtleben:2011fj,Samtleben:2012fb} are yet
of more general type. In this case minimal supersymmetry implies the presence of scalar fields $\phi^I$ joining with the 2-forms $B^I$
into a single supermultiplet. The matrices $N_{ab}$, $M_{IJ}$, and $M^{st}$ in (\ref{S2}) then may become functions of the scalar fields.
This paves the way for a generalisation of the presented framework using 
Lie $n$-algebroids instead of Lie $n$-algebras (cf, e.g., \cite{Mayer:2009wf,Kotov:2010wr,Gruetzmann-Strobl}).
E.g.\ supersymmetry dictates that 
\begin{equation}
N_{ab} = b_{I a b}\,\phi^I
\;,
\end{equation}
such that the relations (\ref{lol}) together with a proper gauge transformation law of the scalar fields
\begin{equation}
\delta_\lambda \phi^I =
-\lambda^a (X_a)_J{}^I\,\phi^J
\;,
\end{equation}
ensures separate gauge invariance of the first term in (\ref{S2}).
Moreover, supersymmetry imposes that $M_{IJ}$ is a constant invertible matrix and that the various constants from (\ref{lol})
defining the theory are in fact related by
\begin{equation}
h^r_I = M_{IJ} \,g^{J r}\;,\qquad b_{I rs} = 2\,M_{IJ} \,d_{rs}^J
\;.
\end{equation}
A final ingredient in the six-dimensional supersymmetric models is a topological term 
that is separately gauge invariant. This term is most compactly given as the boundary contribution
of a seven-dimensional bulk integral~\cite{Samtleben:2011fj}
\begin{equation}\label{CS}
S = \int_{\Sigma_7} M_{IJ} \left(2 d^I_{rs}\, {\cal H}^r \wedge {\cal 
H}^s 
+  { D} {\cal H}^I
\right) \wedge  {\cal H}^J 
\;,
\end{equation}
and necessary for supersymmetry of the complete action functional.
Note that the integrand of \eqref{CS} corresponds to a characteristic class of the 
Q-bundle $\C{N}$ of Figure \ref{Figure:2} according to the construction in \cite{Kotov:2007nr}.

\bigskip
\bigskip

\newpage

\section*{Appendix}

\appendix
\label{appendix}


\section{Explicit proof of Theorem \ref{propositionhenning}}
\label{AppendixA}
\begin{proof}

The setup has been introduced in (\ref{What}) and (\ref{tildetilde}). 
We will show that $\phi^{\ast}\circ\widehat{Q}^{2}=0$ on $\C{C}^{\infty}(\widehat{\C{W}})$ 
implying that $(\widehat{Q}\big|_{\R{Im}(\phi)})^{2}=0$, which proves the theorem. 
Recall the exact form of $\widehat{Q}$ from (\ref{Qhatneu}) with the coefficients given by (\ref{translate}).
Let us first spell out $\widehat{Q}^{2}$:
\begin{eqnarray}
	\widehat{Q}^{2}	&=&
		\left(-\frac{1}{2}C^{e}_{[ab}C^{d}_{c]e}+\frac{1}{6}t^{d}_{D}H^{D}_{abc}\right)\widehat{q}^{a}\widehat{q}^{b}\widehat{q}^{c}\frac{\partial}{\partial \widehat{q}^{d}}
		+\left(-\Gamma^{D}_{cB}t^{a}_{D}-C^{a}_{bc}t^{b}_{B}\right)\widehat{q}^{c}\widehat{q}^{B}\frac{\partial}{\partial \widehat{q}^{a}}
		+t^{a}_{I}\widehat{\widehat{q}}^{I}\frac{\partial}{\partial \widehat{q}^{a}}
		\\&&{}-\left(g^{It}b_{Jta}+\Gamma^{I}_{aJ}\right)\widehat{q}^{a}\widehat{\widehat{q}}^{J}\frac{\partial}{\partial \widehat{q}^{I}}
		+\left(t^{a}_{(A|}\Gamma^{B}_{a|C)}+B_{ACt}g^{Bt}\right)\widehat{q}^{A}\widehat{q}^{C}\frac{\partial}{\partial \widehat{q}^{B}}
		\nonumber\\&&{}+\left(\frac{1}{2}C^{c}_{ab}\Gamma^{B}_{cC}+\frac{1}{2}t^{c}_{C}H^{B}_{cab}+\Gamma^{A}_{[a|C}\Gamma^{B}_{|b]A}-D_{abCt}g^{Bt}\right)\widehat{q}^{a}\widehat{q}^{b}\widehat{q}^{C}\frac{\partial}{\partial \widehat{q}^{B}}
		\nonumber\\&&{}-\left(\frac{1}{4}C^{s}_{[ab}H^{B}_{cd]s}+\frac{1}{6}H^{A}_{[abc}\Gamma^{B}_{d]A}+E_{abcdt}g^{Bt}\right)\widehat{q}^{a}\widehat{q}^{b}\widehat{q}^{c}\widehat{q}^{d}\frac{\partial}{\partial \widehat{q}^{B}}
		\nonumber\\&&{}-g^{It}\left(\frac{1}{2}C^{c}_{ab}b_{Jtc}-g^{Ks}b_{Js[a|}b_{Kt|b]}+D_{abJt}\right)\widehat{q}^{a}\widehat{q}^{b}\widehat{\widehat{q}}^{J}\frac{\partial}{\partial \widehat{\widehat{q}}^{I}}
		+g^{It}\left(t^{a}_{A}b_{Bta}-2B_{ABt}\right)\widehat{q}^{A}\widehat{\widehat{q}}^{B}\frac{\partial}{\partial \widehat{\widehat{q}}^{I}}
		\nonumber\\&&{}+g^{It}\left(2D_{ab(A|t}t^{b}_{|B)}+2\Gamma^{C}_{a(A}B_{B)Ct}+B_{ABs}g^{Cs}b_{Cta}\right)\widehat{q}^{a}\widehat{q}^{A}\widehat{q}^{B}\frac{\partial}{\partial \widehat{\widehat{q}}^{I}}
		\nonumber\\&&{}-g^{It}\left(2C^{s}_{[ab}E_{cde]st}+\frac{1}{6}H^{A}_{[abc}D_{de]At}-E_{[abcd|s}g^{Js}b_{Jt|e]}\right)\widehat{q}^{a}\widehat{q}^{b}\widehat{q}^{c}\widehat{q}^{d}\widehat{q}^{e}\frac{\partial}{\partial \widehat{\widehat{q}}^{I}}
		\nonumber\\&&{}-g^{It}\left(C^{s}_{[ab}D_{c]sBt}+4t^{s}_{B}E_{sabct}-\Gamma^{A}_{[a|B}D_{|bc]At}+\frac{1}{3}H^{A}_{abc}B_{ABt}-g^{Ks}b_{Kt[c}D_{|ab]Bs}\right)\widehat{q}^{a}\widehat{q}^{b}\widehat{q}^{c}\widehat{q}^{B}\frac{\partial}{\partial \widehat{\widehat{q}}^{I}}
		\nonumber
\end{eqnarray}
By definition $\C{C}^{\infty}(\widehat{\C{W}})$ is the set of formal power series in $\widehat{q}^{a},\widehat{q}^{I},\widehat{\widehat{q}}^{J}$. Then the equation $\phi^{\ast}\circ\widehat{Q}^{2}=0$ on $\C{C}^{\infty}(\widehat{\C{W}})$ is equivalent to the fact that the pullback by $\phi$ of each component of $\widehat{Q}^{2}$ vanishes. Consequently, $\phi^{\ast}\circ\widehat{Q}^{2}=0$ is equivalent to the following set of equations :
\begin{align}
	\frac{1}{2}C^{e}_{[ab}C^{d}_{c]e}-\frac{1}{6}t^{d}_{D}H^{D}_{abc}							&=0\label{bobbis1}\;,\\
	-\Gamma^{D}_{cB}t^{a}_{D}-C^{a}_{bc}t^{b}_{B}												&=0\;,\nonumber\\
	t^{a}_{I}g^{It}																			&=0\;,\nonumber\\
	g^{It}b_{Jta}g^{Js}+\Gamma^{I}_{aJ}g^{Js}													&=0\;,\nonumber\\
	t^{a}_{(A|}\Gamma^{B}_{a|C)}+B_{ACt}g^{Bt}												&=0\;,\nonumber\\
	\frac{1}{2}C^{c}_{ab}\Gamma^{B}_{cC}+\frac{1}{2}t^{c}_{C}H^{B}_{cab}+\Gamma^{A}_{[a|C}\Gamma^{B}_{|b]A}-D_{abCt}g^{Bt}
																							&=0\;,\nonumber\\
	\frac{1}{4}C^{s}_{[ab}H^{B}_{cd]s}+\frac{1}{6}H^{A}_{[abc}\Gamma^{B}_{d]A}+E_{abcdt}g^{Bt}	&=0\;,\nonumber\\
	g^{It}\left(\frac{1}{2}C^{c}_{ab}b_{Jtc}g^{Ju}-g^{Ks}b_{Js[a|}b_{Kt|b]}g^{Ju}+D_{abJt}g^{Ju}\right)
																							&=0\;,\nonumber\\
	g^{It}\left(g^{Bs}t^{a}_{A}b_{Bta}-2g^{As}B_{ABt}\right)									&=0\;,\nonumber\\
	g^{It}\left(2D_{ab(A|t}t^{b}_{|B)}+2\Gamma^{C}_{a(A}B_{B)Ct}+B_{ABs}g^{Cs}b_{Cta}\right)	&=0\;,\nonumber\\
	g^{It}\left(2C^{s}_{[ab}E_{cde]st}+\frac{1}{6}H^{A}_{[abc}D_{de]At}-E_{[abcd|s}g^{Js}b_{Jt|e]}\right)
																							&=0\;,\nonumber\\
	g^{It}\left(C^{s}_{[ab}D_{c]sBt}+4t^{s}_{B}E_{sabct}-\Gamma^{A}_{[a|B}D_{|bc]At}+\frac{1}{3}H^{A}_{abc}B_{ABt}-g^{Ks}b_{Kt[c}D_{|ab]Bs}\right)																				&=0\;.\nonumber
\end{align}
Let us recall that the equations needed to establish the Bianchi identities are given by (\ref{lol}). Using that $h^{r}_{I}g^{It}=0$, these equations imply that
\begin{align}
	A^{s}_{ru}d^{I}_{sv}+A^{s}_{rv}d^{I}_{us}-\Gamma'^{I}_{rJ}d^{J}_{uv}			&=0\;,\label{lolbis1bis}\\
	\Gamma'^{J}_{rI}b_{Jst}+A^{u}_{rs}b_{Iut}+A^{u}_{rt}b_{Isu}					&=0\;,\nonumber\\
	A^{u}_{[rp|}f_{u|q]}{}^{s}+A^{u}_{[rq}f_{p]u}{}^{s}-A^{s}_{[r|u}f_{|pq]}{}^{u}	&=0\;,\nonumber\\
	A^{t}_{sr}h^{r}_{I}-\Gamma'^{J}_{sI}h^{t}_{J}									&=0\;,\nonumber\\
	A^{s}_{rt}g^{It}+\Gamma'^{I}_{rK}g^{Ks}										&=0\;,\nonumber
\end{align}
where $\Gamma'^{I}_{rK}\equiv2h^{s}_{K}d^{I}_{rs}-g^{Is}b_{Ksr}$. Recalling that $A^{s}_{rt}=(X_{r})_{t}{}^{s}$ 
and $\Gamma'^{I}_{rK}=(X_{r})_{K}{}^{J}$ encode the action (\ref{covA}), (\ref{covBC}), of the gauge generators 
$-X_{r}$ on the vector spaces $\C{V}_{-1}$ and $\C{V}_{-2}$, respectively, the five equations (\ref{lolbis1bis}) express nothing 
but the fact that the tensors $d^{I}_{rs}, b_{Jts}, f_{pq}{}^{s},h^{t}_{J}, g^{Is}$
define invariant subspaces in the respective tensor products. The same thus holds for any contraction of these
tensors, in particular for the objects defined in (\ref{translate}). Which in particular implies the relations
\begin{align}
			A_{[qa]}^{u}\Gamma^{J}_{uK}+\Gamma'^{I}_{[q|K}\Gamma^{J}_{|a]I}-\Gamma'^{J}_{[q|I}\Gamma^{I}_{|a]K}&=0\;,\\
			3A_{[qa|}^{u}H_{u|bc]}^{J}-\Gamma'^{J}_{[q|I}H^{I}_{|abc]}&=0\;,\nonumber\\
			A_{[qa]}^{u}A_{us}^{t}+A_{[q|s}^{r}A_{|a]r}^{t}-A_{[q|r}^{t}A_{|a]s}^{r}&=0\;,\nonumber\\
			2\Gamma'^{I}_{q(J}B_{K)It}+A^{s}_{qt}B_{JKs}&=0\;,\nonumber\\
				\Gamma'^{I}_{[q|K}D_{|ab]It}+2A_{[qa|}^{u}D_{u|b]Kt}+A_{[q|t}^{s}D_{|ab]Ks}&=0\;,\nonumber\\
				4A^{u}_{[qa|}E_{u|bcd]t}+A^{s}_{[q|t}E_{|abcd]s}&=0
	\;.\nonumber
\end{align}
These relations will be useful to prove the last identities of (\ref{bobbis1}). Another useful equation is obtained by noting that the first equation of (\ref{lol}), when antisymmetrized in $r$ and $u$ indices, gives :
\begin{equation}
-3h^{s}_{K}d^{K}_{v[u}d^{I}_{r]s}\underbrace{-f_{ru}{}^{s}d^{I}_{vs}-f_{[r|v}{}^{s}d^{I}_{u]s}}_{=-\frac{3}{2}f_{ru}{}^{s}d^{I}_{vs}+\frac{3}{2}f_{[ur}{}^{s}d^{I}_{v]s}}=-g^{Is}b_{Ks[r}d^{K}_{u]v}
\;,
\end{equation}
so that we obtain 
\begin{equation}
	\frac{1}{2}f_{[ru}{}^{s}d^{I}_{v]s}-\frac{1}{2}f_{ru}{}^{s}d^{I}_{vs}-h^{s}_{K}d^{I}_{s[r}d^{K}_{u]v}=-\frac{1}{3}g^{Is}b_{Ks[r}d^{K}_{u]v}\;.\label{trocoule}
\end{equation}

Let us now start showing that all equations (\ref{bobbis1}) are satisfied. The first five equations are direct consequences of equations (\ref{lol}).  Now, let us check the sixth equation 
\begin{equation}
\frac{1}{2}C^{c}_{ab}\Gamma^{B}_{cC}+\frac{1}{2}t^{c}_{C}H^{B}_{cab}+\Gamma^{A}_{[a|C}\Gamma^{B}_{|b]A}-D_{abCt}g^{Bt}	=0\;.
\label{bob6}
\end{equation}
First of all, Equation (\ref{trocoule}) gives:
\begin{equation}
-D_{abCt}g^{Bt}=-\frac{1}{2}C^{c}_{ab}h_{C}^{s}d^{B}_{cs}+\frac{1}{2}h^{c}_{C}H^{B}_{cab}+h_{C}^{s}h_{K}^{v}d^{B}_{v[a}d^{K}_{b]s}
\end{equation}
so that the left hand side of equation (\ref{bob6}) becomes
\begin{eqnarray}
\frac{1}{2}f_{ab}{}^{c}g^{Bt}b_{Ctc}-\Gamma^{A}_{[a|C}g^{Bt}b_{At|b]}
&=&\frac12\left(
-A^{c}_{[ab]}g^{Bt}b_{Ctc}-\Gamma'^{A}_{[a|C}g^{Bt}b_{At|b]}+(-g^{Bt}b_{At[a|})g^{As}b_{Cs|b]}\right)
\nonumber\\
&=&\frac12\left(-A^{c}_{[ab]}g^{Bt}b_{Ctc}-\Gamma'^{A}_{[a|C}g^{Bt}b_{At|b]}+\Gamma'^{B}_{[a|A}g^{As}b_{Cs|b]}\right)
\;,\nonumber
\end{eqnarray}
because $\Gamma'^{B}_{[a|A}g^{As}b_{Cs|b]}=(-g^{Bt}b_{At[a|})g^{As}b_{Cs|b]}$. But the last equation 
is just the expression of gauge invariance of the tensor $g^{Bt}b_{Ctb}$, so vanishes identically.
Next, let us check the seventh equation of (\ref{bobbis1})
\begin{equation}
\frac{1}{4}C^{s}_{[ab}H^{B}_{cd]s}+\frac{1}{6}H^{A}_{[abc}\Gamma^{B}_{d]A}+E_{abcdt}g^{Bt}	=0
\;,
\label{bob7}
\end{equation}
whose left-hand side reduces to
\begin{eqnarray}
&=& \frac14A^{s}_{[ab|}H^{B}_{s|cd]}-\frac1{12}\,\Gamma'^{B}_{[a|A}H^{A}_{|bcd]}+\frac1{12}\,g^{Bt}b_{At[a|}H^{A}_{|bcd]}+\frac{1}{12}b_{At[a}f_{bc}{}^{s}d^{A}_{d]s}g^{Bt}\nonumber\\
&=&
\frac14A^{s}_{[ab|}H^{B}_{s|cd]}-\frac1{12}\,\Gamma'^{B}_{[a|A}H^{A}_{|bcd]}
\;,\nonumber
\end{eqnarray}
which is nothing but the expression of the gauge invariance of the tensor $H^{B}_{bcd}$, thus identically zero, which proves equation (\ref{bob7}).

We remain with the last five equations of (\ref{bobbis1}), all of which need to vanish under projection with $g^{It}$. Here, we can make use of
the orthogonality of $g^{It}$ with $k_t^\lambda$, cf.~(\ref{extra}), such that it suffices to show that the terms under projection can be cast into the form
\begin{eqnarray}
	\frac{1}{2}C^{c}_{ab}A^{u}_{ct}-A^{s}_{[b|t}A^{u}_{|a]s}+D_{abJt}g^{Ju}					&=& k_{t}^{\lambda}\,\Xi^{(1)}_{\lambda}\;,\label{bob8}\\
	t^{a}_{B}A^{s}_{at}-2g^{As}B_{ABt}														&=& k_{t}^{\lambda}\,\Xi^{(2)}_{\lambda}\;,\label{bob9}\\
	2D_{ab(A|t}t^{b}_{|B)}+2\Gamma^{C}_{a(A}B_{B)Ct}+B_{ABs}A^{s}_{at}						&=& k_{t}^{\lambda}\,\Xi^{(3)}_{\lambda}\;,\label{bob10}\\
	2C^{s}_{[ab}E_{cde]st}+\frac{1}{6}H^{A}_{[abc}D_{de]At}-E_{[abcd|s}A^{s}_{e]t}				&=& k_{t}^{\lambda}\,\Xi^{(4)}_{\lambda}\;,\label{bob11}\\
	C^{s}_{[ab}D_{c]sBt}+4t^{s}_{B}E_{sabct}-\Gamma^{A}_{[a|B}D_{|bc]At}+\frac{1}{3}H^{A}_{abc}B_{ABt}-A^{s}_{[c|t}D_{|ab]Bs}																										&=& k_{t}^{\lambda}\,\Xi^{(5)}_{\lambda}\label{bob12}
	\;,
\end{eqnarray}
with suitable expressions $\Xi^{(i)}_{\lambda}$. In rewriting (\ref{bobbis1}) as (\ref{bob8})--(\ref{bob12}), 
we have used the identity $A_{rt}^{s}g^{It}=g^{Js}b_{Jtr}g^{It}$, which is the seventh equation of (\ref{lol}).

Let us start with equation (\ref{bob8}): first the last term $g^{Ju}D_{abJt}$ vanishes because $h^{r}_{A}g^{At}=0$. The remainder is $\frac{1}{2}C^{c}_{ab}A^{u}_{ct}-A^{u}_{[a|\sigma}A^{\sigma}_{|b]t}$ which is proportional to
\begin{equation} 
A^{c}_{[ab]}A^{u}_{ct}-2A^{u}_{[a|s}A^{s}_{|b]t}=A^{c}_{[ab]}A^{u}_{ct}+A^{s}_{[a|t}A^{u}_{|b]s}-A^{u}_{[a|s}A^{s}_{|b]t}\;,
\end{equation}
which is nothing more than the expression of the gauge invariance of the tensor $A^{u}_{bt}$, which vanishes, such that equation (\ref{bob8}) is satisfied.
Equation (\ref{bob9}) is easy as well. Indeed 
\begin{align}
	t^{a}_{B}A^{s}_{at}-2g^{As}B_{ABt}
				&=-h^{a}_{B}(-f_{at}{}^{s}+h^{s}_{K}d^{K}_{at})+g^{As}h^{a}_{(A}b_{B)t a}\nonumber\\
				&=h^{a}_{B}(-f_{t a}{}^{s}-h^{s}_{K}d^{K}_{t a}+g^{As}b_{At a})-g^{As}h^{a}_{[B}b_{A]t a}\nonumber\\
				&=-h^{a}_{B}k^{\lambda}_{t}c^{s}_{\lambda a}-g^{As}k^{\lambda}_{t}c_{\lambda BA}\;,\nonumber
\end{align}
which is of the desired form.
Concerning equation (\ref{bob10}), the left hand side is nothing but
\begin{equation}
\frac{1}{6}h_{(A|}^{s}b_{Kt a}d^{K}_{bs}h^{b}_{|B)}-\frac{1}{6}h_{(A|}^{s}b_{Kt b}d^{K}_{as}h^{b}_{|B)}+\Gamma^{C}_{a(A}B_{B)Ct}+\frac{1}{2}B_{ABs}A^{s}_{at}
\;,
\end{equation}
but with the last equation of (\ref{lol}) the first term can be rewritten as
\begin{align}
	&-\frac{1}{6}h_{(A|}^{s}b_{Kt s}d^{K}_{ab}h^{b}_{|B)}-\frac{1}{3}h_{(A|}^{s}b_{Kt b}d^{K}_{as}h^{b}_{|B)}+\Gamma^{C}_{a(A}B_{B)Ct}+\frac{1}{2}B_{ABs}A^{s}_{at}\nonumber\\
	&=-\frac{1}{6}h_{(A|}^{s}b_{Kt b}d^{K}_{as}h^{b}_{|B)}-\frac{1}{3}h_{(A|}^{s}b_{Kt b}d^{K}_{as}h^{b}_{|B)}-d^{C}_{sa}h_{(A}^{s}B_{B)Ct}+\Gamma'^{C}_{a(A}B_{B)Ct}+\frac{1}{2}B_{ABs}A^{s}_{at}
\nonumber
\end{align}
where  $b$ and $s$ have been exchanged in the first term because of the symmetry between $h_{A}^{s}$ and $h_{B}^{b}$. We obtain
\begin{align}
	&\frac{1}{6}h_{(A|}^{s}b_{Kt b}d^{K}_{as}h^{b}_{|B)}-d^{C}_{sa}h_{(A}^{s}B_{B)Ct}+\Gamma'^{C}_{a(A}B_{B)Ct}+\frac{1}{2}B_{ABs}A^{s}_{at}\nonumber\\
	&=-\frac{1}{2}h_{(A|}^{s}b_{Kt b}d^{K}_{as}h^{b}_{|B)}+\frac{1}{4}d^{C}_{sa}h_{(A}^{s}h_{B)}^{b}b_{Ct b} +\frac{1}{4}d^{C}_{sa}h_{(A|}^{s}h_{C}^{b}b_{|B)t b}+\Gamma'^{C}_{a(A}B_{B)Ct}+\frac{1}{2}B_{ABs}A^{s}_{at}\nonumber\\
	&=\frac{1}{4}d^{C}_{sa}h_{(A|}^{s}h_{C}^{b}b_{|B)t b}-\frac{1}{4}h_{(A|}^{s}b_{Kt b}d^{K}_{as}h^{b}_{|B)}+\Gamma'^{C}_{a(A}B_{B)Ct}+\frac{1}{2}B_{ABs}A^{s}_{at}\nonumber\\
	&=\frac{1}{2}h_{(A|}^{s}d^{K}_{as}k^{\lambda}_{t}c_{\lambda K|B)}+\Gamma'^{C}_{a(A}B_{B)Ct}+\frac{1}{2}B_{ABs}A^{s}_{at}\;.\nonumber
\end{align}
Now it is enough to notice that the first term is proportional to $k_{t}^{\lambda}$, and the two last terms express the gauge invariance of the tensor $B_{ABt}$, which then automatically vanish. So we have proven that also equation (\ref{bob10}) is of the desired form.

The two last equations will be more tricky to show. First, lets turn to equation (\ref{bob11}), in which we can use the gauge invariance of $E_{abcd\rho}$ to mix the first and the third term. Since $4A^{s}_{[ea|}E_{s|bcd]t}+A^{s}_{[e|t}E_{|abcd]s}=0$, the left hand side of the equation becomes
\begin{align}
-2C^{s}_{[ab}E_{cde]st}+\frac{1}{6}H^{A}_{[abc}D_{de]At}
\;,\nonumber
\end{align}
which is proportional to 
\begin{align}
&f_{\bar{a}\bar{b}}{}^{s}b_{Kt [\bar{c}}f_{\bar{d}\bar{e}}{}^{v}d_{s]v}^{K}+\frac{1}{3}d^{A}_{s[\bar{a}}f_{\bar{b}\bar{c}}{}^{s}h_{A}^{v}b_{Kt[\bar{d}}d_{\bar{e}]v}^{K}
=f_{\bar{a}\bar{b}}{}^{s}b_{Kt [\bar{c}}f_{\bar{d}\bar{e}}{}^{v}d_{s]v}^{K}+f_{\bar{a}\bar{b}}{}^{s}f_{\bar{c}s}{}^{v}b_{Kt\bar{d}}d_{\bar{e}v}^{K}
\;,
\nonumber
\end{align}
by using the third  equation of (\ref{lol}) on the last term. 
Here, and in the following, we use the notation that all expressions are considered as projected onto the 
fully antisymmetric part in all barred indices.
When we expand the bracket in the first term we have:
\begin{align}
&\frac{1}{4}f_{\bar{a}\bar{b}}{}^{s}b_{Kt \bar{c}}f_{\bar{d}\bar{e}}{}^{v}d_{sv}^{K}-\frac{1}{4}f_{\bar{a}\bar{b}}{}^{s}b_{Kt s}f_{\bar{c}\bar{d}}{}^{v}d_{\bar{e}v}^{K}+\frac{1}{2}f_{\bar{a}\bar{b}}{}^{s}b_{Kt \bar{e}}f_{s\bar{c}}{}^{v}d_{\bar{d}v}^{K}+f_{\bar{a}\bar{b}}{}^{s}f_{\bar{c}s}{}^{v}b_{Kt\bar{d}}d_{\bar{e}v}^{K}\nonumber\\
&=\frac{1}{4}f_{\bar{a}\bar{b}}{}^{s}b_{Kt \bar{c}}f_{\bar{d}\bar{e}}{}^{v}d_{sv}^{K}-\frac{1}{4}f_{\bar{a}\bar{b}}{}^{s}b_{Kt s}f_{\bar{c}\bar{d}}{}^{v}d_{\bar{e}v}^{K}+\frac{3}{2}f_{\bar{a}\bar{b}}{}^{s}f_{\bar{c}s}{}^{v}b_{Kt\bar{d}}d_{\bar{e}v}^{K}\nonumber\\
&=-\frac{3}{4}f_{\bar{a}\bar{b}}{}^{s}b_{Kt s}f_{\bar{c}\bar{d}}{}^{v}d_{\bar{e}v}^{K}+\frac{3}{2}f_{\bar{a}\bar{b}}{}^{s}f_{\bar{c}s}{}^{v}b_{Kt\bar{d}}d_{\bar{e}v}^{K}\nonumber
\end{align}
where we made heavy use of the (anti)-symmetries (such as permuting $s$ and $v$). Now let us use the second equation of (\ref{lol}) on $\frac{4}{3}$ times the first term and the first one of (\ref{lol}) on $\frac{4}{3}$ times the second term to get:
\begin{align}
-f_{\bar{a}\bar{b}}{}^{s}b_{Kt s}f_{\bar{c}\bar{d}}{}^{v}d_{\bar{e}v}^{K}&=-A^{s}_{\bar{a}t}b_{Ks\bar{b}}f_{\bar{c}\bar{d}}{}^{v}d_{\bar{e}v}^{K} -\Gamma'^{J}_{\bar{a}K}b_{Jt \bar{b}}f_{\bar{c}\bar{d}}{}^{v}d_{\bar{e}v}^{K}\nonumber\\
2f_{\bar{a}\bar{b}}{}^{s}f_{\bar{c}s}{}^{v}b_{Kt\bar{d}}d_{\bar{e}v}^{K}&=2f_{\bar{a}\bar{b}}{}^{s}h^{v}_{J}d^{J}_{\bar{c}s}b_{Kt\bar{d}}d_{\bar{e}v}^{K}-2f_{\bar{a}\bar{b}}{}^{s}f_{\bar{c}\bar{e}}{}^{v}b_{Kt\bar{d}}d^{K}_{vs}-2f_{\bar{a}\bar{b}}{}^{s}\Gamma'^{K}_{\bar{c}J}b_{Kt\bar{d}}d^{J}_{\bar{e}s}\nonumber\\
&=-2f_{\bar{a}\bar{b}}{}^{s}f_{\bar{c}\bar{e}}{}^{v}b_{Kt\bar{d}}d^{K}_{vs}-6f_{\bar{a}\bar{b}}{}^{s}h^{v}_{J}d^{J}_{\bar{e}s}b_{Kt\bar{d}}d_{\bar{c}v}^{K}+2f_{\bar{a}\bar{b}}{}^{s}g^{Kt}b_{Jt\bar{c}}b_{Kt\bar{d}}d^{J}_{\bar{e}s}\;.\nonumber
\end{align}
When we add one line to the other, we get
\begin{align}
&(f_{\bar{a}t}{}^{s}-h^{s}_{J}d^{J}_{\bar{a}t}+g^{Js}b_{Jt\bar{a}})b_{Ks\bar{b}}f_{\bar{c}\bar{d}}{}^{v}d_{\bar{e}v}^{K}-2f_{\bar{a}\bar{b}}{}^{s}f_{\bar{c}\bar{e}}{}^{v}b_{Kt\bar{d}}d^{K}_{vs}-8f_{\bar{a}\bar{b}}{}^{s}h^{v}_{J}d^{J}_{\bar{e}s}b_{Kt\bar{d}}d_{\bar{c}v}^{K}+4f_{\bar{a}\bar{b}}{}^{s}g^{Kt}b_{Jt\bar{c}}b_{Kt\bar{d}}d^{J}_{\bar{e}s}\nonumber\\
&=-k^{\alpha}_{t}c_{\alpha\bar{a}}^{s}b_{Ks\bar{b}}f_{\bar{c}\bar{d}}{}^{v}d_{\bar{e}v}^{K}+4f_{\bar{a}\bar{b}}{}^{s}f_{\bar{c}\bar{e}}{}^{v}b_{Kt s}d^{K}_{v\bar{d}}-4f_{\bar{a}\bar{b}}{}^{s}\Gamma'^{K}_{\bar{c}J}b_{Kt\bar{d}}d^{J}_{\bar{e}s}\nonumber\\
&=-k^{\alpha}_{t}c_{\alpha\bar{a}}^{s}b_{Ks\bar{b}}f_{\bar{c}\bar{d}}{}^{v}d_{\bar{e}v}^{K}+4f_{\bar{a}\bar{b}}{}^{s}f_{\bar{c}\bar{e}}{}^{v}b_{Kt s}d^{K}_{v\bar{d}}-4f_{\bar{a}\bar{b}}{}^{s}b_{Kt\bar{d}}A^{v}_{\bar{c}\bar{e}}d_{vs}^{K}-4f_{\bar{a}\bar{b}}{}^{s}b_{Kt\bar{d}}A^{v}_{\bar{c}s}d_{v\bar{e}}^{K}\nonumber\\
&=-k^{\alpha}_{t}c_{\alpha\bar{a}}^{s}b_{Ks\bar{b}}f_{\bar{c}\bar{d}}{}^{v}d_{\bar{e}v}^{K}+4f_{\bar{a}\bar{b}}{}^{s}f_{\bar{c}\bar{e}}{}^{v}b_{Kt s}d^{K}_{v\bar{d}}+\underbrace{4f_{\bar{a}\bar{b}}{}^{s}b_{Kt\bar{d}}f^{v}_{\bar{c}\bar{e}}d_{vs}^{K}}_{-8f_{\bar{a}\bar{b}}{}^{s}b_{Kt s}f^{v}_{\bar{c}\bar{e}}d_{\bar{d}s}^{K}}+4f_{\bar{a}\bar{b}}{}^{s}b_{Kt\bar{d}}f^{v}_{\bar{c}s}d_{v\bar{e}}^{K}-\underbrace{4f_{\bar{a}\bar{b}}{}^{s}b_{Kt\bar{d}}h^{v}_{J}d^{J}_{\bar{c}s}d_{v\bar{e}}^{K}}_{12f_{\bar{a}\bar{b}}{}^{s}b_{Kt\bar{d}}f^{v}_{\bar{c}s}d_{v\bar{e}}^{K}}\nonumber\\
&=-k^{\alpha}_{t}c_{\alpha\bar{a}}^{s}b_{Ks\bar{b}}f_{\bar{c}\bar{d}}{}^{v}d_{\bar{e}v}^{K}-4f_{\bar{a}\bar{b}}{}^{s}f_{\bar{c}\bar{e}}{}^{v}b_{Kt s}d^{K}_{v\bar{d}}-8f_{\bar{a}\bar{b}}{}^{s}b_{Kt\bar{d}}f^{v}_{\bar{c}s}d_{v\bar{e}}^{K}\nonumber
\end{align}
where we used the eighth and the third equation of (\ref{lol}) at the end. Thus, to conclude we have the following equality
\begin{equation}
-f_{\bar{a}\bar{b}}{}^{s}b_{Kt s}f_{\bar{c}\bar{d}}{}^{v}d_{\bar{e}v}^{K}+2f_{\bar{a}\bar{b}}{}^{s}f_{\bar{c}s}{}^{v}b_{Kt\bar{d}}d_{\bar{e}v}^{K}=-k^{\alpha}_{t}c_{\alpha\bar{a}}^{s}b_{Ks\bar{b}}f_{\bar{c}\bar{d}}{}^{v}d_{\bar{e}v}^{K}+4f_{\bar{a}\bar{b}}{}^{s}f_{\bar{c}\bar{d}}{}^{v}b_{Kt s}d^{K}_{v\bar{e}}-8f_{\bar{a}\bar{b}}{}^{s}b_{Kt\bar{d}}f^{v}_{\bar{c}s}d_{v\bar{e}}^{K}
\;,\nonumber
\end{equation}
which is equivalent to the following:
\begin{equation}
-f_{\bar{a}\bar{b}}{}^{s}b_{Kt s}f_{\bar{c}\bar{d}}{}^{v}d_{\bar{e}v}^{K}+2f_{\bar{a}\bar{b}}{}^{s}f_{\bar{c}s}{}^{v}b_{Kt\bar{d}}d_{\bar{e}v}^{K}=-\frac{1}{5}k^{\lambda}_{t}c_{\lambda\bar{a}}^{s}b_{Ks\bar{b}}f_{\bar{c}\bar{d}}{}^{v}d_{\bar{e}v}^{K}
\;,\nonumber
\end{equation}
which is proportional to (\ref{bob11}).

Finally, let us turn to equation (\ref{bob12}). First note that using the following equation
\begin{equation}
\Gamma'^{I}_{[c|B}D_{|ab]It}+2A_{[ca|}^{u}D_{u|b]Bt}+A_{[c|t}^{s}D_{|ab]Bs}=0\;,
\nonumber
\end{equation}
and the fact that $h_{A}^{r}g^{At}=0$, we can rewrite the left hand side of (\ref{bob12}) into:
\begin{equation}
C^{s}_{[ab}D_{c]sBt}+4t^{s}_{B}E_{sabct}+\Gamma^{A}_{[a|B}D_{|bc]At}+\frac{1}{3}H^{A}_{abc}B_{ABt}-2A_{[ca}^{s}D_{b]sBt}
\;,\nonumber\\
\end{equation}
which is proportional to 
\begin{align}
&f^{s}_{\bar{a}\bar{b}}h^{v}_{B}b_{Jt[\bar{c}}d^{J}_{s]v}-h^{s}_{B}b_{Jt[s}f^{v}_{\bar{a}\bar{b}}d^{J}_{\bar{c}]v}-h_{B}^{v}d^{A}_{v\bar{a}}h_{A}^{w}b_{It[\bar{b}}d^{I}_{\bar{c}]w}+\frac{1}{2}h^{v}_{(A}b_{B)t v}d^{A}_{s[\bar{a}}f^{s}_{\bar{b}\bar{c}]}-2f^{s}_{\bar{a}\bar{b}}h_{B}^{v}b_{Jt[\bar{c}}d^{J}_{s]v}\nonumber\\
&=-f^{s}_{\bar{a}\bar{b}}h^{v}_{B}b_{Jt[\bar{c}}d^{J}_{s]v}-h^{s}_{B}b_{Jt[s}f^{v}_{\bar{a}\bar{b}}d^{J}_{\bar{c}]v}-h_{B}^{v}d^{A}_{v\bar{a}}h_{A}^{w}b_{It[\bar{b}}d^{I}_{\bar{c}]w}+\frac{1}{2}h^{v}_{(A}b_{B)t v}d^{A}_{s[\bar{a}}f^{s}_{\bar{b}\bar{c}]}\nonumber\\
&=-\frac{1}{2}f^{s}_{\bar{a}\bar{b}}h^{v}_{B}b_{Jt\bar{c}}d^{J}_{sv}+\hspace{-1cm}\underbrace{\frac{1}{2}f^{s}_{\bar{a}\bar{b}}h^{v}_{B}b_{Jt s}d^{J}_{\bar{c}v}}_{-\frac{1}{2}f^{s}_{\bar{a}\bar{b}}h^{v}_{B}b_{Jt\bar{c}}d^{J}_{vs}-\frac{1}{2}f^{s}_{\bar{a}\bar{b}}h^{v}_{B}b_{Jt v}d^{J}_{s\bar{c}}}\hspace{-1cm}-h^{s}_{B}b_{Jt[s}f^{v}_{\bar{a}\bar{b}}d^{J}_{\bar{c}]v}-h_{B}^{v}d^{A}_{v\bar{a}}h_{A}^{w}b_{It[\bar{b}}d^{I}_{\bar{c}]w}+\frac{1}{2}h^{v}_{(A}b_{B)t v}d^{A}_{s[\bar{a}}f^{s}_{\bar{b}\bar{c}]}\nonumber\\
&=-f^{s}_{\bar{a}\bar{b}}h^{v}_{B}b_{Jt\bar{c}}d^{J}_{sv}-h^{s}_{B}b_{Jt[s}f^{v}_{\bar{a}\bar{b}}d^{J}_{\bar{c}]v}-h_{B}^{v}d^{A}_{v\bar{a}}h_{A}^{w}b_{It[\bar{b}}d^{I}_{\bar{c}]w}-\frac{1}{2}h^{v}_{[B}b_{A]t v}d^{A}_{s[\bar{a}}f^{s}_{\bar{b}\bar{c}]}\;.
\nonumber
\end{align}
Again, all expressions are considered as totally antisymmetric in the barred indices.
Let us look at the first three terms (without considering $h_{B}^{v}$). There are equal to:
\begin{align}
&-f^{s}_{\bar{a}\bar{b}}b_{Jt\bar{c}}d^{J}_{sv}-\frac{1}{4}b_{Jt v}f^{s}_{\bar{a}\bar{b}}d^{J}_{\bar{c}s}+\frac{1}{4}b_{Jt \bar{c}}f^{s}_{\bar{a}\bar{b}}d^{J}_{vs}-\frac{1}{2}b_{Jt \bar{b}}f^{s}_{v\bar{a}}d^{J}_{\bar{c}s}-d^{A}_{v\bar{a}}h_{A}^{w}b_{It\bar{b}}d^{I}_{\bar{c}w}\nonumber\\
&=-\frac{1}{4}f_{\bar{a}\bar{b}}{}^{s}b_{Jt v}d^{J}_{\bar{c}s}-\frac{5}{12}f_{\bar{a}\bar{b}}{}^{s}b_{Jt\bar{c}}d^{J}_{sv}+\frac{1}{6}f_{\bar{a}v}{}^{s}b_{It \bar{b}}d^{I}_{\bar{c}s}-\frac{1}{3}g^{Js}b_{Jt\bar{a}}b_{Is\bar{b}}d^{I}_{\bar{c}v}
\nonumber
\end{align}
where we used the first equation of (\ref{lol}) on the last term of the first line. If we multiply everything by $-4$, we obtain:
\begin{align}
&f_{\bar{a}\bar{b}}{}^{s}b_{Jt v}d^{J}_{\bar{c}s}+\frac{5}{3}f_{\bar{a}\bar{b}}{}^{s}b_{Jt\bar{c}}d^{J}_{sv}-\frac{2}{3}f_{\bar{a}v}{}^{s}b_{It \bar{b}}d^{I}_{\bar{c}s}+\frac{4}{3}g^{Js}b_{Jt\bar{a}}b_{Is\bar{b}}d^{I}_{\bar{c}v}\nonumber\\
&=2g^{Js}b_{Jt\bar{a}}b_{Is\bar{b}}d^{I}_{\bar{c}v}+f_{\bar{a}\bar{c}}{}^{s}b_{It s}d^{I}_{v\bar{b}}-2d^{J}_{\bar{a}s}b_{Jt\bar{c}}h_{I}^{s}d^{I}_{v\bar{b}}\nonumber\\
&=f_{t a}{}^{s}d^{I}_{v\bar{b}}b_{Is\bar{c}}-g^{Js}b_{Jt \bar{a}}d^{I}_{v\bar{b}}b_{Is\bar{c}}+d^{J}_{t\bar{a}}h_{J}^{s}b_{Is\bar{c}}d^{I}_{v\bar{b}}\nonumber\\
&=k_{t}^{\alpha}c_{\alpha \bar{a}}^{s}d_{v\bar{b}}^{I}b_{Is\bar{c}}
\;,\nonumber
\end{align}
where we used the first and eigth equation of (\ref{lol}) to pass from the first line to the second line, and the second equation of (\ref{lol}) to pass from the second line to the third line. Then we have proven that the left hand side of (\ref{bob12}) is proportional to
\begin{equation}
-\frac{1}{4}k_{t}^{\lambda}c_{\lambda \bar{a}}^{s}d_{v\bar{b}}^{I}b_{Is\bar{c}}h_{B}^{v}-\frac{1}{2}k_{t}^{\lambda}c_{\lambda BA}d^{A}_{s[\bar{a}}f^{s}_{\bar{b}\bar{c}]}\;,
\end{equation}
which finishes the proof.
\end{proof}

\section{Truncation of a Lie 3-algebra to a Lie 2-algebra} 
\label{AppendixB}
Here we study the truncation of a Lie 3-algebra $(\C{U}_\bullet\equiv \C{U}_{-2} \oplus \C{U}_{-1} \oplus \C{U}_{0}, l_1, l_2, l_3, l_4)$  to a Lie 2-algebra  $(\widetilde{\C{U}}_\bullet\equiv \widetilde{\C{U}}_{-1} \oplus \C{U}_{0}, \widetilde{l}_1, \widetilde{l}_2, \widetilde{l}_3)$, where $\widetilde{\C{U}}_{-1}=\bigslant{\C{U}_{-1}}{\R{Ker}(t_{(1)})}$, and we will apply the result to the explicit example given in section $\ref{structural}$. Doing so, we will show that the naive truncation of $\C{U}$ down to $\widehat{U}=\C{U}_{-1}\oplus\C{U}_{0}$ does not satisfy the axioms of a Lie 2-algebra. First, observe that the Jacobi identities for $m=1,2,3$ involving only degree zero elements are unchanged for degree reasons (because they are of degree zero). Second, the Jacobi identities of degree strictly less than $-1$ automatically vanish on either truncations. Similarly, any identity involving a degree $-2$ element should not be taken into account when performing the truncation because there are no degree $-2$ elements in the truncated algebras. Thus let us turn our attention to the remaining Jacobi identities which are of degree $-1$ (here $u_{i}\in\C{U}_{0}$ and $v_{j}\in\C{U}_{-1}$):
\begin{align}
	&[[v_{1}]_{1},v_{2}]_{2}+[[v_{2}]_{1},v_{1}]_{2}=[[v_{1},v_{2}]_{2}]_{1}\;,\\
	&[[u_{1},u_{2}]_{2},v]_{2}+[[u_{2},v]_{2},u_{1}]_{2}-[[u_{1},v]_{2},u_{2}]_{2}+[[v]_{1},u_{1},u_{2}]_{3}=-[[u_{1},u_{2},v]_{3}]_{1}\;,\nonumber\\[1ex]
	\begin{split}
&\sum_{\sigma\in Un(4,2)}\chi(\sigma)\big[[u_{\sigma(1)},u_{\sigma(2)}]_{2},u_{\sigma(3)},u_{\sigma(4)}\big]_{3}\nonumber\\
&\hspace{1cm}-\sum_{\sigma\in Un(4,3)}\chi(\sigma)\big[[u_{\sigma(1)},u_{\sigma(2)},u_{\sigma(3)}]_{3},u_{\sigma(4)}
\big]_{2}
=-\big[[u_{\sigma(1)},u_{\sigma(2)}u_{\sigma(3)},u_{\sigma(4)}]_{4}
\big]_{1}\;,
	\end{split}
\end{align}
where we did not write terms involving $[u_{i}]_{1}$ because they vanish by definition. In each of these equations, we observe that the left hand side is the expected Jacobi identity for a Lie 2-algebra. On the right hand side however, we find some terms (all of degree $-1$) which originally appear in the corresponding Jacobi identities for the Lie 3-algebra. All of these terms lie in the image of $[\cdot]_{1}=t_{(2)}$ and as such, in the kernel of $t_{(1)}$, so they all vanish if the Lie 3-algebra $\C{U}$ is truncated down to $\widetilde{\C{U}}$, which renders it a Lie 2-algebra since all Jacobi identities are satisfied. However, if $\C{U}$ is naively truncated to $\widehat{U}$, the right hand sides project as such without modification, which implies that the Jacobi identities on $\widehat{U}$ are not fully satisfied, and thus $\widehat{U}$ cannot be a Lie 2-algebra. Note that we could have truncated the Lie 3-algebra down to the following $\bigslant{\C{U}_{-1}}{\R{Im}(t_{(2)})}\oplus\C{U}_{0}$ and obtain a Lie 2-algebra structure on it, given that all the right hand sides in the above equations vanish when projected on the quotient.

Let us now explain how to truncate the Lie 3-algebra $\C{U}\equiv\widetilde{\C{V}}$ given in section $\ref{structural}$ down to the following Lie 2-algebra:
\begin{equation}\label{Vker2}
\widetilde{\C{V}}' : =\bigslant{\C{V}_{-1}}{\R{Ker}(t)}\stackrel{\widetilde{t}}{\to} \C{V}_0 \,.
\end{equation}
To give the explicit form of the brackets on $\widetilde{\C{V}}'$, we need to follow the same recipe presented after Equation $(\ref{basischoice})$. Let $\widehat{\C{W}}'$ be the following graded N-vector space:
\begin{equation}\label{Whatbis}
\widehat{\C{W}}'_\bullet : = \C{W}_{-1}[1]\stackrel{}{\to}\C{W}_{-1} \, ,
\end{equation}

As in the discussion in subsection $\ref{structural}$, we will choose a set of coordinates on $\widehat{\C{W}}'$ adapted to the truncation. Using the fact that $\C{W}_{-1}$ is isomorphic to $\R{Im}(\llbracket.\rrbracket_{1}) \oplus\bigslant{\C{W}_{-1}}{\R{Im}(\llbracket.\rrbracket_{1})}$, let us first take a set of coordinates $\{\widehat{q}'^{a}\}$ for $\R{Im}(\llbracket.\rrbracket_{1})\subset\C{W}_{-1}$, which we complete by some homogeneous degree 1 elements $\{\widehat{q}'^{r}\}$ into a set of independent coordinates $\{\widehat{q}'^{k}\}=\{\widehat{q}'^{a},\widehat{q}'^{r}\}$ on $\C{W}_{-1}$. As in section $\ref{structural}$, the letters from the beginning of the alphabet (resp. the end) are used to emphasize the fact that the set of elements labelled by these letters is a basis of $\R{Im}(\llbracket.\rrbracket_{1})$ (resp. a supplementary subspace isomorphic to the quotient $\bigslant{\C{W}_{-1}}{\R{Im}(\llbracket.\rrbracket_{1})}$). Their suspended counterpart $\{\widehat{\widehat{q}}'^{k}\}=\{\widehat{\widehat{q}}'^{a},\widehat{\widehat{q}}'^{r}\}$ give a set of coordinates for $\C{W}_{-1}[1]$ adapted to the truncation. Let $\phi':\widehat{\C{W}}\longrightarrow\widehat{\C{W}}'$ be the degree preserving map implicitely defined by:
\begin{align}
	\phi'^{\ast}(\widehat{q}'^{k})				&=\widehat{q}^{k}\,,\nonumber\\
	\phi'^{\ast}(\widehat{\widehat{q}}'^{a})	&=-h^{a}_{I}\widehat{q}^{I}\,,\\
	\phi'^{\ast}(\widehat{\widehat{q}}'^{r})	&=0\nonumber \;.
\end{align}
For the same arguments presented after Equation $(\ref{tildetilde})$ and in subsection \ref{structural}, the map $\phi'$ acts as an isomorphism between $\R{Im}(\llbracket.\rrbracket_{1})[1]^{\ast}$ and $\R{Im}(Q_{(1)}\big|_{\C{W}_{-1}^{\ast}})$. Thus it is fit for the truncation of the NQ-manifold $\widetilde{\C{W}}$ associated to the Lie 3-algebra $\widetilde{\C{V}}$ down to the graded manifold $\widetilde{\C{W}}'=\R{Im}(\phi')$
\begin{equation}\label{Wker2}
\widetilde{\C{W}}'_\bullet  \cong \bigslant{\C{W}_{-2}}{\R{Ker}(\llbracket.\rrbracket_{1})}\stackrel{\llbracket.\rrbracket_{1}}{\to}\C{W}_{-1} \, , 
\end{equation}
associated with the truncated Lie 2-algebra $\widetilde{\C{V}}'$ (see Equation $(\ref{Vker2})$).

	Following Theorem \ref{propositionhenning}, the graded vector space $\R{Im}(\phi')\subset\widehat{\C{W}}'$ is a NQ-manifold when equipped with:
\begin{equation}
	\widetilde{Q}' = (\frac{1}{2}f_{kl}{}^{a}\widehat{q}'^{k}\widehat{q}'^{l}	 +\widehat{\widehat{q}}'^{a})\frac{\partial}{\partial \widehat{q}'^{a}}+\frac{1}{2}f_{kl}{}^{s}\widehat{q}'^{k}\widehat{q}'^{l}\frac{\partial}{\partial \widehat{q}'^{s}}
		+(f_{ka}{}^{b}\widehat{q}'^{k}\widehat{\widehat{q}}'^{a}+\frac{1}{6}h^{b}_{I}d^{I}_{s[k}f_{lm]}{}^{s}\widehat{q}'^{k}\widehat{q}'^{l}\widehat{q}'^{m})\frac{\partial}{\partial \widehat{\widehat{q}}'^{b}}
		\;.
\end{equation}
The Lie 2-algebra structure associated to $\widetilde{\C{V}}'$ can be best described using explicit formulas. To do this, we note by $v'^{k}$ and $w'^{a}$ the coordinates on $\widetilde{\C{V}}'$ associated to $\widehat{q}'^{k}$ and $\widehat{\widehat{q}}'^{a}$ respectively. Then the Lie 2-algebra structure is given by the following brackets (obtained from equations (\ref{crochetsbrackets}, \ref{derivedbracket}) applied to the homological vector field $\widetilde{Q}'$:
\begin{align}
	[w'_{a}]_{1}							&= v'_{a}\;,\nonumber\\
	[v'_{k},v'_{l}]_{2}				&= f_{kl}{}^{m}v'_{m}\;,\nonumber\\
	[v'_{k},w'_{a}]_{2}				&= f_{ka}{}^{b}w'_{b}\;,\nonumber\\
	[v'_{k},v'_{l},v'_{m}]_{3}		&= h^{a}_{I}d^{I}_{s[k}f^{\phantom{}}_{lm]}{}^{s}w'_{a}\;.
\end{align}

\section{Lift of a Lie \texorpdfstring{$n$}{n}-algebra to a Lie \texorpdfstring{$(n+1)$}{(n+1)}-algebra} 
\label{AppendixC}
There is always a trivial way of viewing a Lie $n$-algebra as a Lie $m$-algebra for an $m>n$, just by extending the original complex of length $n$ by zero-dimensional vector spaces to a complex of length $m$ and by adding higher brackets that all vanish. This is not what we will discuss in the present appendix. Instead, there is a canonical lift of any NQ-manifold of degree $n$ to an NQ-manifold of degree $n+1$ by means of the construction presented first in \cite{Kotov:2007nr} and used in the present paper for describing the closed gauge symmetries as inner vertical automorphisms. It is the purpose of the present appendix to translate this into the original language of $L_
\infty$-algebras. Afterwards we illustrate the procedure for an ordinary Lie algebra, lifting it canonically to a Lie 2- and Lie 3-algebra.

Given an NQ-manifold $(\C{W},Q)$ of degree $n$ with coordinates $q^{\alpha}$ (corresponding thus to a Lie $n$-algebra following the discussion in section \ref{section4}), the shifted tangent bundle $(T[1]\C{W},\bar{Q}=\C{L}_{Q}+\bar{\dd})$ is an NQ-manifold of degree $n+1$, with canonical coordinates $q^{\alpha}$ and $\bar{\dd}q^{\alpha}$. This is isomorphic to the direct sum of $\C{W}$ and of $\C{W}[1]$ with coordinates $q^{\alpha}$ and $\bar{q}^{\alpha}\equiv [1]q^{\alpha}$, respectively. Since $\C{L}_{Q}=Q^{\alpha}\frac{\partial}{\partial q^{\alpha}}-\bar{\dd}Q^{\alpha}\frac{\partial}{\partial \bar{\dd}q^{\alpha}}$, this direct sum $\bar{\C{W}}=\C{W}\oplus\C{W}[1]$ becomes a $L_{n+1}[1]$ algebra, when equipped with the following brackets (following Equation \eqref{derivedbracket}):
\begin{align}
\llbracket q_{\alpha_{1}},\ldots,q_{\alpha_{j}}\rrbracket_{j}&=C^{\beta}_{\alpha_{1}\ldots \alpha_{j}}q_{\beta}\nonumber\\
\llbracket \bar{q}_{\alpha}\rrbracket_{1} &= -C_{\alpha}^{\beta}\bar{q}_{\beta}+q_{\alpha}\\
\forall \ j\geq 2\hspace{1cm}\llbracket q_{\alpha_{1}},\ldots,q_{\alpha_{k-1}},\bar{q}_{\alpha_{k}},q_{\alpha_{k+1}},\ldots,q_{\alpha_{j}}\rrbracket_{j}&=-(-1)^{\sum_{i=1}^{k-1}|\alpha_{i}|}C^{\beta}_{\alpha_{1}\ldots \alpha_{j}}\bar{q}_{\beta}\nonumber
\end{align}
where $C^{\beta}_{\alpha_{1}\ldots \alpha_{j}}$ has been defined in Equation \eqref{Qsum}. All other brackets vanish. The unbarred contribution $q_{\alpha}$ in the second line comes from the presence of $\bar{\dd}$ in $\bar{Q}$, whereas the last line comes from the term $-\bar{\dd}Q^{\alpha}\frac{\partial}{\partial \bar{\dd}q^{\alpha}}$ in $\bar{Q}$. One can reformulate these equation in terms of the $j$-ary brackets $\llbracket\ldots\rrbracket^{\C{W}}_{j}$ on $\C{W}$ and the shift functors $[1]:\C{W}[1]\longrightarrow\C{W}$ and $[-1]:\C{W}\longrightarrow\C{W}[1]$:
\begin{align}\label{pouf}
\llbracket q_{\alpha_{1}},\ldots,q_{\alpha_{j}}\rrbracket_{j}&=\llbracket q_{\alpha_{1}},\ldots,q_{\alpha_{j}}\rrbracket^{\C{W}}_{j}\nonumber\\
\llbracket y\rrbracket_{1} &= -[-1]\llbracket [1]y\rrbracket^{\C{W}}_{1}+[1]y\nonumber\\
\forall \ j\geq 2 \,:\hspace{5cm}& \nonumber\\[-1ex]
\llbracket q_{\alpha_{1}},\ldots,q_{\alpha_{k-1}},y,q_{\alpha_{k+1}},\ldots,q_{\alpha_{j}}\rrbracket_{j}&=-(-1)^{\sum_{i=1}^{k-1}|\alpha_{i}|}[-1]\llbracket q_{\alpha_{1}},\ldots,q_{\alpha_{k-1}},[1]y,q_{\alpha_{k+1}},\ldots,q_{\alpha_{j}}\rrbracket_{j}\nonumber\\
\end{align}
for all $q_{\alpha_{i}}\in\C{W}$ and $y\in\C{W}[1]$. If we define $\C{V}$ as the desuspended vector space associated to $\C{W}$ as in section \ref{section4}, then one could equip $\widetilde{\C{V}}\equiv\C{V}\oplus\C{V}[1]$ with an Lie $(n+1)$-algebra structure, using the obvious generalization of \eqref{crochetsbrackets} to the present case.
	
Let us now apply this discussion to the case of a Lie algebra $\mathfrak{g}$, with generators $t_{a}$. The associated NQ-manifold is $(\F{g}[1],Q_{\F{g}}=\frac{1}{2}f_{ab}{}^{c}\xi^{a}\xi^{b}\frac{\partial}{\partial\xi^{c}})$ where $f_{ab}{}^{c}$ is a set of structure constant for $\F{g}$ and $\xi^{a}$ are local coordinates on $\F{g}[1]$. It endows $\F{g}[1]$ with a degree $+1$ bracket $\llbracket \cdot,\cdot\rrbracket_{\F{g}[1]}$, which is nothing but the shifted version of the traditional bracket $[\cdot,\cdot]_{\F{g}}$ on $\F{g}$. Now, let us follow the above construction: $(T[1](\F{g}[1])$ equipped with
\begin{equation}\label{liealgQ}
\bar{Q}=\C{L}_{Q_{\F{g}}}+\bar{\dd}=\frac{1}{2}f_{ab}{}^{c}\xi^{a}\xi^{b}\frac{\partial}{\partial\xi^{c}}-f_{ab}{}^{c}\bar{\dd}\xi^{a}\xi^{b}\frac{\partial}{\partial\bar{\dd}\xi^{c}}+\bar{\dd}\xi^{a}\frac{\partial}{\partial \xi^{a}}
\;,
\end{equation}
is a NQ-manifold of degree 2 isomorphic to $\bar{\F{g}}\equiv\F{g}[1]\oplus\F{g}[2]$, with canonical shift functors $[1]:\F{g}[2]\longrightarrow\F{g}[1]$ and $[-1]:\F{g}[1]\longrightarrow\F{g}[2]$. Thus, using equations \eqref{pouf}, \eqref{liealgQ} we obtain the following brackets on $\bar{\F{g}}$:
\begin{align}\label{pouf2}
\llbracket y\rrbracket_{1} &= [1]y\\
\llbracket x,x'\rrbracket_{2}&=\llbracket x,x'\rrbracket_{\F{g}[1]}\nonumber\\
\llbracket x,y\rrbracket_{2}&=[-1]\llbracket x,[1]y\rrbracket_{\F{g}[1]}\nonumber\\
\llbracket y,y' \rrbracket_{2}&=0\nonumber
\end{align}
for all $x,x'\in\F{g}[1]$ and $y,y'\in\F{g}[2]$. 
Using Equation \eqref{crochetsbrackets} one can equip the desuspended graded vector space $\widetilde{\F{g}}\equiv\F{g}\oplus\F{g}[1]$ with the following brackets:
\begin{align}\label{pouf3}
[ v]_{1} &= [1]v\\
[ u,u']_{2}&=[ u,u']_{\F{g}}\nonumber\\
[ u,v]_{2}&=[-1][ u,[1]v]_{\F{g}}\nonumber\\
[ v,v' ]_{2}&=0\nonumber
\end{align}
for all $u,u'\in\F{g}$ and $v,v'\in\F{g}[1]$. Since we know that $\bar{Q}$ squares to zero, the generalized Jacobi identities are satisfied, turning $\widetilde{\F{g}}$ into a Lie 2-algebra, which is equivalent to a differential crossed module (see for example \cite{Baez:2003fs}). 

We can further apply the process to $T[1](\F{g}[1])$ itself to obtain a Lie 3-algebra from $\widetilde{g}$: Let $\C{M}=T[1]\left(T[1](\F{g}[1])\right)$ be the shifted tangent bundle to $T[1](\F{g}[1])$ with local coordinates $\xi^{a}, \bar{\dd}\xi^{\alpha}$ and $D\xi^{\alpha},D\bar{\dd}\xi^{\alpha}$ of degree $1,2$ and $2,3$, respectively -- where $D$ is a notation for the de Rham differential on $T[1](\F{g}[1])$. $\C{M}$ becomes a NQ-manifold when equipped with the following homological vector field:
\begin{align}\label{liealgQ2}
\bar{\bar{Q}}&=\C{L}'_{\C{L}_{Q_{\F{g}}}+\bar{\dd}}+D\\
&=\frac{1}{2}f_{ab}{}^{c}\xi^{a}\xi^{b}\frac{\partial}{\partial\xi^{c}}-f_{ab}{}^{c}\bar{\dd}\xi^{a}\xi^{b}\frac{\partial}{\partial\bar{\dd}\xi^{c}}-f_{ab}{}^{c}D\xi^{a}\xi^{b}\frac{\partial}{\partial D\xi^{c}}+f_{ab}{}^{c}D\bar{\dd}\xi^{a}\xi^{b}\frac{\partial}{\partial D\bar{\dd}\xi^{c}}\nonumber\\
&\hspace{1cm}+f_{ab}{}^{c}\bar{\dd}\xi^{a}D\xi^{b}\frac{\partial}{\partial D\bar{\dd}\xi^{c}}+\bar{\dd}\xi^{a}\frac{\partial}{\partial \xi^{a}}-D\bar{\dd}\xi^{a}\frac{\partial}{\partial D\xi^{a}}+D\xi^{a}\frac{\partial}{\partial \xi^{a}}+D\bar{\dd}\xi^{a}\frac{\partial}{\partial \bar{\dd}\xi^{a}}\;,\nonumber
\end{align}
where $\C{L}'$ is the Lie derivative on $\C{M}$. As a graded manifold, $\C{M}$ is isomorphic to the direct sum $\bar{\bar{\F{g}}}=\F{g}[1]\oplus\F{g}[2]_{l}\oplus\F{g}[2]_{r}\oplus\F{g}[3]$, where $l$ and $r$ stand for \emph{left} and \emph{right}, respectively. $\F{g}[2]_{r}$ is associated to the degree 2 subspace of $\C{M}$ with coordinates $\bar{\dd}\xi^{a}$, whereas $\F{g}[2]_{l}$ is associated to the degree 2 subspace of $\C{M}$ with coordinates $D\xi^{a}$. Thus there exists two types of shift functors relating the left and right subalgebras $\F{g}[2]_{l}$ and $\F{g}[2]_{r}$ to $\F{g}[1]$ and $\F{g}[3]$ in the following commutative diagram:
\begin{align}
&\F{g}[1]\quad\overset{[-1]_{r}}{\xrightarrow{\hspace*{2cm}}}\quad\F{g}[2]_{r}\nonumber\\
[-1]_{l}&\Big\downarrow\hspace{3.6cm}\Big\downarrow [-1]_{r}\\
&\F{g}[2]_{l}\quad\overset{[-1]_{l}}{\xrightarrow{\hspace*{2cm}}}\quad\F{g}[3]\nonumber
\end{align}
and their respective inverse:
\begin{align}
&\F{g}[1]\quad\overset{[1]_{r}}{\xleftarrow{\hspace*{2cm}}}\quad\F{g}[2]_{r}\nonumber\\
[1]_{l}&\Big\uparrow\hspace{3.6cm}\Big\uparrow [1]_{r}\\
&\F{g}[2]_{l}\quad\overset{[1]_{l}}{\xleftarrow{\hspace*{2cm}}}\quad\F{g}[3]\nonumber
\end{align}
Using equations \eqref{pouf}, and \eqref{liealgQ2}, one obtains the following brackets on $\bar{\bar{\F{g}}}$:
\begin{align}
\llbracket y\rrbracket_{1} &= [1]_{r}y\\
\llbracket w\rrbracket_{1} &= [1]_{l}w\nonumber\\
\llbracket z\rrbracket_{1} &= [1]_{r}z-[1]_{l}z\nonumber\\
\llbracket x,x'\rrbracket_{2}&=\llbracket x,x'\rrbracket_{\F{g}[1]}\nonumber\\
\llbracket x,y\rrbracket_{2}&=[-1]_{r}\llbracket x,[1]_{r}y\rrbracket_{\F{g}[1]}\nonumber\\
\llbracket x,w\rrbracket_{2}&=[-1]_{l}\llbracket x,[1]_{l}w\rrbracket_{\F{g}[1]}\nonumber\\
\llbracket x,z\rrbracket_{2}&=[-2]\llbracket x,[2]z\rrbracket_{\F{g}[1]}\nonumber\\
\llbracket y,w \rrbracket_{2}&=[-2]\llbracket [1]_{r}y,[1]_{l}w\rrbracket_{\F{g}[1]}\nonumber
\end{align}
for all $x,x'\in\F{g}[1]$, $y\in\F{g}[2]_{r}$, $w\in\F{g}[2]_{l}$ and $z\in\F{g}[3]$, all other brackets being zero. From this set of equations and Equation\eqref{crochetsbrackets}, one deduces the brackets on the desuspended version of $\bar{\bar{g}}$, denoted by $\widetilde{\widetilde{g}}\equiv\F{g}\oplus\F{g}[1]_{l}\oplus\F{g}[1]_{r}\oplus\F{g}[2]$ (with canonical shift functors as above):
\begin{align}\label{pouf3b}
[ v]_{1} &= [1]_{r}v\\
[ r]_{1} &= [1]_{l}r\nonumber\\
[ s]_{1} &= [1]_{r}s-[1]_{l}s\nonumber\\
[ u,u']_{2}&=[ u,u']_{\F{g}}\nonumber\\
[ u,v]_{2}&=[-1]_{r}[ u,[1]_{r}v]_{\F{g}}\nonumber\\
[ u,r]_{2}&=[-1]_{l}[ u,[1]_{l}r]_{\F{g}}\nonumber\\
[ u,s]_{2}&=[-2][ u,[2]s]_{\F{g}}\nonumber\\
[ v,r ]_{2}&=-[-2][ [1]_{r}v,[1]_{l}r]_{\F{g}}\nonumber
\end{align}
for all $u,u'\in\F{g}$, $v\in\F{g}[1]_{r}$, $r\in\F{g}[1]_{l}$ and $s\in\F{g}[2]$, all other brackets being zero. Knowing that $\bar{\bar{Q}}$ squares to zero, it ensures that the Jacobi identities are satisfied, thus endowing $\widetilde{\widetilde{\F{g}}}$ with a Lie 3-algebra structure.



\begin{thebibliography}{CDRVP10}
\providecommand{\url}[1]{{\def~{{\textasciitilde}}\texttt{#1}}}
\providecommand{\urlprefix}{}
\expandafter\ifx\csname urlstyle\endcsname\relax
  \providecommand{\doi}[1]{doi:\discretionary{}{}{}#1}\else
  \providecommand{\doi}{doi:\discretionary{}{}{}\begingroup
  \urlstyle{rm}\Url}\fi
\providecommand{\eprint}[2][]{\texttt{#2}}

\bibitem[AP12]{Akyol:2012cq}
\textsc{M.~Akyol} and \textsc{G.~Papadopoulos}: \emph{{$(1,0)$ superconformal
  theories in six dimensions and Killing spinor equations}}, JHEP, \textbf{vol.
  1207}, (2012) 070, \doi{10.1007/JHEP07(2012)070}. \eprint{1204.2167}.

\bibitem[Bae02]{Baez:2002jn}
\textsc{J.~C. Baez}: \emph{Higher {Y}ang-{M}ills theory}.
  \eprint{hep-th/0206130}.

\bibitem[BC04]{Baez:2003fs}
\textsc{J.~C. Baez} and \textsc{A.~S. Crans}: \emph{Higher-dimensional algebra
  {VI}: {L}ie 2-algebras}, Theor.Appl.Categor., \textbf{vol.~12}, (2004)
  492--528. \eprint{math/0307263}.

\bibitem[BGH13]{Bonetti:2012st}
\textsc{F.~Bonetti}, \textsc{T.~W. Grimm}, and \textsc{S.~Hohenegger}:
  \emph{Non-abelian tensor towers and $(2,0)$ superconformal theories}, JHEP,
  \textbf{vol. 1305}, (2013) 129, \doi{10.1007/JHEP05(2013)129}.
  \eprint{1209.3017}.

\bibitem[BHH{\etalchar{+}}09]{Bergshoeff:2009ph}
\textsc{E.~A. Bergshoeff}, \textsc{J.~Hartong}, \textsc{O.~Hohm},
  \textsc{M.~H\"ubscher}, and \textsc{T.~Ortin}: \emph{Gauge theories, duality
  relations and the tensor hierarchy}, JHEP, \textbf{vol.~04}, (2009) 123,
  \doi{10.1088/1126-6708/2009/04/123}. \eprint{arXiv:0901.2054}.

\bibitem[BKS05]{Bojowald:2004wu}
\textsc{M.~Bojowald}, \textsc{A.~Kotov}, and \textsc{T.~Strobl}: \emph{Lie
  algebroid morphisms, {P}oisson sigma models, and off-shell closed gauge
  symmetries}, J.Geom.Phys., \textbf{vol.~54}, (2005) 400--426,
  \doi{10.1016/j.geomphys.2004.11.002}. \eprint{math/0406445}.

\bibitem[BM05]{Breen:2001ie}
\textsc{L.~Breen} and \textsc{W.~Messing}: \emph{Differential geometry of
  gerbes}, Adv.Math., \textbf{vol. 198}, (2005) 732. \eprint{math/0106083}.

\bibitem[BSS08]{Bergshoeff:2007ef}
\textsc{E.~Bergshoeff}, \textsc{H.~Samtleben}, and \textsc{E.~Sezgin}:
  \emph{The gaugings of maximal ${D}=6$ supergravity}, JHEP, \textbf{vol.~03},
  (2008) 068, \doi{10.1088/1126-6708/2008/03/068}. \eprint{0712.4277}.

\bibitem[BSS13]{Bandos:2013jva}
\textsc{I.~Bandos}, \textsc{H.~Samtleben}, and \textsc{D.~Sorokin}:
  \emph{{Duality-symmetric actions for non-Abelian tensor fields}}, Phys.Rev.,
  \textbf{vol. D88}, (2013) 025024, \doi{10.1103/PhysRevD.88.025024}.
  \eprint{1305.1304}.

\bibitem[BV81]{Batalin:1981jr}
\textsc{I.~Batalin} and \textsc{G.~Vilkovisky}: \emph{Gauge algebra and
  quantization}, Phys.Lett., \textbf{vol. B102}, (1981) 27--31,
  \doi{10.1016/0370-2693(81)90205-7}.

\bibitem[CDRVP10]{Coomans:2010xd}
\textsc{F.~Coomans}, \textsc{J.~De~Rydt}, and \textsc{A.~Van~Proeyen}:
  \emph{Generalized gaugings and the field-antifield formalism}, JHEP,
  \textbf{vol. 1003}, (2010) 105, \doi{10.1007/JHEP03(2010)105}.
  \eprint{1001.2560}.

\bibitem[CF01a]{Cattaneo:2001bp}
\textsc{A.~S. Cattaneo} and \textsc{G.~Felder}: \emph{{P}oisson sigma models
  and deformation quantization}, Mod.Phys.Lett., \textbf{vol. A16}, (2001)
  179--190, \doi{10.1142/S0217732301003255}. \eprint{hep-th/0102208}.

\bibitem[CF01b]{Cattaneo:2000iw}
--- \emph{Poisson sigma models and symplectic groupoids}, in \emph{Quantization
  of singular symplectic quotients}, vol. 198 of \emph{Progr. Math.}, pp.
  61--93 (Birkh\"auser, Basel, 2001).

\bibitem[CF03]{Crainic-Fernandes}
\textsc{M.~Crainic} and \textsc{R.~L. Fernandes}: \emph{Integrability of {L}ie
  brackets}, Ann. of Math. (2), \textbf{vol. 157(2)}, (2003) 575--620, ISSN
  0003-486X, \doi{10.4007/annals.2003.157.575}.
  \urlprefix\url{http://dx.doi.org/10.4007/annals.2003.157.575}.

\bibitem[CK12]{Chu:2012um}
\textsc{C.-S. Chu} and \textsc{S.-L. Ko}: \emph{{Non-abelian Action for
  Multiple Five-Branes with Self-Dual Tensors}}, JHEP, \textbf{vol. 1205},
  (2012) 028, \doi{10.1007/JHEP05(2012)028}. \eprint{1203.4224}.

\bibitem[DDF05]{DallAgata:2005mj}
\textsc{G.~Dall'Agata}, \textsc{R.~D'Auria}, and \textsc{S.~Ferrara}:
  \emph{Compactifications on twisted tori with fluxes and free differential
  algebras}, Phys.Lett., \textbf{vol. B619}, (2005) 149--154,
  \doi{10.1016/j.physletb.2005.04.005}. \eprint{hep-th/0503122}.

\bibitem[DFT06]{DAuria:2005er}
\textsc{R.~D'Auria}, \textsc{S.~Ferrara}, and \textsc{M.~Trigiante}:
  \emph{${E}_{7(7)}$ symmetry and dual gauge algebra of {M}-theory on a twisted
  seven-torus}, Nucl. Phys., \textbf{vol. B732}, (2006) 389--400.
  \eprint{hep-th/0504108}.

\bibitem[GHP13]{Greitz:2013pua}
\textsc{J.~Greitz}, \textsc{P.~Howe}, and \textsc{J.~Palmkvist}: \emph{{The
  tensor hierarchy simplified}}. \eprint{1308.4972}.

\bibitem[GS05]{Gruetzmann-Strobl}
\textsc{M.~Gr\"utzmann} and \textsc{T.~Strobl}: \emph{General Yang-Mills type gauge theories for p-form gauge fields: From physics-based ideas to a mathematical framework OR From Bianchi identities to twisted Courant algebroids}. To be published  in  the International Journal of Geometric Methods in Modern Physics.
  \eprint{1407.6759}

\bibitem[HHM11]{Ho:2011ni}
\textsc{P.-M. Ho}, \textsc{K.-W. Huang}, and \textsc{Y.~Matsuo}: \emph{A
  non-abelian self-dual gauge theory in 5+1 dimensions}, JHEP,
  \textbf{vol.~07}, (2011) 021, \doi{10.1007/JHEP07(2011)021}.
  \eprint{1104.4040}.

\bibitem[HT92]{Henneaux:1992ig}
\textsc{M.~Henneaux} and \textsc{C.~Teitelboim}: \emph{Quantization of gauge
  systems}, xxviii+520 pp. (Princeton University Press, Princeton, NJ, 1992),
  ISBN 0-691-08775-X; 0-691-03769-8.

\bibitem[Ike01]{Ikeda:2001fq}
\textsc{N.~Ikeda}: \emph{Deformation of {BF} theories, topological open
  membrane and a generalization of the star deformation}, JHEP, \textbf{vol.
  0107}, (2001) 037. \eprint{hep-th/0105286}.

\bibitem[Ike03]{Ikeda:2002wh}
--- \emph{{C}hern-{S}imons gauge theory coupled with {BF} theory},
  Int.J.Mod.Phys., \textbf{vol. A18}, (2003) 2689--2702,
  \doi{10.1142/S0217751X03015155}. \eprint{hep-th/0203043}.

\bibitem[Kon03]{Kontsevich:1997vb}
\textsc{M.~Kontsevich}: \emph{Deformation quantization of {P}oisson manifolds},
  Lett.Math.Phys., \textbf{vol.~66}, (2003) 157--216. \eprint{q-alg/9709040}.

\bibitem[KS07]{Kotov:2007nr}
\textsc{A.~Kotov} and \textsc{T.~Strobl}: \emph{Characteristic classes
  associated to {Q}-bundles}. To be published  in  the International Journal of Geometric Methods in Modern Physics. \eprint{0711.4106}.

\bibitem[KS10]{Kotov:2010wr}
--- \emph{Generalizing geometry---algebroids and sigma models}, in
  \emph{Handbook of pseudo-{R}iemannian geometry and supersymmetry}, vol.~16 of
  \emph{IRMA Lect. Math. Theor. Phys.}, pp. 209--262 (Eur. Math. Soc.,
  Z\"urich, 2010), \doi{10.4171/079-1/7}.
  \urlprefix\url{http://dx.doi.org/10.4171/079-1/7}.
  \eprint{1004.0632}

\bibitem[KSS14]{Kotov:xxx}
\textsc{A.~Kotov}, \textsc{H.~Samtleben}, and \textsc{T.~Strobl}: Work in
  progress.

\bibitem[LM95]{Lada:1994mn}
\textsc{T.~Lada} and \textsc{M.~Markl}: \emph{Strongly homotopy {L}ie
  algebras}, Comm. Algebra, \textbf{vol.~23(6)}, (1995) 2147--2161, ISSN
  0092-7872, \doi{10.1080/00927879508825335}.
  \urlprefix\url{http://dx.doi.org/10.1080/00927879508825335}.

\bibitem[LM02]{Louis:2002ny}
\textsc{J.~Louis} and \textsc{A.~Micu}: \emph{Type {II} theories compactified
  on {C}alabi-{Y}au threefolds in the presence of background fluxes}, Nucl.
  Phys., \textbf{vol. B635}, (2002) 395--431. \eprint{hep-th/0202168}.

\bibitem[LP10]{Lambert:2010wm}
\textsc{N.~Lambert} and \textsc{C.~Papageorgakis}: \emph{Nonabelian $(2,0)$
  tensor multiplets and 3-algebras}, JHEP, \textbf{vol. 1008}, (2010) 083,
  \doi{10.1007/JHEP08(2010)083}. \eprint{arXiv:1007.2982}.

\bibitem[LS93]{Lada:1992wc}
\textsc{T.~Lada} and \textsc{J.~Stasheff}: \emph{Introduction to {SH} {L}ie
  algebras for physicists}, Int.J.Theor.Phys., \textbf{vol.~32}, (1993)
  1087--1104, \doi{10.1007/BF00671791}. \eprint{hep-th/9209099}.

\bibitem[MS09]{Mayer:2009wf}
\textsc{C.~Mayer} and \textsc{T.~Strobl}: \emph{Lie algebroid {Y}ang-{M}ills
  with matter fields}, J.Geom.Phys., \textbf{vol.~59}, (2009) 1613--1623,
  \doi{10.1016/j.geomphys.2009.07.018}. \eprint{0908.3161}.

\bibitem[MZ12]{Mehta:2012}
\textsc{R.~A. Mehta} and \textsc{M.~Zambon}: \emph{{$L_\infty$}-algebra
  actions}, Differential Geom. Appl., \textbf{vol.~30(6)}, (2012) 576--587,
  ISSN 0926-2245, \doi{10.1016/j.difgeo.2012.07.006}.
  \urlprefix\url{http://dx.doi.org/10.1016/j.difgeo.2012.07.006}.

\bibitem[Pal12]{Palmkvist:2011vz}
\textsc{J.~Palmkvist}: \emph{Tensor hierarchies, {B}orcherds algebras and
  {$E_{11}$}}, JHEP, \textbf{vol. 1202}, (2012) 066,
  \doi{10.1007/JHEP02(2012)066}. \eprint{1110.4892}.

\bibitem[Pal14]{Palmkvist:2013vya}
--- \emph{{The tensor hierarchy algebra}}, J.Math.Phys., \textbf{vol.~55},
  (2014) 011701, \doi{10.1063/1.4858335}. \eprint{1305.0018}.

\bibitem[PS13]{Palmer:2013pka}
\textsc{S.~Palmer} and \textsc{C.~Saemann}: \emph{Six-dimensional $(1,0)$
  superconformal model and higher gauge theory}, J. Math. Phys.,
  \textbf{vol.~54}, (2013) 113509, \doi{10.1063/1.4832395}. \eprint{1308.2622}.

\bibitem[RSW09]{Riccioni:2009xr}
\textsc{F.~Riccioni}, \textsc{D.~Steele}, and \textsc{P.~West}: \emph{The
  {${\rm E}_{11}$} origin of all maximal supergravities: {T}he hierarchy of
  field-strengths}, JHEP, \textbf{vol. 0909}, (2009) 095,
  \doi{10.1088/1126-6708/2009/09/095}. \eprint{0906.1177}.

\bibitem[Sin11]{Singh:2011id}
\textsc{H.~Singh}: \emph{{Super-Yang-Mills and M5-branes}}, JHEP, \textbf{vol.
  1108}, (2011) 136, \doi{10.1007/JHEP08(2011)136}. \eprint{1107.3408}.

\bibitem[SS94]{Schaller:1994uj}
\textsc{P.~Schaller} and \textsc{T.~Strobl}: \emph{Poisson sigma models: {A}
  generalization of $2-d$ gravity {Y}ang-{M}ills systems}, in \emph{Finite
  dimensional integrable systems}, edited by \textsc{A.~Sisakian} and
  \textsc{G.~Pogosian}, p. p.181 (JINR, Dubna, 1994). \eprint{hep-th/9411163}.

\bibitem[SS95]{Schaller:1995xk}
--- \emph{A brief introduction to {P}oisson sigma models}, in \emph{Schladming
  1995, Low dimensional models in statistical physics and quantum field
  theory}, edited by \textsc{H.~Grosse} and \textsc{L.~Pittner}, p. p.321
  (Springer, 1995). \eprint{hep-th/9507020}.

\bibitem[SS13]{Salnikov:2013pwa}
\textsc{V.~Salnikov} and \textsc{T.~Strobl}: \emph{Dirac sigma models from
  gauging}, JHEP, \textbf{vol. 1311}, (2013) 110,
  \doi{10.1007/JHEP11(2013)110}. \eprint{1311.7116}.

\bibitem[SSS09]{Sati:2008eg}
\textsc{H.~Sati}, \textsc{U.~Schreiber}, and \textsc{J.~Stasheff}:
  \emph{{$L_{\infty}$} algebra connections and applications to string- and
  {C}hern-{S}imons $n$-transport}, in \emph{Quantum Field Theory}, edited by
  \textsc{B.~Fauser}, \textsc{J.~Tolksdorf}, and \textsc{E.~Zeidler}
  (Birkh\"auser, 2009). \eprint{0801.3480}.

\bibitem[SSW11]{Samtleben:2011fj}
\textsc{H.~Samtleben}, \textsc{E.~Sezgin}, and \textsc{R.~Wimmer}:
  \emph{{$(1,0)$} superconformal models in six dimensions}, JHEP, \textbf{vol.
  1112}, (2011) 062, \doi{10.1007/JHEP12(2011)062}. \eprint{1108.4060}.

\bibitem[SSW13]{Samtleben:2012fb}
--- \emph{{Six-dimensional superconformal couplings of non-abelian tensor and
  hypermultiplets}}, JHEP, \textbf{vol. 1303}, (2013) 068,
  \doi{10.1007/JHEP03(2013)068}. \eprint{1212.5199}.

\bibitem[SSWW11]{Samtleben:2012mi}
\textsc{H.~Samtleben}, \textsc{E.~Sezgin}, \textsc{R.~Wimmer}, and
  \textsc{L.~Wulff}: \emph{{New superconformal models in six dimensions: Gauge
  group and representation structure}}, PoS, \textbf{vol. CORFU2011}, (2011)
  071. \eprint{1204.0542}.

\bibitem[Str04]{Strobl:2004im}
\textsc{T.~Strobl}: \emph{Algebroid {Y}ang-{M}ills theories}, Phys.Rev.Lett.,
  \textbf{vol.~93}, (2004) 211601, \doi{10.1103/PhysRevLett.93.211601}.
  \eprint{hep-th/0406215}.

\bibitem[SW99]{daSilva-Weinstein}
\textsc{A.~Cannas~da Silva} and \textsc{A.~Weinstein}: \emph{Geometric models
  for non commutative algebras}, vol.~10 of \emph{Berkley Mathematics Lecture
  Notes}, 198 pp. (American Mathematical Society, Providence, RI, 1999).
  \urlprefix\url{http://math.berkley.edu/~alanw/}.

\bibitem[SW05]{Samtleben:2005bp}
\textsc{H.~Samtleben} and \textsc{M.~Weidner}: \emph{The maximal ${D} = 7$
  supergravities}, Nucl.Phys., \textbf{vol. B725}, (2005) 383--419,
  \doi{10.1016/j.nuclphysb.2005.07.028}. \eprint{hep-th/0506237}.

\bibitem[SW12]{Saemann:2012uq}
\textsc{C.~Saemann} and \textsc{M.~Wolf}: \emph{Non-abelian tensor multiplet
  equations from twistor space}. \eprint{1205.3108}.

\bibitem[Va{\u\i}97]{Vaintrob}
\textsc{A.~Y. Va{\u\i}ntrob}: \emph{Lie algebroids and homological vector
  fields}, Uspekhi Mat. Nauk, \textbf{vol.~52(2(314))}, (1997) 161--162, ISSN
  0042-1316, \doi{10.1070/RM1997v052n02ABEH001802}.
  \urlprefix\url{http://dx.doi.org/10.1070/RM1997v052n02ABEH001802}.

\bibitem[Vor05]{Voronov2005}
\textsc{T.~Voronov}: \emph{Higher derived brackets and homotopy algebras}, J.
  Pure Appl. Algebra, \textbf{vol. 202(1-3)}, (2005) 133--153, ISSN 0022-4049,
  \doi{10.1016/j.jpaa.2005.01.010}.
  \urlprefix\url{http://dx.doi.org/10.1016/j.jpaa.2005.01.010}.

\bibitem[WNS08]{deWit:2008ta}
\textsc{B.~de~Wit}, \textsc{H.~Nicolai}, and \textsc{H.~Samtleben}:
  \emph{Gauged supergravities, tensor hierarchies, and {M}-theory}, JHEP,
  \textbf{vol. 0802}, (2008) 044, \doi{10.1088/1126-6708/2008/02/044}.
  \eprint{arXiv:0801.1294}.

\bibitem[WS05]{deWit:2005hv}
\textsc{B.~de~Wit} and \textsc{H.~Samtleben}: \emph{Gauged maximal
  supergravities and hierarchies of nonabelian vector-tensor systems},
  Fortschr. Phys., \textbf{vol.~53}, (2005) 442--449. \eprint{hep-th/0501243}.

\bibitem[WS08]{deWit:2008gc}
--- \emph{The end of the $p$-form hierarchy}, JHEP, \textbf{vol.~08}, (2008)
  015, \doi{10.1088/1126-6708/2008/08/015}. \eprint{0805.4767}.

\bibitem[WST05]{deWit:2004nw}
\textsc{B.~de~Wit}, \textsc{H.~Samtleben}, and \textsc{M.~Trigiante}: \emph{The
  maximal ${D} = 5$ supergravities}, Nucl. Phys., \textbf{vol. B716}, (2005)
  215--247. \eprint{hep-th/0412173}.

\end{thebibliography}

\newcommand{\etalchar}[1]{$^{#1}$}

\end{document}